\documentclass[acmsmall,screen]{acmart}

\usepackage{mathtools}
\usepackage{stmaryrd}
\mathtoolsset{showonlyrefs=true}
\usepackage{tikz,pgfplots}
\pgfplotsset{compat=1.18}
\usetikzlibrary{cd,automata,positioning,calc}
\tikzset{
	initial text=$ $,
	every state/.style={inner sep=2pt, minimum size=18pt}
}
\tikzset{%
    symbol/.style={%
        draw=none,
        every to/.append style={%
            edge node={node [sloped, allow upside down, auto=false]{$#1$}}}
    }
}
\usepackage{mathpartir}
\usepackage{multirow}
\usepackage{wrapfig}

\usepackage{ifthen}
\newboolean{longversion}
\setboolean{longversion}{true}
\newboolean{submission}
\setboolean{submission}{true}
\newcommand{\referappendix}[2]{\ifthenelse{\boolean{longversion}}{Appendix~\ref{#1}}{Appendix~{#2}}}

\usepackage{macros.basic}
\newcommand{\todo}[1]{\textcolor{red}{[ToDo: #1]}}
\newcommand{\answertype}{\mathbf{Prop}}
\newcommand{\answerobj}{\Omega}
\newcommand{\contmonad}[1]{C_{#1}}
\newcommand{\CPS}[1]{#1^{\sharp}}
\renewcommand{\category}[1]{\mathbb{#1}}
\newcommand{\letrec}[4]{\mathbf{let}\ \mathbf{rec}\ #1\ #2\ =\ #3\ \mathbf{in}\ #4}
\newcommand{\letrecchurch}[6]{\mathbf{let}\ \mathbf{rec}\ #1\ (#2 : #3) : #4\ =\ #5\ \mathbf{in}\ #6}
\newcommand{\nontermconst}{\bot}
\newcommand{\join}{\mathbin{\square}}

\newcommand{\operation}[1]{\mathsf{#1}}
\newcommand{\ifexpr}[3]{\mathbf{if}\ #1\ \mathbf{then}\ #2\ \mathbf{else}\ #3}
\newcommand{\emalgsymbol}{\zeta}
\newcommand{\geneff}[1]{\mathbf{gen}_{#1}}

\AtEndPreamble{%
\theoremstyle{acmdefinition}
\newtheorem{remark}[theorem]{Remark}}

\AtBeginDocument{%
  }

\setcopyright{rightsretained}
\copyrightyear{2025}
\acmYear{2025}
\acmDOI{XXXXXXX.XXXXXXX}

\begin{document}

\title{A Category-Theoretic Framework for Syntactic Computation of Generic Weakest Preconditions}

\author{Satoshi Kura}
\email{satoshi.kura@aoni.waseda.jp}
\orcid{0000-0002-3954-8255}
\affiliation{%
  \institution{Waseda University}
  \state{Tokyo}
  \country{Japan}
}

\begin{abstract}
	Weakest preconditions are a useful notion for program verification as they reduce a problem of program verification to a problem of constraint solving.
	Category-theoretic generalisations of weakest preconditions have been studied to capture various computational effects and various properties in a unified framework.
	In this paper, we propose a novel and general relationship between weakest precondition transformers and CPS transformations for higher-order functional languages with general computational effects and recursion.
	Technically, this gives a syntactic counterpart of the categorically-defined generic weakest precondition transformer in [Aguirre \& Katsumata, 2020].
	The usefulness of our results is threefold.
	(1) Since CPS transformations purify effectful programs, various verification problems for \emph{effectful} programs can be reduced to verification problems for \emph{pure} programs.
	This syntactic reduction makes it easier to solve the verification problems and potentially facilitates combinations with other sophisticated verification methods tailored for pure programs.
	(2) We capture two existing verification methods, namely, verification of event sequences~[Kobayashi et al., 2018] and expected cost~[Avanzini et al., 2021] as instances of our framework.
	(3) Our results streamline the process of extending weakest precondition transformers for \emph{imperative} programs to those for \emph{higher-order} programs.
	We show two such extensions: analysis of higher moments of cost and the conditional weakest pre-expectation for higher-order probabilistic programs.
	These extensions demonstrate that our theoretical framework can produce novel verification methods.
\end{abstract}


\maketitle

\section{Introduction}
This paper is about a general framework for verification based on weakest preconditions \cite{dijkstra1975}.
Our framework subsumes two earlier approaches with diverse applications: for safety properties \cite{kobayashi2018} and for the expected cost of randomised algorithms \cite{avanzini2021}, establishing a formal connection between these syntactic verification methods \cite{kobayashi2018,avanzini2021} and categorical semantics of generic weakest preconditions \cite{aguirre2020}.
We also demonstrate how our framework goes beyond these two previous examples by providing new verification methods for (1) higher moments of cost of randomised programs and (2) conditional weakest pre-expectations for probabilistic programs with conditioning.
Technically, our framework is based on the connection between continuation passing and weakest preconditions.
We achieve a syntactic verification method by extracting it from a new categorical foundation.

\paragraph{Background: weakest preconditions and program verification.}
Weakest preconditions~\cite{dijkstra1975} are useful notions for program verification.
Given a program and a Hoare-style specification (i.e.\ a precondition and a postcondition for the program), the satisfaction of the specification can be reduced to constraint solving assuming that we can compute weakest preconditions as formulas.
The combination of weakest preconditions and constraint solvers yields Hoare-logic style verification tools such as Why3~\cite{filliatre2013} and Boogie~\cite{barnett2006}.

\paragraph{Generic weakest preconditions for various computational effects.}
There are many variations of weakest preconditions for various kinds of programs.
For example, for programs that may diverge, there are two variations for total and partial correctness depending on whether we consider that a postcondition is satisfied when a program diverges.
For nondeterministic programs, we have weakest preconditions for may and must correctness depending on whether a postcondition should be satisfied for some/any possible output values.
For probabilistic programs, the weakest pre-expectation~\cite{mciver2001} and the expected runtime transformer~\cite{kaminski2018} are studied as quantitative extensions of weakest preconditions.
To give a uniform account of these variations, there is a line of research~\cite{goncharov2013,hino2016,hasuo2015,aguirre2020} on category-theoretic frameworks of weakest preconditions.
These abstract frameworks successfully capture the semantic structures of the examples above and potentially lead to new applications of weakest preconditions while minimising problem-specific considerations.

\paragraph{Problem.}
However, existing category-theoretic frameworks mainly focus on the general semantics of weakest preconditions and often lack general syntactic aspects.
This makes it difficult to automate the computation of generic weakest preconditions.
For example, \cite{aguirre2020} considers a \emph{semantic} weakest precondition transformer $\mathrm{wp}^{\emalgsymbol}[f] : \category{C}(Y, \answerobj) \to \category{C}(X, \answerobj)$ defined by $\mathrm{wp}^{\emalgsymbol}[f](q) = \emalgsymbol \comp T q \comp f$ where $T$ is a monad on a category $\category{C}$, $\emalgsymbol : T \answerobj \to \answerobj$ is an EM $T$-algebra, and $f : X \to T Y$ represents the interpretation of a program.
Although this gives a nice general \emph{semantic} definition of various weakest precondition transformers, we could hardly expect that we can automatically compute the semantic weakest precondition as is because the interpretation $f : X \to T Y$ of a program is already hard to compute when a program is written in a realistic programming language (e.g.\ one that allows recursive functions).

Suppose we have a program $M$ and a formula $Q$ that represents a postcondition.
When automating the computation of weakest preconditions, we want to compute a formula $P$ that represents the weakest precondition for $M$ and $Q$.
More formally, the problem here is to compute a formula $P$ such that $\interpret{P} = \mathrm{wp}^{\emalgsymbol}[\interpret{M}](\interpret{Q})$, which reads the interpretation of the formula $P$ is equal to the semantic weakest precondition defined by the interpretations of $M$ and $Q$; and this should be done without computing the interpretation $\interpret{M}$.
We call such $P$ the \emph{syntactic} weakest precondition.
In this paper, we consider the case where $M$ is a functional program and how such computation of syntactic weakest preconditions yields syntactic verification methods of functional programs such as \cite{kobayashi2018,avanzini2021}.

\begin{figure}[tbp]
	\centering
	\begin{tikzpicture}[font=\small]
		\node[draw] (syn) at (-1.3, 0) {Syntax};
		\node[draw] (sem) at (-1.3, -1.1) {Semantics};
		\node[draw] (prog) at (0, 0.5) {Program};
		\node[draw] (logic) at (6, 0.5) {Formula};
		\node at (0, 0) (prog_syn) {$M$};
		\node at (6, 0) (log_syn) {$\CPS{M}$};
		\node at (0, -1.1) (prog_sem) {$\mathcal{A}\interpret{M}$};
		\node at (4, -1.1) (wpt) {$\mathrm{wp}^{\emalgsymbol}[\mathcal{A}\interpret{M}]$};
		\node at (6, -1.1) (log_sem) {$\mathcal{A}^{\emalgsymbol} \interpret{\CPS{M}}$};
		\draw[|->] (prog_syn) -- node[midway, right] {interpret} (prog_sem);
		\draw[|->] (log_syn) -- node[midway, left] {interpret} (log_sem);
		\draw[|->] (prog_syn) -- node[midway, above, align=center, font=\normalsize] {CPS transformation} (log_syn);
		\draw[|->] (prog_sem) -- node[midway, below, align=center] (wpt_label) {Semantic weakest \\ precondition transformer} (wpt);
		\path (wpt) -- node[midway] {$=$} (log_sem);
	\end{tikzpicture}
	\caption{CPS transformations as a syntactic counterpart of weakest precondition transformers. Here, $\mathcal{A}\interpret{-}$ is the interpretation of programs, and $\mathcal{A}^{\emalgsymbol}\interpret{-}$ is the interpretation of formulas.}
	\label{fig:cps_wpt}
\end{figure}

\paragraph{Our result.}
In this paper, we provide a general category-theoretic framework for computing \emph{syntactic} weakest preconditions.
Using our framework, we establish a formal connection between two existing syntactic verification methods \cite{kobayashi2018,avanzini2021} and categorical semantics of generic weakest preconditions \cite{aguirre2020}, and also obtain a new method for verifying randomised programs.
In doing so, a CPS transformation~\cite{plotkin1975} plays a crucial role.

Specifically, we consider the following setting.
As a programming language, we consider a functional language with computational effects and recursion.
Pre/post-conditions are written in a generalised version of higher-order modal fixed-point logic (HFL) \cite{viswanathan2004}.
In this setting, we can compute syntactic weakest preconditions using a CPS transformation (the top edge in Fig.~\ref{fig:cps_wpt}).
The soundness of the syntactic weakest preconditions (the equation at the bottom right of Fig.~\ref{fig:cps_wpt}) is proved in Theorem~\ref{thm:cps-is-wpt-with-recursion}.

Intuitively, the connection between weakest preconditions for functional programs and a CPS transformation can be understood as follows.
Given a well-typed program $x : \tau \vdash M : \rho$, the CPS transformation $\CPS{({-})}$ gives a well-typed formula $x : \CPS{\tau} \vdash \CPS{M} : (\CPS{\rho} \to \answertype) \to \answertype$.
Here, we use a type of truth values $\answertype$ as an answer type.
By reordering the arguments of $\CPS{M}$, we can regard $\CPS{M}$ as a function of type $(\CPS{\rho} \to \answertype) \to (\CPS{\tau} \to \answertype)$
Now, we can see that $\CPS{M}$ has the same type as weakest precondition transformers.
The weakest precondition transformer for $M$ is a function of type $(\rho \to \answertype) \to (\tau \to \answertype)$ where $\rho \to \answertype$ is the type of postconditions and $\tau \to \answertype$ is the type of preconditions.
If we have $\CPS{\tau} = \tau$ and $\CPS{\rho} = \rho$, which is the case when $\tau$ and $\rho$ are ground types, then these two types coincide.
Moreover, the CPS-transformed term $\CPS{\tau}$ and the weakest precondition transformer for $M$ perform the same computation: they both pass the result of $M$ to a continuation or a postcondition.
This connection holds for various kinds of weakest preconditions.
To the best of our knowledge, our Theorem~\ref{thm:cps-is-wpt-with-recursion} is the most general one that proves this.

Computing weakest preconditions by a CPS transformation makes it easier to verify programs because the CPS transformation turns \emph{effectful} programs into HFL formulas, which are basically terms of pure lambda calculus.
The verification of pure lambda calculus is usually easier than the verification of effectful programs, and thus, many sophisticated methods are available for the former (e.g., refinement type systems~\cite{flanagan2006}).
Combined with such methods, our syntactic translation from effectful programs to pure terms potentially leads to a powerful approach to program verification.
In this paper, we will focus on its theoretical foundation.

Our framework is general in the sense that it contains several parameters that can be instantiated for various problems of program verification (Fig.~\ref{fig:overview}).
The syntax of programs is parameterised by base types, effect-free constants, and algebraic operations so that our framework can cover various programs, examples of which will be provided in Section~\ref{sec:examples}.
The semantics is also parameterised so that base types, effect-free constants, and algebraic operations can be interpreted appropriately.
The semantics of generic weakest preconditions is defined based on the results of~\cite{aguirre2020}, which uses an Eilenberg--Moore algebra $\emalgsymbol$ as a parameter.
Once these parameters are fixed, the syntax and semantics of HFL formulas are determined accordingly, and we automatically get a sound syntactic translation (CPS transformation) to obtain the weakest precondition.
We don't have to repeat parameter-specific proofs of soundness for each combination of parameters because we have a general soundness theorem (Theorem~\ref{thm:cps-is-wpt-with-recursion}) proved via category-theoretic abstraction.

\begin{figure}[tbp]
	\textbf{Parameters}:
	\begin{itemize}
		\item Syntax of programs $\lambda_c(\Sigma)$ is parameterised by $\Sigma$.
		\item Semantics of programs $\mathcal{A}\interpret{-}$ is parameterised by $\mathcal{A}$.
		\item Semantics of generic wp $\mathrm{wp}^{\emalgsymbol} {[{-}]}$ is parameterised by $\emalgsymbol$.
	\end{itemize}
	\textbf{What we get from our framework}:
	\begin{itemize}
		\item Higher-order modal fixed-point logic (Syntax $\lambda_{\mathrm{HFL}}(\Sigma)$ / Semantics $\mathcal{A}^{\emalgsymbol} \interpret{-}$)
		\item Syntactic computation of weakest preconditions ($\CPS{M}$ such that $\mathrm{wp}^{\emalgsymbol}[\mathcal{A}\interpret{M}]\ =\ \mathcal{A}^{\emalgsymbol} \interpret{\CPS{M}}$)
	\end{itemize}
	\caption{Overview of our framework.}
	\label{fig:overview}
\end{figure}

By instantiating parameters appropriately, we obtain several instances as shown in Table~\ref{tab:instances}.
As the simplest situation, we have instances of total and partial correctness where we do not have any computational effect except for non-termination caused by recursion.
We can also capture two existing works of CPS-based program verification (for safety properties~\cite{kobayashi2018} and for expected cost analyses~\cite{avanzini2021}) as instances, which justifies the design of our framework.
Our framework also yields a new application: a CPS-based verification method for analysing higher moments of cost of functional probabilistic programs.

These instances exemplify the usefulness of our framework.
Firstly, our framework can be used to naturally extend weakest precondition transformers for \emph{imperative} programs to \emph{higher-order} functional programs.
EM algebras used to define generic weakest preconditions~\cite{aguirre2020} are often obtained by analysing weakest precondition transformers for imperative programs.
Once such EM algebras are obtained, our framework provides a way to use the same weakest precondition transformers for higher-order programs.
For example, the instance of cost moment analysis in this paper gives a higher-order extension of~\cite{kura2019}.
Secondly, our framework provides a uniform understanding of existing verification methods~\cite{kobayashi2018,avanzini2021} by establishing a formal connection between these verification methods and the work~\cite{aguirre2020} on semantic weakest preconditions.
Note that their soundness proofs were tightly coupled with their specific problem settings, while our soundness proof uniformly holds for many problems.
This theoretical foundation paves a way to adapt these methods to different situations.
In fact, the instance of cost moment analysis extends the expected cost analysis in~\cite{avanzini2021} as well.

\begin{table*}[tbp]
	\caption{Instances of Theorem~\ref{thm:cps-is-wpt-with-recursion}.
	Here, $U$ is the set of states of a finite automaton, and $2^U$ is the powerset of $U$.
	The computational effect of non-termination is implicitly assumed since we consider languages with recursion.
	More instances can be found in \protect\referappendix{sec:more-instances}{H}.}
	\label{tab:instances}
	\centering
	\begin{tabular}{|c|c|c|c|c|}
		\hline
		Property & Computational effects & Category & Truth values & CPS \\
		\hline
		Total correctness & \multirow{2}{*}{(no effects)} & \multirow{2}{*}{$\omegaCPO$} & $\{ \mathbf{true}, \mathbf{false} \}$ & Example~\ref{ex:total-partial-correctness-cps} \\
		Partial correctness & & & $\{ \mathbf{true}, \mathbf{false} \}$ & Example~\ref{ex:total-partial-correctness-cps} \\
		\hline
		Safety property & output \& nondeterminism & $\omegaCPO$ & $2^U$ & Example~\ref{ex:trace-propery-extended} \\
		\hline
		Expected cost& \multirow{2}{*}{probability \& cost} & \multirow{2}{*}{$\omegaQBS$} & $[0, \infty]$ & Example~\ref{ex:expected-cost-extended} \\
		Cost moment & & & $[0, \infty]^n$ & Example~\ref{ex:cost-moment-cps} \\
		\hline
		Conditional wp & probability \& conditioning & $\omegaQBS$ & $[0, \infty] \times [0, 1]$ & Example~\ref{ex:conditional-wp-cps} \\
		\hline
	\end{tabular}
\end{table*}

Our contributions are summarised as follows.
\begin{itemize}
	\item We provide a general framework to syntactically compute weakest preconditions for various computational effects and various problems of program verification.
	In Theorem~\ref{thm:cps-is-wpt-with-recursion}, we show that a CPS transformation gives syntactic weakest preconditions for functional programs with general computational effects and recursion.
	Since the CPS transformation ``purifies'' effectful programs into pure terms, our framework makes program verification easier and also makes it easier to apply other sophisticated verification methods for pure programs to the verification of effectful programs.
	\item We show that by choosing the parameters of our framework, existing methods for program verification~\cite{kobayashi2018,avanzini2021} can be reproduced as instances of Theorem~\ref{thm:cps-is-wpt-with-recursion}.
	This establishes a formal connection between CPS-based methods for program verification~\cite{kobayashi2018,avanzini2021} and a categorical generalisation of weakest preconditions \cite{aguirre2020}, which gives a theoretical guide to adapt these methods to new situations.
	\item We obtain new verification methods for cost moment analyses and the conditional weakest pre-expectation for higher-order programs as instances of Theorem~\ref{thm:cps-is-wpt-with-recursion}.
	These examples demonstrate the potential of our framework to provide new syntactic verification methods.
\end{itemize}

\section{Examples of Verification Problems}\label{sec:examples}
In this section, we show several problems of program verification as examples and explain how we can use CPS transformations for these problems.
All of the examples are instances of our Theorem~\ref{thm:cps-is-wpt-with-recursion} (see Table~\ref{tab:instances}), which we will later explain in detail.
Section~\ref{subsec:example-total-partial} deals with the simplest problem of verifying total/partial correctness; Section~\ref{subsec:trace-property} is from~\cite{kobayashi2018}; Section~\ref{subsec:expected-cost} is from~\cite{avanzini2021}; and Section~\ref{subsec:cost-moment} is a new instance, which extends~\cite{avanzini2021,kura2019}.
Note that the point of this paper is to provide a uniform framework that subsumes these problems whereas~\cite{kobayashi2018,avanzini2021} are limited to specific problems.
This section aims to illustrate the range of verification problems supported by our framework.
Therefore, we intentionally keep example programs simple.
More advanced programs can be found in, e.g., \cite{kobayashi2018,avanzini2021}, but there is nothing happening beyond the description below: it is just a matter of applying CPS transformations after all.

\subsection{Total/Partial Correctness}\label{subsec:example-total-partial}
To get an idea of the relationship between weakest preconditions and a CPS transformation, we begin with a simple situation.
Suppose that we want to obtain weakest preconditions for the function $\mathrm{fact} : \mathbf{int} \to \mathbf{int}$ that takes an integer $n$, returns the factorial $n!$ if $n \ge 0$, and diverges otherwise.
The program can be written as follows using OCaml-like syntax.
\[ \letrec{\mathrm{fact}}{n}{\ifexpr{n = 0}{1}{n * \mathrm{fact}\ (n - 1)}}{\mathrm{fact}\ n} \]

Given a postcondition $Q$ (a predicate on the output), the weakest precondition is, if we ignore the case where $\mathrm{fact}$ diverges, the predicate $P$ on the input such that an input $n$ satisfies $P(n)$ if and only if the output $\mathrm{fact}\ n$ satisfies $Q (\mathrm{fact}\ n)$.
This is similar to what the CPS-transformed program does.
\begin{align}
	&\mathbf{let}\ \mathbf{rec}\ \mathrm{fact}'\ n\ k \ =\ \ifexpr{n = 0}{k\ 1}{\mathrm{fact}'\ (n - 1)\ (\lambda r. k\ (n * r))}\ \mathbf{in}\ \mathrm{fact}'\ n\ Q \label{eq:total-partial-cpsed}
\end{align}
The type of the CPS-transformed program $\mathrm{fact}'$ is $\mathbf{int} \to (\mathbf{int} \to \answertype) \to \answertype$ where $\answertype$ is an answer type.
It takes a pair of an input $n : \mathbf{int}$ and a continuation $k : \mathbf{int} \to \answertype$ and returns the value that is equal to $k (\mathrm{fact}\ n) : \answertype$.
If we pass a postcondition $Q : \mathbf{int} \to \answertype$ as a continuation, we get a program $\mathrm{fact}'\ n\ Q$ that returns $Q (\mathrm{fact}\ n)$ for a given input $n$.
This is exactly the same as what the weakest precondition does if we interpret $\answertype$ as a type of truth values.
More formally, our Theorem~\ref{thm:cps-is-wpt-with-recursion} ensures that the CPS transformation actually transforms a program into a term of pure lambda calculus (we think of it as a higher-order logic formula) that represents the weakest precondition transformer.

Let's consider how the divergence of the program affects the weakest preconditions, which was left implicit above.
There are two possible ways to deal with divergence in weakest preconditions, which correspond to how we interpret $\mathbf{let}\ \mathbf{rec}$ in~\eqref{eq:total-partial-cpsed}.
If we interpret $\mathbf{let}\ \mathbf{rec}$ as the least fixed point, then we obtain the weakest precondition for \emph{total correctness}, which requires the program to terminate.
On the other hand, if we interpret $\mathbf{let}\ \mathbf{rec}$ as the greatest fixed point, then we obtain the weakest precondition for \emph{partial correctness} (or the weakest \emph{liberal} precondition), which does not require termination.
We will later explain that our framework captures this distinction by how we define order relations on $\answertype$.

\subsection{Safety Property}\label{subsec:trace-property}
We consider the problem of checking safety properties: given a program $M$ that outputs sequences of events and an automaton $\mathfrak{A}$ that accepts a regular language, we consider the problem of checking whether any possible output from the program $M$ is accepted by the automaton $\mathfrak{A}$.
Here, we assume that the automaton $\mathfrak{A}$ is deterministic (i.e.\ has at most one transition for each state-event pair) and that all states are final states, which means that the corresponding regular language is prefix-closed.
A CPS-based verification method for checking the safety property is studied in~\cite[Section~6]{kobayashi2018} and~\cite{kobayashi2009a}.

\begingroup
\setlength{\columnsep}{0pt}
\begin{wrapfigure}{r}{11em}
	\vspace{-1em}
	\begin{tikzpicture}[node distance=2cm,baseline=(q0)]
		\node[state,initial] (q0) {$q_0$};
		\node[state, right of=q0] (q1) {$q_1$};
		\path[->] (q0) edge [loop above] node[right] {read} ()
			(q0) edge node [above] {close} (q1);
	\end{tikzpicture}
	\vspace{-1em}
\end{wrapfigure}
For example, consider verifying that the following program \eqref{eq:trace-example} does not read a file after closing it.
The specification is given by the automaton on the right, which corresponds to the regular expression $(\operation{read})^{*}\ \operation{close}$.
\begin{equation}
	\letrec{f}{x}{\ifexpr{{*}}{\mathrm{close}(x)}{(\mathrm{read}(x); f\ x)}}{f\ \mathrm{some\_file}}
	\label{eq:trace-example}
\end{equation}
Here, $\ifexpr{{*}}{M}{N}$ means nondeterministic branching.
Since we are interested only in the sequence of file operations, we slightly simplify the program.
\begin{equation}
	\letrec{f}{()}{\ifexpr{{*}}{\mathrm{ev\_close}()}{(\mathrm{ev\_read}(); f\ ())}}{f\ ()} \label{eq:example-program-trace-property}
\end{equation}
\endgroup

In \eqref{eq:example-program-trace-property}, we focus on the content of the ``debug log'', which contains sequences of events.
In this case, an event is either $\operation{close}$ or $\operation{read}$.
The functions $\operation{ev\_close}()$ and $\operation{ev\_read}()$ append these events to the ``debug log''.
Now, the problem is whether any possible event sequence in the ``debug log'' matches the specification $(\operation{read})^{*}\ \operation{close}$.

As proposed in~\cite{kobayashi2018}, a safety property can be reduced to the validity checking of a formula of higher-order modal fixed-point logic (HFL)~\cite{viswanathan2004}.
Combined with constraint solvers for HFL, this reduction leads to an automated verification method for safety properties.
We focus only on the reduction to HFL because constraint solving is out of the scope of our paper.
By applying their reduction to the example above, we get the following HFL formula~\eqref{eq:example-formula-trace-property}, and the safety property holds if and only if \eqref{eq:example-formula-trace-property} is true at the initial state $q_0$ of the automaton.
\begin{equation}
	\nu F. \langle \mathrm{close} \rangle \mathbf{true} \land \langle \mathrm{read} \rangle F \label{eq:example-formula-trace-property}
\end{equation}
Here, $\langle a \rangle$ is a modal operator such that $\langle a \rangle \phi$ holds if there exists a transition labelled with $a$ such that $\phi$ holds after the transition.
Note that \eqref{eq:example-formula-trace-property} is interpreted as a set of states of the automaton, and the greatest fixed point operator $\nu F$ is the greatest with respect to the inclusion order.

This reduction from \eqref{eq:example-program-trace-property} to \eqref{eq:example-formula-trace-property} is an instance of our Theorem~\ref{thm:cps-is-wpt-with-recursion}.
We apply a CPS transformation to \eqref{eq:example-program-trace-property}, in which nondeterminism branching and event operations ($\mathrm{ev\_read}()$ / $\mathrm{ev\_close}()$) are mapped to conjunction $\land$ and modal operators $\langle \mathrm{read} \rangle$ / $\langle \mathrm{close} \rangle$, respectively.
Then, we obtain a function $f' : \mathbf{unit} \to (\mathbf{unit} \to \answertype) \to \answertype$ where the answer type $\answertype$ here is a type of assignments of true or false to each state of the automaton.
\begin{equation}
	\letrec{f'}{x\ k}{\langle \mathrm{close} \rangle (k\ x) \land \langle \mathrm{read} \rangle (f'\ x\ k)}{f'\ ()\ (\lambda r. \mathbf{true})} \label{eq:example-cps-trace-property}
\end{equation}
Here, we pass the always-true proposition $\lambda r. \mathbf{true} : \mathbf{unit} \to \answertype$ as a continuation and interpret $\mathbf{let}\ \mathbf{rec}$ as the greatest fixed point (the reason for these choices will be explained later in Example~\ref{ex:trace-property-wp},\ref{ex:trace-property-cps}).
Then, the safety property holds if and only if \eqref{eq:example-cps-trace-property} is true at the initial state $q_0$ by Theorem~\ref{thm:cps-is-wpt-with-recursion}.
Note that \eqref{eq:example-formula-trace-property} is obtained from \eqref{eq:example-cps-trace-property} by defining $F \coloneqq f'\ ()\ (\lambda r. \mathbf{true})$.

\subsection{Expected Cost Analysis}\label{subsec:expected-cost}
Given a randomised program, we consider the problem of estimating the expected cost (or runtime) of the program.
\cite{avanzini2021} proposed a CPS-based method to solve this problem.
In fact, their result~\cite{avanzini2021} can be understood as an instance of our framework, and our CPS transformation and that of~\cite{avanzini2021} coincide in this case.

For example, consider the expected cost of a random walk.
\[ \letrec{f}{n}{\ifexpr{n \le 0}{()}{(f\ (n-1) +_p f\ (n+1))^{\checkmark}}}{f\ 42} \]
Here, $+_p$ is a probabilistic branching operator whose left operand is taken with probability $p$, and $({-})^{\checkmark}$ means incrementing the cost.
This program reads ``if the current state is $n > 0$, then the cost increases by 1, and the next state is $n-1$ with probability $p$ and $n+1$ with probability $1-p$''.
We want to know the expected number of transitions (marked by $({-})^{\checkmark}$) until we reach $n \le 0$.

We can apply our Theorem~\ref{thm:cps-is-wpt-with-recursion} to this problem.
By applying a CPS transformation, we get a pure function $f' : \mathbf{int} \to (\mathbf{unit} \to \answertype) \to \answertype$ where $\answertype = \mathbf{real}^{+}$ is a type of non-negative extended real numbers $r \in [0, \infty]$.
\begin{align}
	&\mathbf{let}\ \mathbf{rec}\ f'\ n\ k\ =\ \ifexpr{n \le 0}{k\ ()}{1 + p \cdot (f'\ (n-1)\ k) + (1-p) \cdot (f'\ (n+1)\ k)}\ \mathbf{in} \\
	&f'\ 42\ (\lambda r. 0)
\end{align}
Here, the CPS transformation maps $+_p$ and $({-})^{\checkmark}$ to a weighted sum $p \cdot ({-}) + (1 - p) \cdot ({-})$ and an addition $1 + ({-})$, respectively.
It should be noted that the meaning of $\mathbf{let}\ \mathbf{rec}$ changes after the CPS transformation.
In this case, $\mathbf{let}\ \mathbf{rec}$ after CPS transformation is interpreted as the least fixed point with respect to the standard order on $[0, \infty]$, the reason for which will be explained later in this paper.
Intuitively, the continuation $k : \mathbf{unit} \to \answertype$ of $f'$ is a function that represents the expected cost of a continuation of $f$, and $f'\ n\ k : \answertype$ represents the expected cost of the program $f\ n$ followed by the continuation of $f$.
Therefore, if we use $\lambda r. 0 : \mathbf{unit} \to \mathbf{real}^{+}$ as a continuation of $f'$, we get the expected cost of $f\ n$ itself.
More advanced examples (e.g.\ cost of higher-order programs) can be found in~\cite{avanzini2021} and \referappendix{sec:advanced-example}{C}.

\subsection{Cost Moment Analysis}\label{subsec:cost-moment}
We consider the same programs as the expected cost analysis (Section~\ref{subsec:expected-cost}) but a different problem here, that is, the problem of estimating higher moments $\mathbb{E}[C^m]$ of cost $C$ instead of the expected cost $\mathbb{E}[C]$.
This extension allows us to obtain more information about the probability distribution of the cost.
For example, we can get tighter upper bounds of tail probabilities using higher moments~\cite{kura2019}.

Our Theorem~\ref{thm:cps-is-wpt-with-recursion} allows us to extend the CPS-based expected cost analysis~\cite{avanzini2021} to analyses of higher moments for functional randomised programs.
This also extends cost moment analyses for imperative programs~\cite{kura2019} to higher-order programs.
Here, we apply a CPS transformation using $\answertype = (\mathbf{real}^{+})^m$ as an answer type.
The intuition is that $(c_1, \dots, c_m) : (\mathbf{real}^{+})^m$ represents a tuple $(\mathbb{E}[C], \dots, \mathbb{E}[C^m])$ of moments, which follows the idea proposed in~\cite{kura2019}.
For simplicity, suppose that we want to know the second moment ($m = 2$).
Using $\answertype = \mathbf{real}^{+} \times \mathbf{real}^{+}$, our CPS transformation gives the following.
\begin{align}
	&\mathbf{let}\ \mathbf{rec}\ f'\ n\ k\ =\ \ifexpr{n \le 0}{k\ ()}{1 \oplus (p \cdot (f'\ (n-1)\ k) + (1-p) \cdot (f'\ (n+1)\ k))}\ \mathbf{in} \\
	&f'\ 42\ (\lambda r. (0, 0))
\end{align}
Here, $f'\ 42\ (\lambda r. (0, 0)) : \mathbf{real}^{+} \times \mathbf{real}^{+}$ represents the pair of the first and the second moments of cost of $f\ 42$.
Note that a function $1 \oplus ({-}) : \mathbf{real}^{+} \times \mathbf{real}^{+} \to \mathbf{real}^{+} \times \mathbf{real}^{+}$ (called \emph{elapse function} in~\cite{kura2019}) is defined by $1 \oplus (x_1, x_2) \coloneqq (1 + x_1, 1 + 2 x_1 + x_2)$, which implements the binomial expansions $\mathbb{E}[(C + 1)^i] = \mathbb{E}[C^i] + i \mathbb{E}[C^{i-1}] + \dots + 1$ for $i = 1, 2$ required for incrementing cost $C \mapsto C + 1$.

\subsection{Conditional Weakest Pre-expectation}\label{subsec:conditional-wp}
The weakest pre-expectation \cite{mciver2001} is an extension of the weakest precondition for probabilistic programs.
To reason about probabilistic programs, real-valued predicates are used instead of boolean-valued predicates, and these are called \emph{pre-/post-expectations}.
The weakest pre-expectation gives the expected value of a post-expectation with respect to the distribution of the results of a computation.

Conditioning is one of the main features of probabilistic programs.
The weakest pre-expectation for imperative probabilistic programs \emph{with conditioning} is studied in \cite{olmedo2018}.
Here, we consider extending the conditional weakest pre-expectation to \emph{higher-order} probabilistic programs with conditioning.
As a running example, consider the following program.
\begin{align}
	\mathbf{let}\ \mathbf{rec}\ f\ x\ =&\quad\mathbf{let}\ b_1, b_2 = \mathrm{bern}(1/2), \mathrm{bern}(1/2)\ \mathbf{in}\ \mathrm{observe}(\lnot b_2 \lor b_1); \\
	&\quad\ifexpr{b_2}{f\ (); f\ ()}{\ifexpr{b_1}{f\ ()}{()}}
\end{align}

The program $f\ ()$ flips two fair coins ($b_1$ and $b_2$), observes that $(b_1, b_2) \in \{ (\mathbf{false}, \mathbf{false}),\allowbreak (\mathbf{true}, \mathbf{false}),\allowbreak (\mathbf{true}, \mathbf{true}) \}$ holds, and decides the number of recursive calls to itself based on the result of the coin flips.
Suppose that we are interested in the conditional termination probability, which is equal to the conditional weakest pre-expectation for the constant post-expectation $1$.

By Theorem~\ref{thm:cps-is-wpt-with-recursion}, we can obtain the conditional weakest pre-expectation by applying a CPS transformation.
As an answer type, we use $\answertype = \mathbf{real}^{+} \times \mathbf{real}_{[0, 1]}$ where $\mathbf{real}_{[0, 1]}$ is the type of real numbers in the unit interval $[0, 1]$.
The intuition is that the first component of type $\mathbf{real}^{+}$ represents the (unnormalised) weakest pre-expectation and the second component of type $\mathbf{real}_{[0, 1]}$ represents the probability that all observations are satisfied.
The CPS-transformed term is given as follows.
\begin{alignat}{3}
	\mathbf{let}\ \mathbf{rec}\ f'\ x\ k =&\quad 1/2 \cdot (1/2 \cdot f'\ ()\ (\lambda x. f'\ ()\ k) &&+ 1/2 \cdot f'\ ()\ k&&) \\
	&\quad 1/2 \cdot (1/2 \cdot (0, 0) &&+ 1/2 \cdot k\ ()&&) 
\end{alignat}
Here, coin flips are CPS-transformed to weighted sums, and the outer weighted sum corresponds to $b_1$ and the inner to $b_2$.
Note that $\mathrm{observe}(\mathbf{true})$ is CPS-transformed to $1 \cdot ({-})$ and $\mathrm{observe}(\mathbf{false})$ is CPS-transformed to $0 \cdot ({-})$.
By passing the post-expectation $k = \lambda x. (1, 1)$, the term $f'\ ()\ k$ gives a pair of the unnormalised termination probability and the probability that all observations are satisfied.
Now, the conditional termination probability is obtained as the quotient of the first component divided by the second component.

\section{Source Language}\label{sec:source-language}
We define the syntax and semantics of a functional language with computational effects and recursion.
Technically, we define the $\lambda_c$-calculus with algebraic effects and recursion following~\cite{katsumata2013}, which is expressive enough to write the programs in Section~\ref{sec:examples}.
We will use it as the source language of the CPS transformation in Section~\ref{subsec:cps-transformation}.

\subsection{Syntax}
\begin{definition}[types and ground types]
	Let $B$ be a set of base types.
	We define the set $\mathbf{Typ}(B)$ of \emph{types} and its subset $\mathbf{GTyp}(B)$ of \emph{ground types} as follows where $b$ ranges over $B$.
	\begin{align}
		\mathbf{Typ}(B) \ni\qquad \rho, \tau &\quad\coloneqq\quad b \mid 1 \mid \rho \times \tau \mid 0 \mid \rho + \tau \mid \rho \to \tau \\
		\mathbf{GTyp}(B) \ni\qquad \rho, \tau &\quad\coloneqq\quad b \mid 1 \mid \rho \times \tau \mid 0 \mid \rho + \tau
	\end{align}
	That is, types are built from base types $b \in B$, unit type $1$, product types ${\times}$, empty type $0$, coproduct types ${+}$, and function types $\to$.
	We write $\underline{n} \coloneqq 1 + \dots + 1$ for the $n$-fold coproduct of $1$ (Table~\ref{tab:syntactic-sugar}).
\end{definition}

The source language is parameterised to cover various situations.
We define the parameter for syntax as follows.
\begin{definition}[$\lambda_c$-signature]\label{def:lambda-c-signature}
	A \emph{$\lambda_c$-signature} $\Sigma$ is a tuple $(B, K, O, \mathrm{ar}, \mathrm{car})$ where $B$ is a set of base types, $K$ is a set of symbols for effect-free constants, $O$ is a set of symbols for algebraic operations, and $\mathrm{ar}, \mathrm{car} : K \cup O \to \mathbf{GTyp}(B)$ are functions assigning arities and coarities, respectively, to constants and algebraic operations.
	We sometimes write $\Sigma = (B, K, O) = (B, \{ \dots, c : \mathrm{ar}(c) \rightarrowtriangle \mathrm{car}(c), \dots \}, \{ \dots, o : \mathrm{ar}(o) \rightarrowtriangle \mathrm{car}(o), \dots \})$ as a convenient notation.
	We say $o \in O$ is a an \emph{$n$-ary algebraic operation} if $\mathrm{ar}(o) = \underline{n}$ and $\mathrm{car}(o) = 1$.
\end{definition}

For example, we often consider the following base types: the type of integers is denoted by $\mathbf{int}$, and the type of real numbers is denoted by $\mathbf{real}$.
Typical examples of effect-free constants include basic arithmetic operators (e.g.\ ${+} : \mathbf{int} \times \mathbf{int} \rightarrowtriangle \mathbf{int}$) and comparison operators (e.g.\ ${\le} : \mathbf{int} \times \mathbf{int} \rightarrowtriangle 1 + 1$ where $1 + 1$ is used as the type of boolean values).
Our language include algebraic operations as primitives that cause computational effects.
The probabilistic branching operator ${+}_{p} : 1 + 1 \rightarrowtriangle 1$ for probabilistic programs is an example of an algebraic operation.
The term $M_1 +_{p} M_2$ tosses a biased coin and invokes $M_1$ with probability $p$ and $M_2$ with probability $1 - p$.

\begin{definition}[terms]\label{def:source-term}
	Given a $\lambda_c$-signature $\Sigma = (B, K, O, \mathrm{ar}, \mathrm{car})$, \emph{terms} of the $\lambda_c$-calculus are defined by variables, effect-free constant for $c \in K$, generic effect for $o \in O$, nullary and binary tuples, projections, lambda abstractions, applications, injections, nullary and binary case analyses for $M : 0$ and $M : \rho_1 + \rho_2$, respectively, and recursion:
	\begin{align}
		M, N \quad\coloneqq\quad &x \mid c\ M \mid \geneff{o}\ M \mid () \mid (M, N) \mid \pi_i M \mid \lambda x {:} \rho. M \mid M\ N \\
		&\quad \mid \iota_i M \mid \delta(M) \mid \delta(M, x_1 {:} \rho_1. N_1, x_2 {:} \rho_2. N_2) \mid \mu f : \rho \to \tau. M \qquad\text{where $i \in \{1, 2\}$.} \label{eq:lambda-c-calculus}
	\end{align}
	We sometimes omit type annotations and write, e.g., $\delta(M, x_1. N_1, x_2. N_2)$, $\lambda x. M$, and $\mu f. M$.
	Types and terms are sometimes referred to as \emph{$\lambda_c(\Sigma)$-types} and \emph{$\lambda_c(\Sigma)$-terms} to make explicit that they are types and terms of the $\lambda_c$-calculus with a $\lambda_c$-signature given by $\Sigma$.
	We often omit $()$ in $c\ ()$ and $\geneff{o}\ ()$ and just write $c$ and $\geneff{o}$, respectively.
	Table~\ref{tab:syntactic-sugar} defines useful syntactic sugar for terms.
\end{definition}
It is well-known~\cite{plotkin2003} that \emph{generic effects} bijectively correspond to \emph{algebraic operations}.
For example, the algebraic operation for probabilistic branching ${+}_{p} : 1 + 1 \rightarrowtriangle 1$ corresponds to the generic effect $\geneff{{+}_{p}}\ () : 1 + 1$, which represents a program that simply tosses a coin.
In Definition~\ref{def:source-term}, we formulate the source language using generic effects because it makes our CPS transformation slightly simpler.
Using the bijective correspondence, we can define a $\lambda_c$-term $o\ (M, N)$ as syntactic sugar for $M\ (\geneff{o}\ N)$ where $o \in O$, $M : \mathrm{ar}(o) \to \rho$, and $N : \mathrm{car}(o)$.

\begin{definition}[well-typed terms]
	A \emph{context} is a list of pairs of variables and types: $\Gamma \coloneqq x_1 {:} \rho_1, \dots, x_n {:} \rho_n$.
	A \emph{well-typed term} $\Gamma \vdash M : \rho$ is defined by standard typing rules (see \referappendix{subsec:source-typing-rules}{A.1} for the full typing rules).
	Specifically, the typing rules for effect-free constants, generic effects, nullary case analysis, and recursion are as follows.
	\begin{mathpar}
		\inferrule{
			\Gamma \vdash M : \mathrm{ar}(c)
		}{
			\Gamma \vdash c\ M : \mathrm{car}(c)
		}
		\and
		\inferrule{
			\Gamma \vdash M : \mathrm{car}(o)
		}{
			\Gamma \vdash \geneff{o}\ M : \mathrm{ar}(o)
		}
		\and
		\inferrule{
			\Gamma \vdash M : 0
		}{
			\Gamma \vdash \delta(M) : \rho
		}
		\and
		\inferrule{
			\Gamma, f : \rho \to \tau \vdash M : \rho \to \tau
		}{
			\Gamma \vdash \mu f : \rho \to \tau. M : \rho \to \tau
		}
	\end{mathpar}
	Note that the typing rule for $\geneff{o}\ M$ might be a bit confusing, but we follow the standard terminology for the arity and the coarity of algebraic effects.
\end{definition}

\begin{table}
	\caption{List of syntactic sugar for types and terms.}\label{tab:syntactic-sugar}
	\begin{tabular}{l|l}
		Syntactic sugar & Meaning \\
		\hline
		$\mathbf{bool}$ \hspace{2em} (boolean type) & $1 + 1$ \\
		$\underline{n}$ & $\underline{0} = 0$, $\underline{1} = 1$, and $\underline{n+1} = \underline{n} + 1$ if $n \ge 1$ \\
		$\mathbf{let}\ x = M\ \mathbf{in}\ N$ & $(\lambda x. N)\ M$ \\
		$M; N$ & $\mathbf{let}\ x = M\ \mathbf{in}\ N$\hspace{2em} (if $x$ does not occur in $N$)\\
		$\letrecchurch{f}{x}{\rho}{\tau}{M}{N}$ & $\mathbf{let}\ f = \mu f. \lambda x. M\ \mathbf{in}\ N$ \\
		$\mathbf{if}\ M\ \mathbf{then}\ N_1\ \mathbf{else}\ N_2$ \hspace{2em} (for $M : \mathbf{bool}$) & $\delta(M, z_1. N_1, z_2. N_2)$ \hspace{2em} ($z_1, z_2$ are fresh)
	\end{tabular}
\end{table}

\begin{example}[total/partial correctness]\label{ex:total-partial-correctness-syntax}
	As the simplest case, we consider a $\lambda_c$-signature $\Sigma = (B, K, O)$ with no algebraic operation $O = \emptyset$.
	In this situation, we can write a program like the factorial function in Section~\ref{subsec:example-total-partial} if $B$ contains a type $\mathbf{int}$ of integers and $K$ contains basic operators like comparison ${\le}_{\mathbf{int}} : \mathbf{int} \times \mathbf{int} \rightarrowtriangle 1 + 1$ and multiplication $({\cdot}) : \mathbf{int} \times \mathbf{int} \rightarrowtriangle \mathbf{int}$.
\end{example}

\begin{example}[safety property]\label{ex:trace-property-syntax}
	We define a $\lambda_c$-signature that covers the situations in Section~\ref{subsec:trace-property}.
	Let $E$ be a finite set of \emph{events} and $\Sigma = (B, K, O)$ be a $\lambda_c$-signature where $O$ consists of a unary operation $\operation{event}_a : 1 \rightarrowtriangle 1$ for outputting an event $a \in E$ and a binary operation ${\join} : 1 + 1 \rightarrowtriangle 1$ for nondeterministic branching.
	Now, we can write nondeterministic branching $\ifexpr{{*}}{M}{N}$ and event operations $\mathrm{ev\_read}()$ and $\mathrm{ev\_close}()$ in Section~\ref{subsec:trace-property} using generic effects because the generic effect $\geneff{\join} : 1 + 1$ for ${\join}$ represents nondeterministic choice between true and false; and the generic effect $\geneff{\operation{event}_a} : 1$ represents outputting an event $a \in E$.
	Note that the output operation is sometimes written as $\operation{write}_a$ in the literature (e.g.~\cite{plotkin2003}), but we stick to the same notation as~\cite{kobayashi2018}.
\end{example}

\begin{example}[expected cost analysis and cost moment analysis]\label{ex:expected-cost-syntax}
	We define a $\lambda_c$-signature for expected cost analyses~\cite{avanzini2021} (Section~\ref{subsec:expected-cost}) and cost moment analyses (Section~\ref{subsec:cost-moment}).
	As algebraic operations, we consider a binary probabilistic branching (or Bernoulli distribution) ${+}_p : 1 + 1 \rightarrowtriangle 1$ for any $p \in [0, 1]$ and a unary tick operator $({-})^{\checkmark} : 1 \rightarrowtriangle 1$ for incrementing the accumulated cost.
	We also allow continuous distributions, which makes the setting slightly beyond~\cite{avanzini2021}.
	For simplicity, we consider the uniform distribution on the unit interval $[0, 1]$ as the only continuous distribution.
	Thus, we define a $\lambda_c$-signature $\Sigma = (B, K, O)$ by $O = \{ {+}_p \mid p \in [0, 1] \} \cup \{ ({-})^{\checkmark}, \operation{unif} \}$ and $\mathbf{real} \in B$ where $\operation{unif} : \mathbf{real} \rightarrowtriangle 1$ is an algebraic operation for sampling from the uniform distribution, that is, $\geneff{\operation{unif}} : \mathbf{real}$ samples a real number from the uniform distribution over $[0, 1]$.
	We define $\lambda_c$-terms $M^{\checkmark}$ and $M +_{p} N$ as syntactic sugars for $M^{\checkmark} \coloneqq \geneff{({-})^{\checkmark}}; M$ and $M +_{p} N \coloneqq \ifexpr{\geneff{{+}_p}}{M}{N}$, respectively.
\end{example}

\begin{example}[conditional weakest preexpectation]\label{ex:conditional-wp-syntax}
	We define a $\lambda_c$-signature for conditional weakest preexpectation~\cite{olmedo2018}.
	Similarly to Example~\ref{ex:expected-cost-syntax}, we consider algebraic operations for probabilistic branching ${+}_p : 1 + 1 \rightarrowtriangle 1$ and sampling from the uniform distribution $\operation{unif} : \mathbf{real} \rightarrowtriangle 1$.
	Moreover, we add an algebraic operation for (soft) conditioning $\operation{score} : 1 \rightarrowtriangle \mathbf{real}_{[0, 1]}$ where $\mathbf{real}_{[0, 1]}$ is the type of real numbers in the unit interval $[0, 1]$.
	That is, $\geneff{\operation{score}}\ L$ re-weights the likelihood of the program trace by multiplying $L \in [0, 1]$.
	Hard conditioning in Section~\ref{subsec:conditional-wp} is a special case of soft conditioning: We can think of $\mathrm{observe}(\mathbf{true})$ as $\geneff{\operation{score}}\ 1$ and $\mathrm{observe}(\mathbf{false})$ as $\geneff{\operation{score}}\ 0$.
	To sum up, we define a $\lambda_c$-signature $\Sigma = (B, K, O)$ by $O = \{ {+}_p : 1 + 1 \rightarrowtriangle 1 \mid p \in [0, 1] \} \cup \{ \operation{unif} : \mathbf{real} \rightarrowtriangle 1, \operation{score} : 1 \rightarrowtriangle \mathbf{real}_{[0, 1]} \}$ and $\mathbf{real}, \mathbf{real}_{[0, 1]} \in B$.
\end{example}

\subsection{Semantics}
We will explain the interpretation of the source language $\lambda_c$.
The interpretation is rather standard, and the outline is as follows.

A categorical model of pure simply typed $\lambda$-calculus with product/coproduct types is given by a bicartesian closed category $\category{C}$~\cite{lambek1986}.
Product and coproduct types are interpreted by categorical products and coproducts in $\category{C}$, respectively, and function types are interpreted by exponential objects in $\category{C}$.
A well-typed term $x_1 : \sigma_1, \dots, x_n : \sigma_n \vdash M : \sigma$ is interpreted as $\interpret{M} : \interpret{\sigma_1} \times \dots \times \interpret{\sigma_n} \to \interpret{\sigma}$.

When we add computational effects to simply typed $\lambda$-calculus, we need a strong monad $T$ to interpret computational effects~\cite{moggi1989}.
In the $\lambda_c$-calculus, function types are interpreted by Kleisli exponentials $\interpret{\sigma \to \tau} = \exponential{\interpret{\sigma}}{T \interpret{\tau}}$ instead of mere exponentials because functions may cause computational effects.
Now, a well-typed term $x_1 : \sigma_1, \dots, x_n : \sigma_n \vdash M : \sigma$ is interpreted as a morphism $\interpret{M} : \interpret{\sigma_1} \times \dots \times \interpret{\sigma_n} \to T \interpret{\sigma}$.

If we further add recursion to a language, we need a (parameterised) fixed-point operator~\cite{simpson2000}, which gives a fixed point $f^{\dagger} : X \to Y$ of a morphism $f : X \times Y \to Y$.
We consider $\omegaCPO$-enriched bicartesian closed categories so that we can define a fixed-point operator by least fixed points in hom-$\omega$cpos.

\subsubsection{Without Recursion}\label{subsubsec:source-semantics-without-recursion}
For the sake of simplicity, we define the semantics of the recursion-free fragment first and then extend it to the full $\lambda_c$-calculus with recursion.

Let $\category{C}$ be a bicartesian closed category (i.e.\ a category with finite products, finite coproducts, and exponential objects) and $T : \category{C} \to \category{C}$ be a strong monad on $\category{C}$ with unit $\eta^{T}_X : X \to T X$, multiplication $\mu^{T}_X : T^2 X \to T X$, and strength $\strength^T_{X, Y} : X \times T Y \to T (X \times Y)$ where $X, Y \in \category{C}$.
Note that we often omit subscripts of natural transformations when they are clear from the context.

The interpretation of our source language $\lambda_c$ is parameterised by the following data.
\begin{definition}[$\lambda_c(\Sigma)$-structure]\label{def:lambda-c-structure}
	Let $\Sigma$ be a $\lambda_c$-signature.
	A \emph{$\lambda_c(\Sigma)$-structure} is a tuple $\mathcal{A} = (\category{C}, T, A, a)$ where $\category{C}$ is a bicartesian closed category, $T$ is a strong monad on $\category{C}$, $A : B \to \category{C}$ is a mapping that assigns an interpretation to each base type, and $a$ assigns an interpretation to each constant and algebraic operation as follows.
	First, we extend $A : B \to \category{C}$ to $\mathcal{A}\interpret{-} : \mathbf{GTyp}(B) \to \category{C}$ using the bicartesian structure of $\category{C}$.
	Then, the function $a$ assigns a morphism $a(c) : \mathcal{A} \interpret{\mathrm{ar}(c)} \to \mathcal{A} \interpret{\mathrm{car}(c)}$ for each $c \in K$, and $a(o) : \mathcal{A}\interpret{\mathrm{car}(o)} \to T \mathcal{A}\interpret{\mathrm{ar}(o)}$ for each $o \in O$.
\end{definition}

\begin{definition}[interpretation of the $\lambda_c$-calculus]\label{def:interpretation-lambda-c-calculus}
	Let $\mathcal{A} = (\category{C}, T, A, a)$ be a $\lambda_c(\Sigma)$-structure.
	We define the interpretation $\mathcal{A}\interpret{-}$ of the $\lambda_c$-calculus as follows.
	For each type $\rho \in \mathbf{Typ}(B)$, $\mathcal{A}\interpret{\rho} \in \category{C}$ is defined by $\mathcal{A}\interpret{\rho \to \tau} = \exponential{\mathcal{A}\interpret{\rho}}{T \mathcal{A}\interpret{\tau}}$ and for other type constructions, defined using the bicartesian structure of $\category{C}$.
	We interpret contexts by $\mathcal{A} \interpret{x_1 : \rho_1, \dots, x_n : \rho_n} \coloneqq \mathcal{A} \interpret{\rho_1} \times \dots \mathcal{A} \interpret{\rho_n} \in \category{C}$.
	For each well-typed term $\Gamma \vdash M : \rho$, its interpretation $\mathcal{A} \interpret{M} : \mathcal{A} \interpret{\Gamma} \to T \mathcal{A} \interpret{\rho}$ is defined in the standard way.
	Specifically, constants and algebraic operations are interpreted as follows.
	\begin{gather}
		\mathcal{A} \interpret{c\ M} \ \coloneqq\ T(a(c)) \comp \mathcal{A} \interpret{M} \qquad\qquad
		\mathcal{A} \interpret{\geneff{o}\ M} \ \coloneqq\ \mu^T \comp T a(o) \comp \mathcal{A} \interpret{M}
	\end{gather}
	See \referappendix{subsec:source-semantics}{A.2} for the full definition.
\end{definition}

\subsubsection{With Recursion}
To interpret recursion, we consider fixed-point operators defined by the least fixed point in $\omega$cpos.
Let $T$ be a strong monad on a cartesian closed category $\category{C}$.
An \emph{Eilenberg--Moore $T$-algebra} (or \emph{EM algebra}) is a pair of an object $A \in \category{C}$ and a morphism $\emalgsymbol : T A \to A$ such that $\emalgsymbol \comp \eta^T_A = \identity{A}$ and $\emalgsymbol \comp T \emalgsymbol = \emalgsymbol \comp \mu^T_A$.
We recall two facts on EM algebras for later use.
First, $\mu^T_X : T^2 X \to T X$ is a (free) EM algebra for any $X \in \category{C}$.
Second, if $\category{C}$ is cartesian closed and $T$ is a strong monad, then $\exponential{X}{A}$ has an EM algebra structure for any $X \in \category{C}$ and an EM algebra $\emalgsymbol : T A \to A$.

A \emph{uniform $T$-fixed-point operator} is a mapping $({-})^{\dagger} : \category{C}(T X, T X) \to \category{C}(1, T X)$ satisfying the fixed-point property $f^{\dagger} = f \comp f^{\dagger}$ and the uniformity $g \comp h = h \comp f \implies g^{\dagger} = h \comp f^{\dagger}$ for any $f : T X \to T X$, $g : T Y \to T Y$, and $h : T X \to T Y$ such that $h$ is a morphism between free EM algebras (i.e.\ $h \comp \mu^T_X = \mu^T_Y \comp T h$).
It is known that any uniform $T$-fixed-point operator can be extended to a parameterized uniform fixed-point operator $({-})^{\dagger} : \category{C}(X \times A, A) \to \category{C}(X, A)$ for each $X \in \category{C}$ and EM $T$-algebra $\emalgsymbol : T A \to A$ (see~\cite{hasegawa2002} for details); and we use this to interpret recursion.

Let $\omegaCPO$ be the category of $\omega$cpos and Scott-continuous functions.
An \emph{$\omegaCPO$-enriched bicartesian closed category} is a bicartesian closed category $\category{C}$ such that homsets are $\omega$cpos (i.e., objects in $\omegaCPO$); and composition ${-} \comp {-}$, tupling $\langle {-}, {-} \rangle$, cotupling $[{-}, {-}]$, and currying $\Lambda_{X, Y, Z} : \category{C}(X \times Y, Z) \to \category{C}(X, \exponential{Y}{Z})$ are Scott-continuous functions (i.e., morphisms in $\omegaCPO$).
For example, $\omegaCPO$ itself is an $\omegaCPO$-enriched bicartesian closed category with the pointwise order on each hom-set.
The underlying ordinary bicartesian closed category of $\category{C}$ is denoted by $\category{C}_0$.
A \emph{pseudo-lifting strong monad} $T$ on $\category{C}$ is an ordinary strong monad $T$ on $\category{C}_0$ that has a generic effect $\bot_{T 0} : 1 \to T 0$ such that $\bot_{T X} \coloneqq T {?} \comp \bot_{T 0} : 1 \to T X$ is the least morphism in $\category{C}(1, T X)$ for any $X \in \category{C}$~\cite{katsumata2013} where ${?} : 0 \to X$ is the unique morphism from the initial object $0$.
The simplest example of a pseudo-lifting strong monad is the lifting monad $({-})_{\bot} \coloneqq ({-}) \cup \{ \bot \}$ on $\omegaCPO$.
For any strong monad $T$, if we have a strong monad morphism $\phi : ({-})_{\bot} \to T$, then $T$ is a pseudo-lifting strong monad.

Given a pseudo-lifting strong monad $T$, we can define a uniform $T$-fixed-point operator $f \mapsto f^{\dagger}$ by the least fixed point of $f \comp ({-}) : \category{C}(1, T X) \to \category{C}(1, T X)$, that is, $f^{\dagger} \coloneqq \sup_n (f^n \comp \bot_{T X})$.
It follows that the parameterized uniform fixed-point operator $({-})^{\dagger} : \category{C}(X \times A, A) \to \category{C}(X, A)$ induced by a pseudo-lifting strong monad is the least fixed point of $f \comp \tupling{\identity{}}{-} : \category{C}(X, A) \to \category{C}(X, A)$.
\begin{lemma}\label{lem:parameterised-uniform-T-fixed-point-lfp}
	Let $T$ be a pseudo-lifting strong monad on a $\omegaCPO$-enriched cartesian closed category $\category{C}$.
	For any EM algebra $\emalgsymbol : T A \to A$, and morphisms $f : X \times A \to A$ and $g : X \to A$,
	if $f \comp \tupling{\identity{}}{g} \le g$, then $f^{\dagger} \le g$.
	\qed
\end{lemma}
By Lemma~\ref{lem:parameterised-uniform-T-fixed-point-lfp}, for any pseudo-lifting strong monad $T$ on $\omegaCPO$, EM algebra $\emalgsymbol : T A \to A$, and morphism $f : X \times A \to A$, the parameterised uniform fixed-point operator $({-})^{\dagger}$ gives a function $f^{\dagger} : X \to A$ such that for any element $x \in X$, $f^{\dagger}(x)$ is the least fixed point of $f(x, {-}) : A \to A$.

Now, we extend the interpretation defined in Section~\ref{subsubsec:source-semantics-without-recursion} to recursive programs.
\begin{definition}
	An \emph{$\omegaCPO$-enriched $\lambda_c(\Sigma)$-structure} is a tuple $\mathcal{A} = (\category{C}, T, A, a)$ where $\category{C}$ is an $\omegaCPO$-enriched bicartesian closed category, $T$ is a pseudo-lifting strong monad, and $(\category{C}_0, T, A, a)$ is a $\lambda_c(\Sigma)$-structure.
	From now on, we always consider the $\omegaCPO$-enriched setting and may omit ``$\omegaCPO$-enriched''.
\end{definition}

\begin{definition}[interpretation of recursion]
	We extend the interpretation $\mathcal{A} \interpret{-}$ in Def.~\ref{def:interpretation-lambda-c-calculus} using the uniform $T$-fixed-point operator $({-})^{\dagger}$ induced by the pseudo-lifting strong monad $T$.
	\[ \mathcal{A} \interpret{\mu f. M} \ \coloneqq\ (\mathcal{A} \interpret{M})^{\dagger} \qquad \text{where \hspace{1em} $\mathcal{A} \interpret{M} : \mathcal{A} \interpret{\Gamma} \times (\exponential{\mathcal{A} \interpret{\rho}}{T \mathcal{A} \interpret{\tau}}) \to \exponential{\mathcal{A} \interpret{\rho}}{T \mathcal{A} \interpret{\tau}}$} \]
	Here, we can apply $({-})^{\dagger}$ to $\mathcal{A} \interpret{M}$ because $\exponential{\mathcal{A} \interpret{\rho}}{T \mathcal{A} \interpret{\tau}}$ has an EM $T$-algebra structure.
\end{definition}

\begin{example}[total/partial correctness]\label{ex:total-partial-correctness-semantics}
	An interpretation for the $\lambda_c$-signature in Example~\ref{ex:total-partial-correctness-syntax} is given by an $\omegaCPO$-enriched $\lambda_c(\Sigma)$-structure $\mathcal{A} = (\omegaCPO, ({-})_{\bot}, A, a)$ where $({-})_{\bot}$ is the lifting monad, $A(\mathbf{int})$ is defined by $(\mathbb{Z}, {=})$, and the interpretation $a(c)$ of each effect-free constant $c \in K$ is defined in the obvious way.
\end{example}

For safety properties, we define a strong monad by an algebraic theory in $\omegaCPO$.
An \emph{algebraic theory} is defined by a pair $(\Sigma, E)$ of a set $\Sigma$ of operations of at most countable arities and a set $E$ of equations and inequations between terms constructed from $\Sigma$.
A \emph{$(\Sigma, E)$-algebra} is defined by an $\omega$cpo $A$ together with an interpretation of each operation over $A$ such that all (in)equations in $E$ are satisfied.
It is known that an algebraic theory induces a strong monad $T$ on $\omegaCPO$ where $T X$ is a free algebra generated by $X$ (see~\cite[Section~6]{abramsky1994} and~\cite{hyland2006}).

\begin{example}[safety property]\label{ex:trace-property-semantics}
To interpret the $\lambda_c$-signature in Example~\ref{ex:trace-property-syntax}, we define an $\omegaCPO$-enriched $\lambda_c(\Sigma)$-structure by $\mathcal{A} = (\omegaCPO, T^P, A, a)$.
To define a pseudo-lifting strong monad $T^P$, we consider an algebraic theory $\mathcal{T}_P$ in $\omegaCPO$ defined by a nullary operation $\nontermconst$, a unary operation $\operation{event}_a$ for any $a \in E$, and a binary operation ${\join}$ together with the following (in)equations.
\begin{gather}
x \join x\ =\ x \qquad
x \join y\ =\ y \join x \qquad
(x \join y) \join z\ =\ x \join (y \join z) \\
\operation{event}_a(x \join y)\ =\ \operation{event}_a(x) \join \operation{event}_a(y) \qquad
x\ \ge\ \nontermconst
\label{eq:hoare-powerdomain}
\end{gather}
Intuitively, $\nontermconst$ represents divergence, $\operation{event}_a$ represents outputting an event $a$, and ${\join}$ represents nondeterministic branching, that is, $x \join y$ intuitively means $\ifexpr{{*}}{x}{y}$ in Section~\ref{subsec:trace-property}.
Note that the first three equations are the axioms for Plotkin powerdomains (or semilattices).
We will show later that this theory is consistent (i.e.\ we cannot derive $x = y$ for two different variables $x, y$) by giving a non-trivial $\mathcal{T}_P$-algebra.
Now, for any $X \in \omegaCPO$, we define $T^P X$ by a free $\mathcal{T}_P$-algebra generated by $X$.
Then, it is straightforward to define interpretations of base types $B$, effect-free constants $K$, and algebraic operations $O = \{ \operation{event}_a, {\join} \}$.
\end{example}

For expected cost analysis, cost moment analysis, and conditional weakest pre-expectations, we use the $\omegaCPO$-enriched bicartesian closed category of $\omega$qbses~\cite{vakar2019} to interpret continuous distributions and recursion.
We recall basic definitions.

\begin{definition}[$\omegaQBS$~{\cite[Def.~3.5]{vakar2019}}]
	A \emph{quasi-Borel space} (or qbs) is a tuple $(|P|, M_P)$ where $|P|$ is a set, and $M_P \subseteq (\exponential{\real}{|P|})$ is a set of \emph{random elements} satisfying a certain condition (see~\cite{heunen2017} for details).
	An \emph{$\omega$qbs} is a tuple $P = (|P|, M_P, {\le_P})$ where $(|P|, M_P)$ is a qbs and $(|P|, {\le_P})$ is an $\omega$cpo.
	A \emph{morphism} $f : P \to Q$ between $\omega$qbses is a function $f : |P| \to |Q|$ such that $f$ is a morphism between underlying qbses (i.e.\ if $\alpha \in M_P$, then $f \comp \alpha \in M_Q$) and $f$ is Scott-continuous w.r.t.\ the underlying $\omega$cpos.
	Let $\omegaQBS$ be the category of $\omega$qbses and morphisms between them.
\end{definition}
For example, $\mathbb{W} = ([0, \infty], \mathbf{Meas}(\real, [0, \infty]), {\le}_{[0, \infty]})$ is the $\omega$qbs of real weights~\cite[Example~3.6, 3.7]{vakar2019} where $\mathbf{Meas}(\real, [0, \infty])$ is the set of measurable functions from $\real$ to $[0, \infty]$ and ${\le}_{[0, \infty]}$ is the standard order on $[0, \infty]$.
It is known that $\omegaQBS$ is bicartesian closed and $\omegaCPO$-enriched by the pointwise order~\cite{vakar2019}.

To combine the computational effects of probability, cost, and nontermination induced by recursion, we consider combining strong monads via distributive laws.
\begin{definition}[distributive law]
	A \emph{distributive law} between strong monads is a natural transformation $d : T S \to S T$ satisfying the following five equations for each $X, Y \in \category{C}$.
	\begin{gather}
		\begin{aligned}
			\eta^S_{T X} &\ =\ d_X \comp T \eta^S_X \qquad\qquad & d \comp T \mu^S &\ =\  \mu^S_{T X} \comp S d_X \comp d_{S X} \\
			S \eta^T_{X} &\ =\ d_X \comp \eta^T_{S X} & d_X \comp \mu^T_{S X} &\ =\ S \mu^T_X \comp d_{T X} \comp T d_X
		\end{aligned} \\
		d_{X \times Y} \comp T \strength^S_{X, Y} \comp \strength^T_{X, S Y} \ =\ S \strength^T_{X, Y} \comp \strength^S_{X, T Y}  \comp (X \times d_Y)
	\end{gather}
\end{definition}

\begin{lemma}\label{lem:strong-distributive-law}
	If there is a distributive law $d : T S \to S T$ between strong monads, then $S T$ has a strong monad structure, and there exist strong monad morphisms $\phi_S : S \to S T$ and $\phi_T : T \to S T$.
	\qed
\end{lemma}

Lemma~\ref{lem:strong-distributive-law} is useful to define a $\lambda_c(\Sigma)$-structure for the composite strong monad $S T$.
Given a strong monad morphism $\phi : T_1 \to T_2$ and (an interpretation of) a generic effect $a(o) : C \to T_1 D$ for $T_1$, we have a generic effect for $T_2$ given by $\phi \comp a(o) : C \to T_2 D$.
Thus, by Lemma~\ref{lem:strong-distributive-law}, a distributive law $d : T S \to S T$ naturally induces generic effects for $S T$ from those for $S$ and $T$.

\begin{example}[expected cost analysis and cost moment analysis]\label{ex:expected-cost-cost-moment-semantics}
	To interpret the $\lambda_c$-signature in Example~\ref{ex:expected-cost-syntax}, we define an $\omegaCPO$-enriched $\lambda_c(\Sigma)$-structure by $\mathcal{A} = (\omegaQBS, P(\mathbb{W} \times ({-})_{\bot}), A, a)$ where $P(\mathbb{W} \times ({-})_{\bot})$ is the composite of three strong monads on $\omegaQBS$.
	(1) The lifting monad $({-})_{\bot}$ that adds a bottom element to an $\omega$qbs~\cite[Section~3]{vakar2019}.
	(2) The writer monad $\mathbb{W} \times ({-})$ induced by the additive monoid $(\mathbb{W}, 0, {+})$.
	(3) The probabilistic powerdomain monad $P$~\cite[Section~4]{vakar2019}.
	Using distributive laws between these strong monads, we get a pseudo-lifting strong monad $P(\mathbb{W} \times ({-})_{\bot})$ (see \referappendix{sec:detail-instance}{G} for details).
	The writer monad gives a natural interpretation of the tick operation $({-})^{\checkmark}$, and the probabilistic powerdomain monad gives interpretations of the probabilistic branching ${+}_p$ and the uniform distribution $\operation{unif}$.
	Thus, the composite monad $P(\mathbb{W} \times ({-})_{\bot})$ inherits these interpretations.
\end{example}

\begin{example}[conditional weakest preexpectation]\label{ex:conditional-wp-semantics}
	To interpret the $\lambda_c$-signature in Example~\ref{ex:conditional-wp-syntax}, we define an $\omegaCPO$-enriched $\lambda_c(\Sigma)$-structure by $\mathcal{A} = (\omegaQBS, P S, A, a)$, in which we use the composite of the probabilistic powerdomain monad $P$ and a monad $S$ that can interpret conditioning $\operation{score}$.
	We define $S$ by the free algebras of the following algebraic theory.
	\[ \bot \le x \qquad\qquad \operation{score}_1(x) = x \qquad\qquad \operation{score}_r(\operation{score}_s(x)) = \operation{score}_{r \cdot s}(x) \quad \text{for each $r, s \in [0, 1]$} \]
	where $\bot$ is a nullary operation (i.e.\ a constant) and $\operation{score}$ is an operation with arity $1$ and coarity $[0, 1]$.
	For $X = (|X|, M_X, {\le}_X) \in \omegaQBS$, we define $S X = (|S X|, M_{S X}, {\le}_{S X}) \in \omegaQBS$ by the free algebra generated by $X$.
	Concretely, $S X$ is given as follows.
	The underlying set is defined by $|S X| \coloneqq \{ \operation{score}_r(\bot^{S X}) \mid r \in [0, 1] \} + \{ \operation{score}_r(x) \mid r \in [0, 1], x \in X \}$.
	We often omit the superscript in $\bot^{S X}$.
	The set of random elements $M_{S X}$ is induced by the bijection $|S X| \cong [0, 1] + [0, 1] \times |X|$ and the standard construction of products and coproducts of quasi-Borel spaces.
	The order relation ${\le}_{S X}$ is defined by (a) $\operation{score}_r(\bot) \le \operation{score}_s(\bot)$ if and only if $r \ge s$, (b) $\operation{score}_r(\bot^{S X}) \le_{S X} \operation{score}_s(x)$ if and only if $r \ge s$ for any $x \in X$, and (c) $\operation{score}_r(x) \le_{S X} \operation{score}_s(y)$ if and only if $r = s$ and $x \le_X y$ for any $r, s \in [0, 1]$ and $x, y \in X$.
	We can verify that $S$ is a pseudo-lifting strong monad where the unit $\eta^S$ is defined by $\eta(x) \coloneqq \operation{score}_1(x)$ and the multiplication $\mu^S$ is defined by re-weighting $\mu^S(\operation{score}_r(\operation{score}_s(x))) = \operation{score}_{r \cdot s}(x)$ for any $x \in X$ or $x = \bot$.
	Note that the interpretation of $\operation{score}$ in $S$ distinguishes $\operation{score}_1(\bot) = \bot$ (diverge) and $\operation{score}_0(\bot)$ ($\operation{score}$ by 0 and then diverge), which is important when we reason about the conditional termination probability of the following program.
	\begin{equation}
		\letrec{\mathrm{diverge}}{x}{\mathrm{diverge}\ x}{\quad () +_{1/2} (\mathrm{observe}(\mathbf{false});\ \mathrm{diverge}\ ())}
		\label{eq:score-diverge}
	\end{equation}
	The conditional termination probability of~\eqref{eq:score-diverge} should be 1.
	However, if we have $\bot = \operation{score}_0(\bot)$, then we cannot distinguish~\eqref{eq:score-diverge} with $() +_{1/2} \mathrm{diverge}\ ()$, whose termination probability is $1/2$.
\end{example}

\section{Weakest Precondition Transformers}\label{sec:wpt}
In this section, we explain \emph{semantic} weakest precondition transformers studied in~\cite{aguirre2020}.
Let $T$ be a monad on $\category{C}$ and $f : X \to T Y$ be a morphism that represents the interpretation of a program (Def.~\ref{def:interpretation-lambda-c-calculus}).
Here, we forget the syntax of programs and instead focus on the semantics.
Take an object $\answerobj \in \category{C}$ that represents a set of truth values.
A typical choice of $\answerobj$ is the two-element boolean algebra, but there are other examples as we will show later.
We call a morphism $q : Y \to \answerobj$ a postcondition and $p : X \to \answerobj$ a precondition.
\begin{definition}\label{def:weakest-precondition}
	Let $T$ be a monad on $\category{C}$ and $\emalgsymbol : T \answerobj \to \answerobj$ be an EM $T$-algebra.
	For each $f : X \to T Y$ and $q : Y \to \answerobj$, we define $\mathrm{wp}^{\emalgsymbol}[f](q) : X \to \answerobj$ as follows.
	\begin{equation}
		\mathrm{wp}^{\emalgsymbol}[f](q) \quad=\quad \emalgsymbol \comp T q \comp f \label{eq:wp-for-lax-slice}
	\end{equation}
\end{definition}
Note that for pure programs, we don't need an EM $T$-algebra since $T = \mathrm{Id}_{\category{C}}$ is the identity functor, and the weakest precondition $\mathrm{wp}[f](q) = q \comp f$ is just the inverse image of a given postcondition $q$ along a program $f$.

Def.~\ref{def:weakest-precondition} has a couple of nice properties.
Firstly, \eqref{eq:wp-for-lax-slice} can capture various properties of programs with various computational effects as we will explain below.
Secondly, if $\answerobj$ has a certain order structure (i.e./ $\answerobj$ is an \emph{ordered object}) and $\emalgsymbol : T \answerobj \to \answerobj$ is monotone with respect to the order structure of $\answerobj$, then it is shown that~\eqref{eq:wp-for-lax-slice} does give the \emph{weakest} precondition with respect to the order structure of $\answerobj$.
Lastly, there is a bijective correspondence between EM monotone $T$-algebras and weakest precondition transformers that have a certain kind of compositionality (see~\cite[Cor.~4.5,4.6]{aguirre2020} for details).
In other words, EM monotone $T$-algebras give ``nice'' weakest precondition transformers, and conversely, ``nice'' weakest precondition transformers are given only by EM monotone $T$-algebras.

However, Def.~\ref{def:weakest-precondition} does not tell us much about how to syntactically compute weakest preconditions.
To automate weakest-precondition-based program verification, we usually want to compute a formula that represents the weakest precondition, but it is not obvious how to obtain such formulas from the syntax-free definition~\eqref{eq:wp-for-lax-slice}.
This is why we develop a syntactic counterpart of Def.~\ref{def:weakest-precondition} in Section~\ref{sec:wpt-cps}.

Most of the following examples are mild extensions of~\cite{aguirre2020} to domain theoretic models like $\omegaCPO$ and $\omegaQBS$, but Example~\ref{ex:trace-property-wp} is new.
\begin{example}[total correctness]\label{ex:total-correctness-wp}
	Let $\answerobj = (\{ \mathbf{false}, \mathbf{true} \}, {\le}) \in \omegaCPO$ where $\mathbf{false} \le \mathbf{true}$.
	The weakest precondition for total correctness is given by the EM algebra $\emalgsymbol_{\mathrm{tot}} : \answerobj_{\bot} \to \answerobj$ such that $\emalgsymbol_{\mathrm{tot}}(\bot) = \mathbf{false}$.
	That is, for each $f : X \to Y_{\bot}$ and $Q : Y \to \answerobj$, we have
	\[ \mathrm{wp}^{\emalgsymbol_{\mathrm{tot}}}[f](Q) \quad=\quad \{ x \in X \mid \exists y \in Y. f(x) = y \land Q(y) = \mathbf{true} \} \]
	by identifying a subset of $X$ with its characteristic function $X \to \Omega$.
\end{example}

\begin{example}[partial correctness]\label{ex:partial-correctness-wp}
	Consider $\answerobj^{\op} = (\{ \mathbf{false}, \mathbf{true} \}, {\ge}) \in \omegaCPO$, which has the opposite order structure of Example~\ref{ex:total-correctness-wp}.
	The weakest precondition for partial correctness is given by the EM algebra $\emalgsymbol_{\mathrm{par}} : (\answerobj^{\op})_{\bot} \to \answerobj^{\op}$ such that $\emalgsymbol_{\mathrm{par}}(\bot) = \mathbf{true}$ (here, the opposite order ${\ge}$ ensures that $\emalgsymbol_{\mathrm{par}}$ is Scott-continuous).
	For each $f : X \to Y_{\bot}$ and $Q : Y \to \answerobj^{\op}$, we have
	\[ \mathrm{wp}^{\emalgsymbol_{\mathrm{par}}}[f](Q) \quad=\quad \{ x \in X \mid f(x) = \bot \lor (\exists y \in Y. f(x) = y \land Q(y) = \mathbf{true}) \}. \]
\end{example}

\begin{remark}
	In Example~\ref{ex:partial-correctness-wp}, we used the opposite order ${\ge}$.
	We will later explain that this reversal of the order changes the interpretation of fixed points (Def.~\ref{def:interpretation-response-calculus}) in the target language of the CPS transformation, and thus, recursion in a program corresponds to the \emph{greatest} fixed points in the weakest precondition for partial correctness.
	This should \emph{not} be confused with the reversal of the implication order with respect to which the weakest precondition is the \emph{weakest}.
	Note also that we cannot use the trivial order $(\{ \mathbf{false}, \mathbf{true} \}, {=}) \in \omegaCPO$ to define an EM $({-})_{\bot}$-algebra $\emalgsymbol$ in Example~\ref{ex:total-correctness-wp},\ref{ex:partial-correctness-wp} because if $\emalgsymbol(\bot) = \mathbf{false}$, then by monotonicity of $\emalgsymbol$, we have $\mathbf{false} = \emalgsymbol(\bot) = \emalgsymbol(\mathbf{true}) = \mathbf{true}$, which is a contradiction.
	Similarly, $\emalgsymbol(\bot) = \mathbf{true}$ also leads to a contradiction.
\end{remark}

\begin{example}[safety property]\label{ex:trace-property-wp}
	We define an EM $T^P$-algebra $\emalgsymbol_{\mathrm{tr}} : T^P \answerobj \to \answerobj$ where $T^P$ is a monad defined in Example~\ref{ex:trace-property-semantics} and then explain that safety properties can be expressed as weakest preconditions for $\emalgsymbol_{\mathrm{tr}}$.
	Let $\mathfrak{A}$ be a deterministic finite automaton $(U, \delta, q_0, F)$ where $U$ is a finite set of states, $\delta \subseteq U \times E \times U$ is a transition relation, $q_0 \in U$ is an initial state, and $F$ is a set of final states.
	Here, we say $\mathfrak{A}$ is \emph{deterministic} if for any $q \in U$ and $a \in E$, there is at most one $q' \in U$ such that $(q, a, q') \in \delta$.
	We also assume that all states are final states $U = F$.
	We write $q \xrightarrow{a} q'$ if $(q, a, q') \in \delta$.
	The language accepted by $\mathfrak{A}$ is denoted by $L(\mathfrak{A})$.

	Now, consider an $\omega$cpo $\answerobj = (2^U, \supseteq)$ (note the \emph{opposite} inclusion order $\supseteq$).
	This means that each truth value $Q \in 2^U$ assigns true or false to each state of $\mathfrak{A}$.
	We define a $\mathcal{T}_P$-algebra on $\answerobj$ as follows.
	\begin{gather}
		\nontermconst^{\answerobj}\ \coloneqq\ U \qquad
		x \join^{\answerobj} y\ \coloneqq\ x \cap y \qquad
		\operation{event}^{\answerobj}_a(x)\ \coloneqq\ \langle a \rangle x\ \coloneqq\ \{ q \in U \mid \exists q' \in x, q \xrightarrow{a} q' \}
		\label{eq:trace-algebra}
	\end{gather}
	Note that operations defined in~\eqref{eq:trace-algebra} are Scott-continuous.
	Note also that $\operation{event}_a(x \join y) = \operation{event}_a(x) \join \operation{event}_a(y)$ holds because we assumed that $(U, \delta)$ is deterministic.
	This $\mathcal{T}_P$-algebra defines an EM $T^P$-algebra $\emalgsymbol_{\mathrm{tr}} : T^P \answerobj \to \answerobj$.

	The weakest precondition transformer defined by $\emalgsymbol_{\mathrm{tr}}$ corresponds to safety properties for the automaton $\mathfrak{A}$.
	For simplicity, consider a morphism $f : 1 \to T^P 1$ that represents a program whose input and output are the unit type.
	In this situation, we can regard $\mathrm{wp}^{\emalgsymbol_{\mathrm{tr}}}[f]$ as a function of type $\answerobj \to \answerobj$ by identifying $\answerobj \cong \omegaCPO(1, \answerobj)$.
	Relating $\mathrm{wp}^{\emalgsymbol_{\mathrm{tr}}}[f]$ with the safety property for $f$ is a bit tricky because of the trickiness of constructing a free $\mathcal{T}_P$-algebra $T^P X$ in Example~\ref{ex:trace-property-semantics}.
	Here, we outline how we manage it (see \referappendix{subsec:free-algebra-nondet}{G.1} for details).
	First, we define another algebraic theory $\mathcal{T}_H$ by adding $x \join y \ge x$ to $\mathcal{T}_P$, which corresponds to considering Hoare powerdomains instead of Plotkin powerdomains.
	Then, we can concretely construct a free $\mathcal{T}_H$-algebra $H 1$ generated by a terminal object $1$.
	Here, $H 1$ is given by $\{ Y \subseteq E^{*} \times (1 + \{ \bot \}) \mid \text{$Y$ satisfies a certain condition} \}$ where $E^{*}$ is the set of finite sequences of elements in $E$.
	Since a $\mathcal{T}_H$-algebra is a $\mathcal{T}_P$-algebra, we have a unique morphism $h^H : T^P 1 \to H 1$ by the freeness of $T^P 1$.
	Then, we can show
	\[ \mathrm{wp}^{\emalgsymbol_{\mathrm{tr}}}[f](U) \quad=\quad \bigcap_{(s, x) \in h^H(f(\star))} \langle s \rangle U \]
	where $\langle a_1 \dots a_n \rangle x \coloneqq \langle a_1 \rangle \dots \langle a_n \rangle x$ is a shorthand notation for a sequence of events and $\star \in 1$; and therefore
	\begin{equation}
		q_0 \in \mathrm{wp}^{\emalgsymbol_{\mathrm{tr}}}[f](U) \quad\iff\quad \forall (s, x) \in h^H(f(\star)),\ s \in L(\mathfrak{A}).
		\label{eq:wp-trace-iff}
	\end{equation}
	Since $\{ s \mid (s, x) \in h^H(f(\star)) \}$ is (the prefix closure of) the set of sequences of events output by $f$, we can rephrase \eqref{eq:wp-trace-iff} as ``the safety property for $f : 1 \to T^P 1$ is true if and only if $\mathrm{wp}^{\emalgsymbol_{\mathrm{tr}}}[f](U)$ is true at the initial state $q_0$ of the given automaton $\mathfrak{A}$''.
	Intuitively, $\mathrm{wp}^{\emalgsymbol_{\mathrm{tr}}}[f]$ takes a set of ``post-states'', runs the given automaton $\mathfrak{A}$ backwards, and returns the set of ``pre-states'' such that for any pre-state and any output string, there exists a run of $\mathfrak{A}$ that finishes at a post-state.
	Note that characterising the safety property by $q_0 \in \mathrm{wp}^{\emalgsymbol_{\mathrm{tr}}}[f](U)$ is a novel result to the best of our knowledge.
\end{example}

For expected cost analyses and cost moment analyses, we defined a composite monad using distributive laws (Example~\ref{ex:expected-cost-cost-moment-semantics}).
The following lemma~\cite[Section~2]{beck1969}\cite[Theorem~2.4.3]{manes2007} is useful when we define EM algebras for such a composite monad.
\begin{lemma}\label{lem:distributive-law-em-algebra}
	Let $d : T S \to S T$ be a distributive law.
	There is a bijection between (i) EM $ST$-algebras $\emalgsymbol^{ST} : S T \answerobj \to \answerobj$ and (ii) pairs of an EM $S$-algebra $\emalgsymbol^{S} : S \answerobj \to \answerobj$ and an EM $T$-algebra $\emalgsymbol^{T} : T \answerobj \to \answerobj$ that satisfy the \emph{composite law} $\emalgsymbol^{T} \comp T \emalgsymbol^{S} = \emalgsymbol^{S} \comp S \emalgsymbol^{T} \comp d$.
	\qed
\end{lemma}

\begin{example}[expected cost analysis]\label{ex:expected-cost-wp}
	Let $\answerobj = \mathbb{W}$ be the $\omega$qbs of real weights.
	We define an EM $P(\mathbb{W} \times ({-})_{\bot})$-algebra on $\answerobj$ as the composite of three EM algebras for $({-})_{\bot}$, $\mathbb{W} \times ({-})$, and $P$.
	(1) The EM $({-})_{\bot}$-algebra $\emalgsymbol^{({-})_{\bot}} : \mathbb{W}_{\bot} \to \mathbb{W}$ maps the bottom element to $0$.
	(2) An EM $(\mathbb{W} \times {-})$-algebra is defined by the addition $({+}) : \mathbb{W} \times \mathbb{W} \to \mathbb{W}$.
	(3) An EM $P$-algebra $\emalgsymbol^{P} : P \mathbb{W} \to \mathbb{W}$ is defined by the expectation of probability distributions.

	Then, these EM algebras satisfy the composite law of Lemma~\ref{lem:distributive-law-em-algebra}, and we obtain an EM $P(\mathbb{W} \times ({-})_{\bot})$-algebra by $\emalgsymbol_{\mathrm{ex}} \coloneqq \emalgsymbol^P \comp P({+}) \comp P(\mathbb{W} \times \emalgsymbol^{({-})_{\bot}})$.
	If the postcondition is the constant function 0, then the weakest precondition for $\emalgsymbol_{\mathrm{ex}}$ is given as follows.
	\[ \mathrm{wp}^{\emalgsymbol_{\mathrm{ex}}}[f](0) \quad=\quad \emalgsymbol^{P} \comp P \pi_1 \comp f \]
	Here, $P \pi_1 \comp f : X \to P \mathbb{W}$ corresponds to a probability distribution of the cost of $f : X \to P(\mathbb{W} \times Y_{\bot})$.
	Therefore, $\mathrm{wp}^{\emalgsymbol_{\mathrm{ex}}}[f](0)$ is the expected cost of $f$.
\end{example}

\begin{example}[cost moment analysis]\label{ex:cost-moment-wp}
	We consider the $\omega$qbs $\answerobj = \mathbb{W}^n$ of $n$-tuples of real weights and define an EM algebra structure by the following combination of EM algebras.
	\begin{itemize}
		\item The EM $({-})_{\bot}$-algebra $\emalgsymbol_n^{({-})_{\bot}} : (\mathbb{W}^n)_{\bot} \to \mathbb{W}^n$ maps the bottom element to $(0, \dots, 0)$.
		\item We define a $\mathbb{W}$-module ${\oplus} : \mathbb{W} \times \mathbb{W}^n \to \mathbb{W}^n$ such that the $i$-th component of $a \oplus (b_1, \dots, b_n)$ is $a^i + \sum_{j = 1}^i \binom{i}{j}  a^{i-j} b_j$.
			This defines an EM $(\mathbb{W} \times {-})$-algebra.
			Note that $\oplus$ is called the \emph{elapse function} in~\cite{kura2019} and essential for the extension to higher moments.
		\item We have an EM $P$-algebra $\emalgsymbol_n^{P} : P (\mathbb{W}^n) \to \mathbb{W}^n$ as the $n$-fold product of $\emalgsymbol^{P} : P \mathbb{W} \to \mathbb{W}$.
	\end{itemize}
	By Lemma~\ref{lem:distributive-law-em-algebra}, we have an EM $P(\mathbb{W} \times ({-})_{\bot})$-algebra $\emalgsymbol_{\mathrm{mo}, n} \coloneqq \emalgsymbol_n^P \comp P({\oplus}) \comp P(\mathbb{W} \times \emalgsymbol_n^{({-})_{\bot}})$.
	The weakest precondition for the constant postcondition $\mathbf{0} = (0, \dots, 0)$ is given by
	\[ \mathrm{wp}^{\emalgsymbol_{\mathrm{mo}, n}}[f](\mathbf{0}) \quad=\quad \emalgsymbol_n^P \comp P \mathrm{pow}_n \comp P \pi_1 \comp f \]
	where $\mathrm{pow}_n : \mathbb{W} \to \mathbb{W}^n$ is defined by $\mathrm{pow}_n(x) = (x, x^2, \dots, x^n)$.
	That is, the $i$-th component of the weakest precondition is the $i$-th moment of the distribution $P \pi_1 \comp f$ of cost.
\end{example}

\begin{example}[conditional weakest preexpectation]\label{ex:conditional-wp-wp}
	Recall the monad defined in Example~\ref{ex:conditional-wp-semantics}.
	Let $\answerobj_1 \coloneqq \mathbb{W}$, $\answerobj_2 = [0, 1]^{\op}$, and $\answerobj \coloneqq \answerobj_1 \times \answerobj_2$ where $[0, 1]^{\op}$ has the opposite order of $[0, 1]$.
	We aim to define an EM $P S$-algebra $\emalgsymbol_{\mathrm{cwp}} : P S \answerobj \to \answerobj$ as the product of two EM algebras $\emalgsymbol_{\mathrm{cwp}, 1} : P S \answerobj_1 \to \answerobj_1$ and $\emalgsymbol_{\mathrm{cwp}, 2} : P S \answerobj_2 \to \answerobj_2$.
	Since the expectation defines EM $P$-algebras $\emalgsymbol^{P} : P \mathbb{W} \to \mathbb{W}$ and $\emalgsymbol^{P}_{[0, 1]^{\op}} : P [0, 1]^{\op} \to [0, 1]^{\op}$, it remains to define EM $S$-algebras on $\mathbb{W}$ and $[0, 1]^{\op}$.
	Let $\emalgsymbol^S_1 : S \mathbb{W} \to \mathbb{W}$ be a morphism defined by $\emalgsymbol^S_1(\operation{score}_r(\bot)) = 0$ and $\emalgsymbol^S_1(\operation{score}_r(x)) = r \cdot x$ for any $x \in \mathbb{W}$; and $\emalgsymbol^S_2 : S [0, 1]^{\op} \to [0, 1]^{\op}$ be $\emalgsymbol^S_2(\operation{score}_r(\bot)) = r$ and $\emalgsymbol^S_2(\operation{score}_r(x)) = r \cdot x$ for any $x \in [0, 1]$.
	By Lemma~\ref{lem:distributive-law-em-algebra}, we obtain EM $P S$-algebras $\emalgsymbol_{\mathrm{cwp}, 1} = \emalgsymbol^{P} \comp P \emalgsymbol^S_1$ and $\emalgsymbol_{\mathrm{cwp}, 2} = \emalgsymbol^{P}_{[0, 1]^{\op}} \comp P \emalgsymbol^S_2$.
	The weakest precondition defined by $\emalgsymbol_{\mathrm{cwp}}$ is given as follows.
	\begin{equation}
		\mathrm{wp}^{\emalgsymbol_{\mathrm{cwp}}}[f](\tupling{Q_1}{Q_2}) \quad=\quad \tupling{\mathrm{wp}^{\emalgsymbol_{\mathrm{cwp}, 1}}[f](Q_1)}{\mathrm{wp}^{\emalgsymbol_{\mathrm{cwp}, 2}}[f](Q_2)} \label{eq:conditional-wp}
	\end{equation}
	Let $l : S X \to [0, 1]$ be the ``likelihood'' function defined by $l(\operation{score}_r(\bot)) = r$ and $l(\operation{score}_r(x)) = r$; and $v : S X \to X + \{ \bot \}$ be the ``result value'' defined by $v(\operation{score}_r(\bot)) = \bot$ and $v(\operation{score}_r(x)) = x$.
	The first component of~\eqref{eq:conditional-wp} is the integral of $Q_1$ re-weighted by the likelihood.
	\[ \mathrm{wp}^{\emalgsymbol_{\mathrm{cwp}, 1}}[f](Q_1)(x) \quad=\quad \int \emalgsymbol^S_1 \comp S Q_1 \,\mathrm{d} (f(x)) \quad=\quad \int_{\{ y \in S X \mid v(y) \neq \bot\}} l(y) \cdot Q_2(v(y)) \,\mathrm{d} (f(x))(y) \]
	If $Q_2$ is the constant function $1$, then the second component of~\eqref{eq:conditional-wp} is the integral of the likelihood.
	\[ \mathrm{wp}^{\emalgsymbol_{\mathrm{cwp}, 2}}[f](Q_2)(x) \quad=\quad \int \emalgsymbol^S_2 \comp S Q_2 \,\mathrm{d} (f(x)) \quad=\quad \int l \,\mathrm{d} (f(x)) \]
	Finally, the conditional weakest preexpectation is defined by  
	\[ \mathrm{cwp}[f](Q) \quad\coloneqq\quad \frac{\mathrm{wp}^{\emalgsymbol_{\mathrm{cwp}, 1}}[f](Q)}{\mathrm{wp}^{\emalgsymbol_{\mathrm{cwp}, 2}}[f](1)} \quad=\quad \int_{\{ y \in S X \mid v(y) \neq \bot\}} \frac{l(y)}{\int l \,\mathrm{d} (f(x))} \cdot Q_2(v(y)) \,\mathrm{d} (f(x))(y). \]
	This captures the conditional weakest preexpectation studied in \cite{olmedo2018}.
\end{example}

\section{CPS transformation}
We define a CPS transformation for the $\lambda_c$-calculus.
Our CPS transformation makes a clear distinction between the source language (the $\lambda_c$-calculus in Section~\ref{sec:source-language}) and the target language (Section~\ref{subsec:target-language}).
Since the syntax and the semantics of the source language is parameterised by a $\lambda_c$-signature $\Sigma$ and a $\lambda_c(\Sigma)$-structure $\mathcal{A}$, respectively, the target language is also parameterised by the same $\Sigma$ and $\mathcal{A}$.

\subsection{Target Language}\label{subsec:target-language}
Our target language is defined as a simply typed lambda calculus with a designated type $\answertype$ for truth values and (least) fixed points with respect to a given order relation on $\answertype$.
The main difference between the source and the target language is that the target language is \emph{pure} while the source language is \emph{effectful}.
This pureness makes it easier to reason about the target language.

\subsubsection{Syntax}
We define a target language $\lambda_{\mathrm{HFL}}$ of our CPS transformation based on the \emph{response calculus}~\cite{fuhrmann2004} and extend it with modal operators.
Our target language has an \emph{answer type} $\answertype$ as a type of results of continuations.
We interpret $\answertype$ as a type of truth values.
When $\answertype$ is interpreted as $\{ \mathbf{true}, \mathbf{false} \}$, our target language can be understood as a variant higher-order modal fixed-point logic~\cite{viswanathan2004}, but in general, we don't necessarily interpret $\answertype$ as $\{ \mathbf{true}, \mathbf{false} \}$.
This generality enables the target language $\lambda_{\mathrm{HFL}}$ to express various weakest precondition transformers.

\begin{definition}\label{def:target-language}
	Given a $\lambda_c$-signature $\Sigma = (B, K, O, \mathrm{ar}, \mathrm{car})$, we define \emph{$\lambda_{\mathrm{HFL}}(\Sigma)$-types} $\rho, \tau$ and \emph{$\lambda_{\mathrm{HFL}}(\Sigma)$-terms} $M, N$ as follows.
	\begin{align}
		\rho, \tau & \quad\coloneqq\quad \answertype \mid b \mid 1 \mid \rho \times \tau \mid 0 \mid \rho + \tau \mid \rho \to \answertype &&\text{where $b \in B$} \\
		M, N    & \quad\coloneqq\quad x \mid c\ M \mid o\ M \mid () \mid (M, N) \mid \pi_i M \mid \lambda x : \rho. M \mid M\ N \mid \iota_i M \\
    	&\qquad\quad \mid \delta(M) \mid \delta(M, x_1 {:} \rho_1. N_1, x_2 {:} \rho_2. N_2) \mid \mu f {:} \rho \to \answertype. M &&\text{where $c \in K$ and $o \in O$}
	\end{align}
	We call ``$o\ M$'' and ``$\mu f. M$'' in $\lambda_{\mathrm{HFL}}(\Sigma)$-terms (an application of) a \emph{modal operator} and a \emph{fixed point}, respectively.
	Similarly to the source language $\lambda_c$, we define $\letrec{f}{x}{M}{N}$ as syntactic sugar using $\mu f. M$.
	If $o \in O$ is an $n$-ary operation $o : \underline{n} \rightarrowtriangle \underline{1}$, we define a notation $o(M_1, \dots, M_n)$ by
	\begin{equation}
		o\ (M_1, \dots, M_n) \quad\coloneqq\quad o\ (\lambda x. \delta(x, x_1. M_1, \dots, x_n. M_n), ()) \label{eq:n-ary-algebraic-operation-notation}
	  \end{equation}
	where $\delta(x, x_1. M_1, \dots, x_n. M_n)$ is a notation for case analyses for finite coproduct types derived from $\delta({-})$ and $\delta({-}, x_1. M_1, x_2. M_2)$.

	A \emph{well-typed $\lambda_{\mathrm{HFL}}(\Sigma)$-term} $\Gamma \vdash M : \rho$ is defined like the usual simply typed lambda calculus, but some parts of the typing rules are restricted to the answer type $\answertype$ (Fig.~\ref{fig:typing-response-calculus}).
	Notably, the codomain of the function type is restricted to $\answertype$.
	The full definition of typing rules can be found in \referappendix{subsec:target-typing-rules}{B.1}.
\end{definition}

In the situation considered in Example~\ref{ex:trace-property-syntax}, the modal operator for $\operation{event}_a$ corresponds to the modal operator $\langle a \rangle$ used in the conventional HFL in~\cite{viswanathan2004}, as we will later explain in Example~\ref{ex:trace-property-cps}.
Compared to higher-order modal fixed-point logic~\cite{viswanathan2004}, modal operators $o\ M$ in our target language are generalized to arbitrary $o \in O$.
For example, for expected cost analysis, we will consider in Example~\ref{ex:expected-cost-cps} the modal operator for integration, which corresponds to the algebraic operation $\operation{unif}$ for sampling.
On the other hand, logical connectives (e.g.\ conjunction, disjunction, and negation) are not included in Definition~\ref{def:target-language} to keep the language minimal.
We can add them when the interpretation of the answer type has a sufficient structure (e.g.\ an internal lattice structure) to interpret them (see Section~\ref{subsec:extending-target-language}).

\begin{figure}[tbp]
	\small
	\begin{mathpar}
		\inferrule{
			\Gamma \vdash M : (\mathrm{ar}(o) \to \answertype) \times \mathrm{car}(o)
		}{
			\Gamma \vdash o\ M : \answertype
		}
		\and
		\inferrule{
			\Gamma, f : \rho \to \answertype \vdash M : \rho \to \answertype
		}{
			\Gamma \vdash \mu f. M : \rho \to \answertype
		}
	\end{mathpar}
	\caption{Selected typing rules for the target language.}
	\label{fig:typing-response-calculus}
\end{figure}

\subsubsection{Semantics}
To interpret the target language, we use an $\omegaCPO$-enriched $\lambda_c(\Sigma)$-structure $\mathcal{A} = (\category{C}, T, A, a)$ together with an EM $T$-algebra $\emalgsymbol : T \answerobj \to \answerobj$.
Most of the interpretation are done by the $\lambda_c(\Sigma)$-structure $\mathcal{A}$ without using $T$, just like the standard interpretation of pure STLC, but there are a few exceptions: the answer type $\answertype$, modal operators $o \in O$, and fixed points.
This is where we need the EM $T$-algebra $\emalgsymbol : T \answerobj \to \answerobj$ for interpretation.
\begin{definition}[interpretation of $\lambda_{\mathrm{HFL}}(\Sigma)$-types/terms]\label{def:interpretation-response-calculus}
	We define the interpretation $\mathcal{A}^{\emalgsymbol} \llbracket {-} \rrbracket$ of $\lambda_{\mathrm{HFL}}(\Sigma)$-types/terms as follows.
	For base types and the answer type, we define $\mathcal{A}^{\emalgsymbol} \interpret{b} \coloneqq A b$ and $\mathcal{A}^{\emalgsymbol} \llbracket \answertype \rrbracket \coloneqq \answerobj$, and for other types, we extend $\mathcal{A}^{\emalgsymbol} \interpret{-}$ using the bicartesian closed structure of $\category{C}$ (note $\mathcal{A}^{\emalgsymbol} \interpret{\rho \to \tau} \coloneqq \exponential{\mathcal{A}^{\emalgsymbol} \interpret{\rho}}{\mathcal{A}^{\emalgsymbol} \interpret{\tau}}$).
	For contexts, we define $\mathcal{A}^{\emalgsymbol} \interpret{x_1 : \rho_1, \dots, x_n : \rho_n} \coloneqq \mathcal{A}^{\emalgsymbol} \interpret{\rho_1} \times \dots \times \mathcal{A}^{\emalgsymbol} \interpret{\rho_n}$.
	For any well-typed term $\Gamma \vdash M : \rho$, the interpretation $\mathcal{A}^{\emalgsymbol} \interpret{M} : \mathcal{A}^{\emalgsymbol} \interpret{\Gamma} \to \mathcal{A}^{\emalgsymbol} \interpret{\rho}$ is defined by the standard interpretation of simply typed lambda calculus together with the following interpretation of effect-free constants, modal operators, and fixed points.
	\begin{gather}
		\mathcal{A}^{\emalgsymbol} \interpret{c\ M} \ \coloneqq\ a(c) \comp \mathcal{A}^{\emalgsymbol} \interpret{M} \qquad\qquad
    	\mathcal{A}^{\emalgsymbol} \interpret{\mu f. M} \ \coloneqq\ (\mathcal{A}^{\emalgsymbol} \interpret{M})^{\dagger} \\
		\mathcal{A}^{\emalgsymbol} \interpret{o\ M} \ \coloneqq \ \emalgsymbol \comp T \mathbf{ev} \comp \strength^T \comp (\identity{} \times a(o)) \comp \mathcal{A}^{\emalgsymbol} \interpret{M}
	\end{gather}
	Here, $\mathbf{ev}_{X, Y} : (\exponential{X}{Y}) \times X \to Y$ is the evaluation morphism.
	In the interpretation of fixed points, the EM algebra structure on $\exponential{\mathcal{A}^{\emalgsymbol} \interpret{\rho}}{\answerobj}$ allows us to use the uniform fixed-point operator $({-})^{\dagger}$.
	The full definition of $\mathcal{A}^{\emalgsymbol} \interpret{-}$ can be found in \referappendix{subsec:target-semantics}{B.2}.
\end{definition}
Note that different choices of EM algebras give different interpretations of the target language (especially, of fixed points and modal operators).
For example, a fixed point in a $\lambda_{\mathrm{HFL}}(\Sigma)$-term is interpreted as the \emph{least} fixed point with respect to the order structure of $\answerobj$.
If we use the opposite of the ``standard'' order of $\answerobj$, then a fixed point is interpreted as the \emph{greatest} fixed point with respect to the ``standard'' order.
Such examples can be found in Example~\ref{ex:total-partial-correctness-cps} and other examples in Section~\ref{sec:wpt-cps}.
Note also that an EM $T$-algebra $\emalgsymbol : T \answerobj \to \answerobj$ is also used in Section~\ref{sec:wpt} to define a weakest precondition transformer.
Using the same EM $T$-algebra $\emalgsymbol$ for defining weakest preconditions and for interpreting the target language is important in our main theorem (Theorem~\ref{thm:cps-is-wpt-with-recursion}).

\subsection{CPS}\label{subsec:cps-transformation}
Following~\cite{fuhrmann2004}, we define a CPS transformation.
\begin{definition}[CPS transformation]\label{def:cps-transformation}
	Let $\Sigma$ be a $\lambda_c$-signature.
	A CPS transformation $\CPS{({-})}$ is defined as a mapping from $\lambda_c(\Sigma)$-types/terms to $\lambda_{\mathrm{HFL}}(\Sigma)$-types/terms in Fig.~\ref{fig:cps-transformation}.
\end{definition}

\begin{lemma}
	If $\Gamma \vdash M : \rho$ is a well-typed $\lambda_c(\Sigma)$-term, then $\CPS{\Gamma} \vdash \CPS{M} : (\CPS{\rho} \to \answertype) \to \answertype$ is a well-typed $\lambda_{\mathrm{HFL}}(\Sigma)$-term.
	\qed
\end{lemma}

Note that Def.~\ref{def:cps-transformation} produces many administrative redexes.
Efficient implementation of the CPS transformation is orthogonal to our main problem (i.e.\ the soundness of syntactic computation of weakest preconditions) and out of the scope of this paper.

\begin{figure}[tbp]
	\textbf{Types and Contexts}
	\begin{gather}
		\CPS{b} \coloneqq b \qquad\quad
		\CPS{1} \coloneqq 1 \qquad\quad
		\CPS{0} \coloneqq 0 \qquad\quad
		\CPS{(\rho_1 \times \rho_2)} \coloneqq \CPS{\rho_1} \times \CPS{\rho_2} \qquad\quad
		\CPS{(\rho_1 + \rho_2)} \coloneqq \CPS{\rho_1} + \CPS{\rho_2} \\
		\CPS{(\rho \to \tau)} \coloneqq (\CPS{\rho} \times (\CPS{\tau} \to \answertype)) \to \answertype \qquad\qquad
		\CPS{(x_1 : \rho_1, \dots, x_n : \rho_n)} \coloneqq x_1 : \CPS{\rho_1}, \dots, x_n : \CPS{\rho_n}
	\end{gather}
	\textbf{Terms}
	\begin{alignat}{2}
		&\CPS{x} &&\quad\coloneqq\quad \lambda k.\ k\ x \\
		&\CPS{(\geneff{o}\ M)} &&\quad\coloneqq\quad \lambda k.\ \CPS{M}\ (\lambda m. o\ (k, m)) \\
		&\CPS{(c\ M)} &&\quad\coloneqq\quad \lambda k.\ \CPS{M}\ (\lambda m. k\ (c\ m)) \\
		&\CPS{(\pi_i M)} &&\quad\coloneqq\quad \lambda k.\ \CPS{M}\ (\lambda m. k\ (\pi_i m)) \\
		&\CPS{()} &&\quad\coloneqq\quad \lambda k.\ k\ () \\
		&\CPS{(M_1, M_2)} &&\quad\coloneqq\quad \lambda k.\ \CPS{M_1}\ (\lambda m_1. \CPS{M_2}\ (\lambda m_2. k\ (m_1, m_2))) \\
		&\CPS{(\delta(M))} &&\quad\coloneqq\quad \lambda k.\ \CPS{M} (\lambda m. \delta(m)) \\
		&\CPS{(\iota_i M)} &&\quad\coloneqq\quad \lambda k.\ \CPS{M}\ (\lambda m. k\ (\iota_i m)) \\
		&\CPS{(\delta(M, x_1. M_1, x_2. M_2))} &&\quad\coloneqq\quad \lambda k.\ \CPS{M} (\lambda m. \delta(m, x_1. \CPS{M_1}\ k, x_2. \CPS{M_2}\ k)) \\
		&\CPS{(\lambda x. M)} &&\quad\coloneqq\quad \lambda k.\ k\ (\lambda (x, h). \CPS{M}\ h) \\
		&\CPS{(M\ N)} &&\quad\coloneqq\quad \lambda k.\ \CPS{M}\ (\lambda m. \CPS{N}\ (\lambda n. m\ (n, k))) \\
		&\CPS{(\mu f. M)} &&\quad\coloneqq\quad \lambda k.\ k (\mu f. \lambda x. \CPS{M}\ (\lambda m. m\ x)) \\
		&\CPS{(\letrec{f}{x}{M}{N})} &&\quad\coloneqq\quad \letrec{f}{(x, k)}{\CPS{M}\ k}{\CPS{N}}
	\end{alignat}
	\caption{The CPS transformation. We implicitly assume that newly introduced variables are fresh. We use $\lambda (x, y). M \coloneqq \lambda z. M[\pi_1 z/x, \pi_2 z/y]$ as syntactic sugar where $M[N/x]$ is the substitution defined as usual.}
	\label{fig:cps-transformation}
\end{figure}

\section{Relating Weakest Preconditions and the CPS Transformation}\label{sec:wpt-cps}
\subsection{Main Result}

\begin{definition}[stable binary coproducts]\label{def:stable-finite-coproducts}
	A binary coproduct $A_1 \xrightarrow{\iota_1} A \xleftarrow{\iota_2} A_2$ is \emph{stable} if for any morphism $f : X \to A$, there exists a pullback $\iota_i^{*} X$ of $X$ along $\iota_i$ ($i = 1, 2$) such that $\iota_1^{*} X \to X \leftarrow \iota_2^{*} X$ is a coproduct diagram.
	A bicartesian closed category is \emph{stable} if it has stable binary coproducts.
\end{definition}

\begin{theorem}\label{thm:cps-is-wpt-with-recursion}
	Suppose that the following parameters are given: (i) a $\lambda_c$-signature $\Sigma$ (Def.~\ref{def:lambda-c-signature}) (ii) an $\omegaCPO$-enriched $\lambda_c(\Sigma)$-structure $\mathcal{A} = (\category{C}, T, A, a)$ (Def.~\ref{def:lambda-c-structure}) such that $\category{C}$ is stable (Def.~\ref{def:stable-finite-coproducts}) (iii) an Eilenberg--Moore $T$-algebra $\emalgsymbol : T \answerobj \to \answerobj$.
	For any well-typed $\lambda_c(\Sigma)$-term $x_1 {:} \rho_1, \dots, x_n {:} \rho_n \vdash M : \rho$ (Def.~\ref{def:source-term}), if $\rho_1, \dots, \rho_n, \rho$ are ground types (i.e.\ types without $\to$), then for any postcondition $x : \rho \vdash P : \answertype$ in the target language (i.e.\ $P$ is a $\lambda_{\mathrm{HFL}}(\Sigma)$-term, Def.~\ref{def:target-language}), we have
	\begin{equation}
		\mathrm{wp}^{\emalgsymbol}[\mathcal{A} \interpret{M}] (\mathcal{A}^{\emalgsymbol} \interpret{P}) \quad=\quad \mathcal{A}^{\emalgsymbol} \interpret{\CPS{M}\ (\lambda x : \rho. P)} \label{eq:cps_is_wpt}
	\end{equation}
	where $\mathrm{wp}^{\emalgsymbol}[{-}]$ is defined by \eqref{eq:wp-for-lax-slice}, $\CPS{({-})}$ is the CPS transformation (Def.~\ref{def:cps-transformation}), and $\mathcal{A} \interpret{-}$ and $\mathcal{A}^{\emalgsymbol} \interpret{-}$ are the interpretations of $\lambda_c(\Sigma)$-terms (Def.~\ref{def:interpretation-lambda-c-calculus}) and $\lambda_{\mathrm{HFL}}(\Sigma)$-terms (Def.~\ref{def:interpretation-response-calculus}), respectively.
\end{theorem}
\begin{proof}
	See \referappendix{sec:main-theorem-proof}{E}.
\end{proof}
Theorem~\ref{thm:cps-is-wpt-with-recursion} has three parameters.
A $\lambda_c$-signature is a parameter for syntax, a $\lambda_c(\Sigma)$-structure is for semantics, and an EM algebra is for the weakest precondition and the semantics of the target language.
More precisely, the last parameter, an EM algebra $\emalgsymbol : T \answerobj \to \answerobj$, defines (i) truth values $\mathcal{A}^{\emalgsymbol} \interpret{\answertype} = \answerobj$, (ii) meaning of the weakest precondition transformer $\mathrm{wp}^{\emalgsymbol}[{-}]$, (iii) the interpretation $\mathcal{A}^{\emalgsymbol} \interpret{o\ M}$ of modal operators in $\lambda_{\mathrm{HFL}}(\Sigma)$-terms, and (iv) the interpretation $\mathcal{A}^{\emalgsymbol} \interpret{\mathbf{let}\ \mathbf{rec}\ \dots}$ of fixed points in $\lambda_{\mathrm{HFL}}(\Sigma)$-terms.

Theorem~\ref{thm:cps-is-wpt-with-recursion} has a few assumptions.
Firstly, $\category{C}$ must be stable, which is not very restrictive since all the categories that we used in examples ($\omegaCPO$ and $\omegaQBS$) are stable.
Secondly, Theorem~\ref{thm:cps-is-wpt-with-recursion} does not allow function types in the context $\Gamma$ and the type $\rho$ of a well-typed $\lambda_c(\Sigma)$-term $\Gamma \vdash M : \rho$.
This restriction ensures that the type of the left-hand side and the right-hand side of \eqref{eq:cps_is_wpt} are the same because we have $\CPS{\rho} = \rho$ for any ground type $\rho \in \mathbf{GTyp}(B)$.
Compared with the studies~\cite{kobayashi2018,avanzini2021} that deal with special cases of Theorem~\ref{thm:cps-is-wpt-with-recursion}, our assumption on contexts and types is less restrictive than theirs because their results are limited to closed terms of type $1$ or $b$ (a base type).
Note that we can still use higher-order recursive or non-recursive functions in a program $M$ internally as long as function types don't appear at the top level.

In Theorem~\ref{thm:cps-is-wpt-with-recursion}, the postcondition $x : \rho \vdash P : \answertype$ does not refer to variables $x_1, \dots, x_n$ in the context of $x_1 : \rho_1, \dots, x_n : \rho_n \vdash M : \rho$.
This restriction is not essential because we can allow $P$ to refer to those variables as follows.
\begin{corollary}
	Suppose that we have the same parameters as Theorem~\ref{thm:cps-is-wpt-with-recursion}.
	For any well-typed $\lambda_c(\Sigma)$-term $x_1 : \rho_1, \dots, x_n : \rho_n \vdash M : \rho$, if $\rho_1, \dots, \rho_n, \rho$ are ground types, then for any postcondition $x_1 : \rho_1, \dots, x_n : \rho_n, x : \rho \vdash P : \answertype$, we have
	$\mathrm{wp}^{\emalgsymbol}[\tupling{\identity{}}{\mathcal{A} \interpret{M}}] (\mathcal{A}^{\emalgsymbol} \interpret{P}) = \mathcal{A}^{\emalgsymbol} \interpret{\CPS{M}\ (\lambda x : \rho. P)}$.
\end{corollary}
\begin{proof}
	Let $\rho' = \rho_1 \times \dots \times \rho_n \times \rho$.
	Apply Theorem~\ref{thm:cps-is-wpt-with-recursion} to a $\lambda_c(\Sigma)$-term $x_1 : \rho_1, \dots, x_n : \rho_n \vdash (x_1, \dots, x_n, M) : \rho'$ and a postcondition \\$y : \rho' \vdash P[\pi_1\ y/x_1, \dots, \pi_n\ y/x_n, \pi_{n+1}\ y/x] : \answertype$.
\end{proof}

By applying Theorem~\ref{thm:cps-is-wpt-with-recursion} to EM algebras in Section~\ref{sec:wpt} (Example~\ref{ex:total-correctness-wp},\ref{ex:partial-correctness-wp},\ref{ex:trace-property-wp},\ref{ex:expected-cost-wp},\ref{ex:cost-moment-wp}), we obtain all examples in Section~\ref{sec:examples} as instances of Theorem~\ref{thm:cps-is-wpt-with-recursion}.
We explain below how modal operators and fixed points in CPS-transformed programs are interpreted in each instance based on our general framework, which turns out to be essentially the same as what happened in Section~\ref{sec:examples}.

\begin{example}[total/partial correctness]\label{ex:total-partial-correctness-cps}
	The CPS transformation in Section~\ref{subsec:example-total-partial} does give weakest preconditions.
	We apply Theorem~\ref{thm:cps-is-wpt-with-recursion} to the EM algebras in Example~\ref{ex:total-correctness-wp},\ref{ex:partial-correctness-wp}.
	For total correctness, fixed points in the target language are interpreted as the least fixed points with respect to (the pointwise extension of) the order $\mathbf{false} \le \mathbf{true}$.
	For partial correctness, they are interpreted as greatest fixed points because we use $\answerobj^{\op} = (\{ \mathbf{false}, \mathbf{true} \}, {\ge})$ instead of $\answerobj = (\{ \mathbf{false}, \mathbf{true} \}, {\le})$. Here, recall that $\mathcal{A}$ is defined in Example~\ref{ex:total-partial-correctness-semantics}
	\begin{align}
		\mathcal{A}^{\emalgsymbol_{\mathrm{tot}}} \interpret{\mu f {:} \rho \to \answertype. M}(\gamma) &= \mathrm{lfp}^{\exponential{\mathcal{A}^{\emalgsymbol_{\mathrm{tot}}}\interpret{\rho}}{\answerobj}}(\mathcal{A}^{\emalgsymbol_{\mathrm{tot}}} \interpret{M}(\gamma, {-})) \\
		\mathcal{A}^{\emalgsymbol_{\mathrm{par}}} \interpret{\mu f {:} \rho \to \answertype. M}(\gamma) &= \mathrm{lfp}^{\exponential{\mathcal{A}^{\emalgsymbol_{\mathrm{par}}}\interpret{\rho}}{\answerobj^{\op}}}(\mathcal{A}^{\emalgsymbol_{\mathrm{par}}} \interpret{M}(\gamma, {-})) = \mathrm{gfp}^{\exponential{\mathcal{A}^{\emalgsymbol_{\mathrm{par}}}\interpret{\rho}}{\answerobj}}(\mathcal{A}^{\emalgsymbol_{\mathrm{par}}} \interpret{M}(\gamma, {-}))
	\end{align}
\end{example}

\begin{example}[safety property]\label{ex:trace-property-cps}
	Let $\vdash M : 1$ be a (closed) $\lambda_c(\Sigma)$-term.
	By Example~\ref{ex:trace-property-wp}, the safety property for $M$ holds if and only if $q_0 \in \mathrm{wp}[\mathcal{A} \interpret{M}](U)$ where $\mathcal{A}$ is defined in Example~\ref{ex:trace-property-semantics}.
	By Theorem~\ref{thm:cps-is-wpt-with-recursion}, we can reduce the problem of verifying a safety property to the problem of model checking the formula $q_0 \in \mathcal{A}^{\emalgsymbol_{\mathrm{tr}}} \interpret{\CPS{M} (\lambda \_. \mathbf{true})}$ via the CPS transformation.
	Here, we slightly extend the target language with $\mathcal{A}^{\emalgsymbol_{\mathrm{tr}}} \interpret{\mathbf{true}} = U$ (see Section~\ref{subsec:extending-target-language} for details).
	This gives the same translation as~\cite[Thm~3]{kobayashi2018}.

	The generic effect for $\operation{event}_a$ in a $\lambda_c(\Sigma)$-term is CPS-transformed to the corresponding modal operator
	\[ \CPS{(\geneff{\operation{event}_a})} \quad=\quad \lambda k : 1 \to \answertype.\ \operation{event}_a(k\ ()). \]
	Note that we use the notation \eqref{eq:n-ary-algebraic-operation-notation} for the unary operation $\operation{event}_a$.
	In our target language, $\operation{event}_a$ is interpreted as
	\[ \mathcal{A}^{\emalgsymbol_{\mathrm{tr}}} \interpret{\operation{event}_a(M)}(x) \quad=\quad \langle a \rangle (\mathcal{A}^{\emalgsymbol_{\mathrm{tr}}} \interpret{M}(x)). \]
	The modal operator $\langle a \rangle$ here has the same meaning as the modal operator in the HFL used in~\cite{kobayashi2018}. 
	Compared to their target language, our target language lacks the dual modal operator $[a]$ because $[a]$ is not used for the reduction for the safety property.
	The target language with dual modal operators is discussed in \referappendix{sec:duality}{I}.
	Such an extension allows us to take the negation of a formula and is useful when considering may-reachability.

	The generic effect for nondeterministic branching $\join$ is CPS-transformed to the corresponding modal operator
	\[ \CPS{(\geneff{\join})} \quad=\quad \lambda k : 1 + 1 \to \answertype.\ k\ (\iota_1\ ()) \join k\ (\iota_2\ ()) \]
	where we use the notation \eqref{eq:n-ary-algebraic-operation-notation} for the binary infix operator $\join$.
	The interpretation of the modal operator for $\join$ is given by conjunction.
	\[ \mathcal{A}^{\emalgsymbol_{\mathrm{tr}}} \interpret{M \join N}(\gamma) \quad=\quad \mathcal{A}^{\emalgsymbol_{\mathrm{tr}}} \interpret{M}(\gamma) \cap \mathcal{A}^{\emalgsymbol_{\mathrm{tr}}} \interpret{N}(\gamma) \]

	Similarly to Example~\ref{ex:total-partial-correctness-cps}, the interpretation of a fixed point is given by the greatest fixed point with respect to the inclusion order on $2^U$ because we consider the opposite order $\answerobj = (2^U, {\supseteq})$.
\end{example}

\begin{example}[expected cost analysis]\label{ex:expected-cost-cps}
	Given a well-typed $\lambda_c(\Sigma)$-term $\Gamma \vdash M : \rho$ such that $\rho$ and types in $\Gamma$ are ground types, its expected cost (i.e.\ the expected number of $\mathbf{tick}$ operations invoked by $M$) is given by the weakest precondition $\mathrm{wp}^{\emalgsymbol_{\mathrm{ex}}}[\mathcal{A}\interpret{M}](0) = \emalgsymbol^{P} \comp P \pi_1 \comp \mathcal{A}\interpret{M} : \mathcal{A}\interpret{\Gamma} \to \mathbb{W}$ by Example~\ref{ex:expected-cost-wp}.
	By Theorem~\ref{thm:cps-is-wpt-with-recursion}, the CPS transformation gives the expected cost $\mathrm{wp}^{\emalgsymbol_{\mathrm{ex}}}[\mathcal{A}\interpret{M}](0) = \mathcal{A}^{\emalgsymbol_{\mathrm{ex}}} \interpret{\CPS{M}\ (\lambda x. 0)}$.
	This corresponds to~\cite[Thm~4.15]{avanzini2021}.

	Generic effects in $\lambda_c(\Sigma)$-terms are CPS-transformed to corresponding modal operators.
	\begin{gather}
		\CPS{(\geneff{\operation{unif}})} = \lambda k.\ \operation{unif}(k, ()) \qquad
		\CPS{(M^{\checkmark})} = \lambda k.\ (\CPS{M}\ k)^{\checkmark} \qquad
		\CPS{(M_1 +_p M_2)} = \lambda k.\ (\CPS{M_1}\ k +_p \CPS{M_2}\ k)
	\end{gather}
	That is, the modal operator for tick $({-})^{\checkmark}$, probabilistic branching ${+_p}$ (or a Bernoulli distribution), and sampling from the uniform distribution $\operation{unif}$ correspond to the addition of $1$, a weighted sum, and the integration over the uniform distribution, respectively.
	\begin{align}
		&\mathcal{A}^{\emalgsymbol_{\mathrm{ex}}} \interpret{M^{\checkmark}}(\gamma) &&=\qquad 1 + \mathcal{A}^{\emalgsymbol_{\mathrm{ex}}} \interpret{M}(\gamma) \\
		&\mathcal{A}^{\emalgsymbol_{\mathrm{ex}}}\interpret{\operation{unif}(\lambda x : \mathbf{real}.\ M, ())}(\gamma) &&=\qquad \int_{[0, 1]} \mathcal{A}^{\emalgsymbol_{\mathrm{ex}}}\interpret{M}(\gamma, x) \,\mathrm{d} x \\
		&\mathcal{A}^{\emalgsymbol_{\mathrm{ex}}} \interpret{M_1 +_p M_2}(\gamma) &&=\qquad p \cdot \mathcal{A}^{\emalgsymbol_{\mathrm{ex}}} \interpret{M_1}(\gamma) + (1 - p) \cdot \mathcal{A}^{\emalgsymbol_{\mathrm{ex}}} \interpret{M_2}(\gamma)
	\end{align}
	Fixed points are interpreted as the least fixed points with respect to the standard order of $[0, \infty]$.
\end{example}

\begin{example}[cost moment analysis]\label{ex:cost-moment-cps}
	For any $\lambda_c(\Sigma)$-term $\Gamma \vdash M : \rho$ such that $\rho$ and types in $\Gamma$ are ground types, the tuple of moments of cost is given by the weakest precondition $\mathrm{wp}^{\emalgsymbol_{\mathrm{mo}, n}}[\mathcal{A}\interpret{M}](\mathbf{0}) = \emalgsymbol^{P}_n \comp P \mathrm{pow}_n \comp P \pi_1 \comp \mathcal{A}\interpret{M} : \mathcal{A}\interpret{\Gamma} \to \mathbb{W}^n$ by Example~\ref{ex:cost-moment-wp}.
	By Theorem~\ref{thm:cps-is-wpt-with-recursion}, the CPS transformation gives the moments of cost $\mathrm{wp}^{\emalgsymbol_{\mathrm{mo}, n}}[\mathcal{A}\interpret{M}](\mathbf{0}) = \mathcal{A}^{\emalgsymbol_{\mathrm{mo}, n}} \interpret{\CPS{M}\ (\lambda x. \mathbf{0})}$.

	Similarly to Example~\ref{ex:expected-cost-cps}, the modal operator for probabilistic branching ${+_p}$ and sampling $\operation{unif}$ are interpreted as the (component-wise) weighted sum and the (component-wise) integration, respectively.
	On the other hand, the modal operator for tick is interpreted by the elapse function $\mathcal{A}^{\emalgsymbol_{\mathrm{mo}, n}} \interpret{M^{\checkmark}}(x) = 1 \oplus \mathcal{A}^{\emalgsymbol_{\mathrm{mo}, n}} \interpret{M}(x)$.
\end{example}

\begin{example}[conditional weakest preexpectation]\label{ex:conditional-wp-cps}
	We apply the CPS transformation to obtain the conditional weakest preexpectation of a $\lambda_c(\Sigma)$-term $\Gamma \vdash M : \rho$.
	Based on the observations in Example~\ref{ex:conditional-wp-wp}, we aim to obtain the weakest precondition component-wise.
	We consider two types $\answertype_1$ and $\answertype_2$, which are interpreted by $\answerobj_1$ and $\answerobj_2$, respectively.
	By Theorem~\ref{thm:cps-is-wpt-with-recursion} and Example~\ref{ex:conditional-wp-wp}, we have $\mathrm{wp}^{\emalgsymbol_{\mathrm{cwp}, i}}[\mathcal{A}\interpret{M}](\mathcal{A}^{\emalgsymbol_{\mathrm{cwp}, i}}\interpret{Q_i}) = \mathcal{A}^{\emalgsymbol_{\mathrm{cwp}, i}}\interpret{\CPS{M}\ (\lambda x. Q_i)}$ for $i = 1, 2$ and $x : \rho \vdash Q_i : \answertype_i$.
	The modal operator $\operation{score}$ for conditioning is interpreted by the following multiplication.
	\[ \mathcal{A}^{\emalgsymbol_{\mathrm{cwp}, i}}\interpret{\operation{score}(\lambda x:1. M, N)}(\gamma) \quad=\quad \mathcal{A}^{\emalgsymbol_{\mathrm{cwp}, i}}\interpret{N}(\gamma) \cdot \mathcal{A}^{\emalgsymbol_{\mathrm{cwp}, i}}\interpret{M[()/x]}(\gamma) \]
	The interpretation of the modal operators for $\operation{unif}$ and ${+}_p$ is the same as Example~\ref{ex:expected-cost-cps}.
	Therefore, these modal operators are interpreted in the same way for both $i = 1$ and $i = 2$.
	However, the interpretations of fixed points are different: $\mathcal{A}^{\emalgsymbol_{\mathrm{cwp}, 1}}\interpret{\mu f. M}$ is the least fixed point with respect to the standard order of $[0, \infty]$, whereas $\mathcal{A}^{\emalgsymbol_{\mathrm{cwp}, 2}}\interpret{\mu f. M}$ is the greatest with respect to the standard order of $[0, 1]$, according to order relations defined on $\answerobj_1$ and $\answerobj_2$.
\end{example}

\subsection{Extending the Target Language}\label{subsec:extending-target-language}
If $\mathcal{A}^{\emalgsymbol} \interpret{\answertype} = \answerobj$ has an algebraic structure like a lattice structure, we can extend the target language using operators of the algebraic structure.
Although this doesn't essentially change the CPS transformation, such extensions are useful to rewrite tricky modal operators with other well-known term constructors.

\begin{definition}[extended $\lambda_{\mathrm{HFL}}$-terms]
	Suppose we have an $n$-ary operator $\mathrm{op}^{\answerobj} \in \category{C}(\answerobj^n, \answerobj)$.
	We extend the syntax of $\lambda_{\mathrm{HFL}}$-terms by $M \coloneqq \dots \mid \mathrm{op}(M_1, \dots, M_n)$ with the following typing rule.
	\begin{mathpar}
		\inferrule{
			\Gamma \vdash M_1 : \answertype \\
			\dots \\
			\Gamma \vdash M_n : \answertype
		}{
			\Gamma \vdash \mathrm{op}(M_1, \dots, M_n) : \answertype
		}
	\end{mathpar}
	The interpretation is given by
	\[ \mathcal{A}^{\emalgsymbol}\interpret{\mathrm{op}(M_1, \dots, M_n)} \quad=\quad \mathrm{op}^{\answerobj} \comp \langle \mathcal{A}^{\emalgsymbol}\interpret{M_1}, \dots, \mathcal{A}^{\emalgsymbol}\interpret{M_n} \rangle. \]
\end{definition}

A typical example is when $\answerobj$ is an internal bounded distributive lattice.
\begin{definition}[internal bounded distributive lattice]
	A \emph{bounded distributive lattice internal to $\category{C}$} is a tuple $(\answerobj, \top, {\land}, \bot, {\lor})$ where $\answerobj \in \category{C}$; and $\top, \bot : 1 \to \answerobj$ and ${\land}, {\lor} : \answerobj^2 \to \answerobj$ are morphisms in $\category{C}$ that satisfies the equational axioms of bounded distributive lattices.
	That is, $\lor$ and $\land$ are idempotent, commutative, and associative binary operations; $\bot$ and $\top$ are the unit element for $\lor$ and $\land$, respectively; and $\lor$ and $\land$ satisfy the absorption and the distributive laws.
\end{definition}

\begin{example}[safety property, continued from Example~\ref{ex:trace-property-cps}]\label{ex:trace-propery-extended}
	For safety property, $\answerobj = (2^U, {\supseteq})$ has an internal bounded distributive lattice structure $(\answerobj, \mathbf{true}, {\land}, \mathbf{false}, {\lor})$ defined by $(\answerobj, U, {\cap}, \emptyset, {\cup})$.
	Note that the internal bounded distributive lattice structure is not ``reversed'' here although we use the reversed inclusion order ${\supseteq}$ for $\answerobj$.
	We extend $\lambda_{\mathrm{HFL}}$-terms by $M, N \coloneqq \dots \mid \mathbf{true} \mid \mathbf{false} \mid M \land N \mid M \lor N$.
	Using the extended $\lambda_{\mathrm{HFL}}$-terms, we can replace ${\join}$ with $\land$ because they are semantically equivalent: 
	$\mathcal{A}^{\emalgsymbol_{\mathrm{tr}}} \interpret{M \join N} = \mathcal{A}^{\emalgsymbol_{\mathrm{tr}}} \interpret{M \land N}$.
	Thus, we can redefine our CPS transformation as
	\[ \CPS{(\geneff{\join})} \quad=\quad \lambda k : 1 + 1 \to \answertype.\ k\ (\iota_1\ ()) \land k\ (\iota_2\ ()).\]
\end{example}

\begin{example}[expected cost analysis, continued from Example~\ref{ex:expected-cost-cps}]\label{ex:expected-cost-extended}
	For expected cost analysis, $\answerobj = \mathbb{W}$ has the additive and the multiplicative monoid structure.
	So, we extend $\lambda_{\mathrm{HFL}}$-terms with $({+}), ({\cdot}) : \mathbb{W}^2 \to \mathbb{W}$ and constants $w \in \mathbb{W}$.
	That is, we define $M, N \coloneqq \dots \mid w \mid M + N \mid M \cdot N$.
	Modal operators $({-})^{\checkmark}$ and $M_1 +_p M_2$ in extended $\lambda_{\mathrm{HFL}}$-terms are semantically equivalent to $1 + ({-})$ and $p \cdot M_1 + (1 - p) \cdot M_2$, respectively, where the subtraction in $1 - p$ is a meta-level operation.
	Now, we redefine the CPS transformation as
	\[ \CPS{(M^{\checkmark})}\ =\ \lambda k. 1 + \CPS{M}\ k \qquad \CPS{(M_1 +_p M_2)}\ =\ \lambda k. p \cdot \CPS{M_1}\ k + (1 - p) \cdot \CPS{M_2}\ k. \]
\end{example}

We note that when $\emalgsymbol : T \answerobj \to \answerobj$ has a structure for the de Morgan duality, we can extend the $\lambda_{\mathrm{HFL}}$-terms with negation, which clarifies, for example, the duality between total correctness and partial correctness (see \referappendix{sec:duality}{I}).
We can also extend $\lambda_{\mathrm{HFL}}$-terms with quantifiers as explained in \referappendix{sec:quantifiers}{J}.

\section{Related Work}
\subsection{Generic Weakest Preconditions}
The weakest precondition transformer proposed by Dijkstra~\cite{dijkstra1975} is for guarded command language (GCL), which is an imperative language with nondeterminism.
Dijkstra's weakest precondition transformer is extended to, for example, \emph{probabilistic} GCL~\cite{mciver2001,kaminski2018} and further extended to a probabilistic functional language~\cite{avanzini2021}.
\emph{Separation logic}~\cite{reynolds2002} is an extension of Hoare logic that is more suited for reasoning about pointers.
Recently, a concurrent extension of separation logic, Iris~\cite{jung2015}, is applied to verify, for example, Rust programs~\cite{jung2018a} and effect handlers~\cite{devilhena2021}.

There are several works that aim to give uniform accounts of various weakest preconditions in category-theoretic frameworks~\cite{hasuo2015,martin2006,aguirre2020,goncharov2013}.
A framework based on fibrations and monad liftings \cite{aguirre2020} captures a wide class of generic weakest preconditions, and our Def.~\ref{def:weakest-precondition} is based on their work.
However, most of such categorical frameworks (including \cite{aguirre2020}) lack syntactic aspects of weakest preconditions or are limited to imperative programs whereas our framework focuses on syntactic computation of weakest preconditions for higher-order functional programs.
We also extended the list of examples in~\cite{aguirre2020} by adding safety properties and may/must-reachability and by considering domain theoretic models.
Categorical semantics for separation logic is studied using BI-hyperdoctrines \cite{kammar2017,biering2007,bizjak2018,polzer2020}.
It might be possible to instantiate our result to separation logic using EM algebras over BI-hyperdoctrines, but we leave it as future work.

\subsection{CPS Transformations and Weakest Preconditions}
For first-order imperative languages, the relationship between continuation-passing style and weakest precondition transformers is already observed \cite{jensen1978,audebaud1999}.
As for higher-order functional languages, this relationship has a few applications.
One is the expected cost analyses of probabilistic programs \cite{avanzini2021}.
Another application is trace properties for programs with nondeterminism and output \cite{kobayashi2018,kobayashi2009a}, although they didn't make it explicit that behind their method (except for linear-time temporal properties~\cite[Section~7]{kobayashi2018}) is the relationship between CPS transformations and weakest preconditions.
Note that the correctness of these works is proved for specific computational effects, while our result gives a general framework that subsumes these works.

\subsection{Dijkstra Monads}
The relationship between CPS transformations and weakest preconditions is used for \emph{Dijkstra monads}~\cite{swamy2013a,ahman2017,maillard2019}, which integrate computation of weakest precondition transformers with dependent type systems.
Our framework is closely related to Dijkstra monads, but existing results on Dijkstra monads have difficulty in handling programs with nondeterminism and probabilities.
That is, their results cannot be applied to some of our instances (e.g.\ safety properties and expected cost analyses).
To the best of our knowledge, our framework is the most \emph{general} one for \emph{syntactic} weakest preconditions, and none of the existing frameworks can cover all instances in Table~\ref{tab:instances}.

More concretely, a general recipe to derive Dijkstra monads using a CPS transformation is proposed in \cite{ahman2017}, but the correctness of their result is proved using deterministic semantics, which makes it difficult to apply their result to programs with nondeterminism or probabilities.
Their work was followed by a categorical exposition of Dijkstra monads~\cite{maillard2019}.
Their key insights are that Dijkstra monads correspond to \emph{monadic relations}.
The Dijkstra monads considered in \cite{ahman2017} correspond to a subclass of monadic relations, which are obtained from the continuation monad pseudo-transformer.
In contrast, our setting (Def.~\ref{def:weakest-precondition}) corresponds to another subclass of monadic relations derived from Eilenberg--Moore algebras.
The work \cite{maillard2019} explains the syntactic aspect only for the former subclass.
Our Theorem~\ref{thm:cps-is-wpt-with-recursion} provides the syntactic counterpart for the latter subclass.

\section{Conclusions and Future Work}
We provided a general framework for syntactic computation of generic weakest preconditions for effectful functional programs with recursion.
We instantiated our framework to various problems of program verification.

In future work, we aim to extend our framework to effect handlers~\cite{hillerstrom2017}.
Another direction is to look for a way to solve logical constraints since we have a reduction from programs to logical constraints.
We would like to also seek more instances of our framework, such as separation logic and verification of probabilistic programs with conditioning.

\bibliographystyle{ACM-Reference-Format}
\bibliography{cps-wpt}

\ifthenelse{\boolean{longversion}}{
	\appendix
	\section{Source Language}

Unlike the source language defined in Section~\ref{sec:source-language}, the source language in this section is defined using \emph{algebraic operations} (denoted by $o_{\rho}\ M$).
However, this is just a matter of taste because we can define generic effects as a syntactic sugar as follows.
\[ \mathbf{gen}_o\ M \coloneqq o_{\mathrm{ar}(o)}\ (\lambda y. y, M) \]

\subsection{Typing Rules}\label{subsec:source-typing-rules}
\begin{mathpar}
	\inferrule{
		(x : \rho) \in \Gamma
	}{
		\Gamma \vdash x : \rho
	}
	\and
	\inferrule{
		\Gamma \vdash M : \mathrm{ar}(c)
	}{
		\Gamma \vdash c\ M : \mathrm{car}(c)
	}
	\and
	\inferrule{
		\Gamma \vdash M : (\mathrm{ar}(o) \to \rho) \times \mathrm{car}(o)
	}{
		\Gamma \vdash o_{\rho}\ M : \rho
	}
	\and
	\inferrule{ }{
		\Gamma \vdash () : 1
	}
	\and
	\inferrule{
		\Gamma \vdash M : \rho \\
		\Gamma \vdash N : \tau
	}{
		\Gamma \vdash (M, N) : \rho \times \tau
	}
	\and
	\inferrule{
		\Gamma \vdash M : \rho_1 \times \rho_2
	}{
		\Gamma \vdash \pi_1\ M : \rho_1
	}
	\and
	\inferrule{
		\Gamma \vdash M : \rho_1 \times \rho_2
	}{
		\Gamma \vdash \pi_2\ M : \rho_2
	}
	\and
	\inferrule{
		\Gamma \vdash M : 0
	}{
		\Gamma \vdash \delta(M) : \rho
	}
	\and
	\inferrule{
		\Gamma \vdash M : \rho
	}{
		\Gamma \vdash \iota_1\ M : \rho + \tau
	}
	\and
	\inferrule{
		\Gamma \vdash M : \tau
	}{
		\Gamma \vdash \iota_2\ M : \rho + \tau
	}
	\and
	\inferrule{
		\Gamma \vdash M : \rho_1 + \rho_2 \\
		\Gamma, x_1 : \rho_1 \vdash N_1 : \tau \\
		\Gamma, x_2 : \rho_2 \vdash N_2 : \tau
	}{
		\Gamma \vdash \delta(M, x_1 {:} \rho_1. N_1, x_2 {:} \rho_2. N_2) : \tau
	}
	\and
	\inferrule{
		\Gamma, x : \rho \vdash M : \tau
	}{
		\Gamma \vdash \lambda x : \rho. M : \rho \to \tau
	}
	\and
	\inferrule{
		\Gamma \vdash M : \rho \to \tau \\
		\Gamma \vdash N : \rho
	}{
		\Gamma \vdash M\ N : \tau
	}
	\and
	\inferrule{
		\Gamma, f : \rho \to \tau, x : \rho \vdash M : \tau \\
		\Gamma, f : \rho \to \tau \vdash N : \tau'
	}{
		\Gamma \vdash \letrecchurch{f}{x}{\rho}{\tau}{M}{N} : \tau'
	}
\end{mathpar}

\subsection{Semantics}\label{subsec:source-semantics}

Let $\mathcal{A} = (\category{C}, T, A, a)$ be a $\lambda_c(\Sigma)$-structure.

\textbf{Types}: $\mathcal{A}\interpret{\rho} \in \category{C}$.
\begin{gather}
	\mathcal{A}\interpret{b} = A b \qquad
	\mathcal{A}\interpret{1} = 1 \qquad
	\mathcal{A}\interpret{\rho \times \tau} = \mathcal{A}\interpret{\rho} \times \mathcal{A}\interpret{\tau} \\
	\mathcal{A}\interpret{0} = 0 \qquad
	\mathcal{A}\interpret{\rho + \tau} = \mathcal{A}\interpret{\rho} + \mathcal{A}\interpret{\tau} \qquad
	\mathcal{A}\interpret{\rho \to \tau} = \exponential{\mathcal{A}\interpret{\rho}}{T \mathcal{A}\interpret{\tau}}
\end{gather}

\textbf{Contexts}: $\mathcal{A}\interpret{\Gamma} \in \category{C}$.
\begin{gather}
	\mathcal{A}\interpret{\cdot} = 1 \qquad
	\mathcal{A}\interpret{\Gamma, x : \rho} = \mathcal{A}\interpret{\Gamma} \times \mathcal{A}\interpret{\rho}
\end{gather}

\textbf{Terms}:
For each well-typed term $\Gamma \vdash M : \rho$, we define $\mathcal{A} \interpret{M} : \mathcal{A} \interpret{\Gamma} \to T \mathcal{A} \interpret{\rho}$ by
\begin{align}
	\mathcal{A} \interpret{\Gamma, x : \rho \vdash y : \tau} &= \begin{cases}
		\mathcal{A} \interpret{\Gamma \vdash y : \tau} \comp \pi_1 & x \neq y \\
		\eta^T \comp \pi_2 & x = y
	\end{cases} \\
	\mathcal{A} \interpret{c\ M} & = T(a(c)) \comp \mathcal{A} \interpret{M} \\
	\mathcal{A} \interpret{o_{\rho}\ M} & = \mu^T \comp T(\Lambda^{-1}(a(o)_{\mathcal{A}\interpret{\rho}})) \comp \mathcal{A} \interpret{M} \\
	\mathcal{A} \interpret{()} &= \eta^T \comp {!} \\
	\mathcal{A} \interpret{(M_1, M_2)} &= \mu^T \comp T \strength^T \comp T \braiding \comp \strength^T \comp \braiding \comp \langle \mathcal{A} \interpret{M_1}, \mathcal{A} \interpret{M_2} \rangle \\
	\mathcal{A} \interpret{\pi_i\ M} &= T \pi_i \comp \mathcal{A} \interpret{M} \\
	\mathcal{A} \interpret{\delta(M)} &= T {?} \comp \mathcal{A} \interpret{M} \\
	\mathcal{A} \interpret{\iota_i\ M} &= T \iota_i \comp \mathcal{A} \interpret{M} \\
	\mathcal{A} \interpret{\delta(M, x_1 : \rho_1. M_1, x_2 : \rho_2. M_2)} &= \mu^T \comp T [\mathcal{A} \interpret{M_1}, \mathcal{A} \interpret{M_2}] \comp T [\identity{} \times \iota_1, \identity{} \times \iota_2]^{-1} \comp \strength^T \comp \langle \identity{}, \mathcal{A} \interpret{M} \rangle \\
	\mathcal{A} \interpret{\lambda x : \rho. M} &= \eta^T \comp \Lambda (\mathcal{A} \interpret{M}) \\
	\mathcal{A} \interpret{M\ N} &= \mu^T \comp T \mathbf{ev} \comp \mu^T \comp T \strength^T \comp T \braiding \comp \strength^T \comp \braiding \comp \langle \mathcal{A} \interpret{M}, \mathcal{A} \interpret{N} \rangle \\
	\mathcal{A} \interpret{\letrec{f}{x}{M}{N}} &= \mathcal{A} \interpret{N} \comp \langle \identity{}, (\Lambda (\mathcal{A} \interpret{M}))^{\dagger} \rangle
\end{align}
where
\begin{itemize}
	\item For $f : X \to Y$ and $g : X \to Z$, $\tupling{f}{g} : X \to Y \times Z$ is the tupling.
	\item $\pi_1 : X \times Y \to X$ and $\pi_2 : X \times Y \to Y$ are the first and the second projection.
	\item ${!} : X \to 1$ is a unique morphism to a terminal object.
	\item For $f : X \to Z$ and $g : Y \to Z$, $\cotupling{f}{g} : X + Y \to Z$ is the cotupling.
	\item $\iota_1 : X \to X + Y$ and $\iota_2 : Y \to X + Y$ are coprojection.
	\item ${?} : 0 \to X$ is a unique morphism from an initial object.
	\item $\mathbf{ev}_{X, Y} : (\exponential{X}{Y}) \times X \to Y$ is the evaluation morphism.
	\item $\Lambda_{X, Y, Z} : \category{C}(X \times Y, Z) \to \category{C}(X, \exponential{Y}{Z})$ is the currying.
	\item $[\identity{} \times \iota_1, \identity{} \times \iota_2]^{-1} : X \times (Y + Z) \to X \times Y + X \times Z$ is the inverse of the distributivity isomorphism.
	\item $\braiding_{X, Y} : X \times Y \to Y \times X$ is the braiding.
	\item $\eta^T_X : X \to T X$ is the unit of a monad $T$.
	\item $\mu^T_X : T^2 X \to T X$ is the multiplication of a monad $T$.
	\item $\strength^T_{X, Y} : X \times T Y \to T (X \times Y)$ is the strength of a strong monad $T$.
\end{itemize}

\section{Target Language}

\subsection{Typing Rules}\label{subsec:target-typing-rules}
\begin{mathpar}
	\inferrule{
		(x : \rho) \in \Gamma
	}{
		\Gamma \vdash x : \rho
	}
	\and
	\inferrule{
		\Gamma \vdash M : \mathrm{ar}(c)
	}{
		\Gamma \vdash c\ M : \mathrm{car}(c)
	}
	\and
	\inferrule{
		\Gamma \vdash M : (\mathrm{ar}(o) \to \answertype) \times \mathrm{car}(o)
	}{
		\Gamma \vdash o\ M : \answertype
	}
	\and
	\inferrule{ }{
		\Gamma \vdash () : 1
	}
	\and
	\inferrule{
		\Gamma \vdash M : \rho \\
		\Gamma \vdash N : \tau
	}{
		\Gamma \vdash (M, N) : \rho \times \tau
	}
	\and
	\inferrule{
		\Gamma \vdash M : \rho_1 \times \rho_2
	}{
		\Gamma \vdash \pi_1\ M : \rho_1
	}
	\and
	\inferrule{
		\Gamma \vdash M : \rho_1 \times \rho_2
	}{
		\Gamma \vdash \pi_2\ M : \rho_2
	}
	\and
	\inferrule{
		\Gamma \vdash M : 0
	}{
		\Gamma \vdash \delta(M) : \answertype
	}
	\and
	\inferrule{
		\Gamma \vdash M : \rho
	}{
		\Gamma \vdash \iota_1\ M : \rho + \tau
	}
	\and
	\inferrule{
		\Gamma \vdash M : \tau
	}{
		\Gamma \vdash \iota_2\ M : \rho + \tau
	}
	\and
	\inferrule{
		\Gamma \vdash M : \rho_1 + \rho_2 \\
		\Gamma, x_1 : \rho_1 \vdash M_1 : \answertype \\
		\Gamma, x_2 : \rho_2 \vdash M_2 : \answertype
	}{
		\Gamma \vdash \delta(M, x_1 : \rho_1. M_1, x_2 : \rho_2. M_2) : \answertype
	}
	\and
	\inferrule{
		\Gamma, x : \rho \vdash M : \answertype
	}{
		\Gamma \vdash \lambda x : \rho. M : \rho \to \answertype
	}
	\and
	\inferrule{
		\Gamma \vdash M : \rho \to \answertype \\
		\Gamma \vdash N : \rho
	}{
		\Gamma \vdash M\ N : \answertype
	}
	\and
	\inferrule{
	\Gamma, f : \rho \to \answertype, x : \rho \vdash M : \answertype \\
	\Gamma, f : \rho \to \answertype \vdash N : \tau
	}{
	\Gamma \vdash \letrec{f}{(x : \rho)}{M}{N} : \tau
	}
\end{mathpar}

\subsection{Semantics}\label{subsec:target-semantics}
Let $\emalgsymbol : T \answerobj \to \answerobj$ be a $T$-algebra.
We define $\mathcal{A}^{\emalgsymbol} \llbracket {-} \rrbracket$ as follows.

For types:
\begin{gather}
	\mathcal{A}^{\emalgsymbol}\interpret{b} = A b \qquad
	\mathcal{A}^{\emalgsymbol} \interpret{\answertype} = \answerobj \qquad
	\mathcal{A}^{\emalgsymbol}\interpret{1} = 1 \qquad
	\mathcal{A}^{\emalgsymbol}\interpret{\rho \times \tau} = \mathcal{A}^{\emalgsymbol}\interpret{\rho} \times \mathcal{A}^{\emalgsymbol}\interpret{\tau} \\
	\mathcal{A}^{\emalgsymbol}\interpret{0} = 0 \qquad
	\mathcal{A}^{\emalgsymbol}\interpret{\rho + \tau} = \mathcal{A}^{\emalgsymbol}\interpret{\rho} + \mathcal{A}^{\emalgsymbol}\interpret{\tau} \qquad
	\mathcal{A}^{\emalgsymbol}\interpret{\rho \to \tau} = \exponential{\mathcal{A}^{\emalgsymbol}\interpret{\rho}}{\mathcal{A}^{\emalgsymbol}\interpret{\tau}}
\end{gather}

For contexts, we define $\mathcal{A}^{\emalgsymbol} \interpret{\cdot} = 1$ and $\mathcal{A}^{\emalgsymbol} \interpret{\Gamma, x : \rho} = \mathcal{A}^{\emalgsymbol} \interpret{\Gamma} \times \mathcal{A}^{\emalgsymbol} \interpret{\rho}$.

For any well-typed term $\Gamma \vdash M : \rho$, the interpretation $\mathcal{A}^{\emalgsymbol} \interpret{\Gamma \vdash M : \rho} : \mathcal{A}^{\emalgsymbol} \interpret{\Gamma} \to \mathcal{A}^{\emalgsymbol} \interpret{\rho}$ (or denoted simply by $\mathcal{A}^{\emalgsymbol} \interpret{M}$ if there is no fear of confusion) is defined as follows.
\begin{align}
	\mathcal{A}^{\emalgsymbol} \interpret{\Gamma, x : \rho \vdash y : \tau} &= \begin{cases}
		\mathcal{A}^{\emalgsymbol} \interpret{\Gamma \vdash y : \tau} \comp \pi_1 & x \neq y \\
		\pi_2 & x = y
	\end{cases} \\
	\mathcal{A}^{\emalgsymbol} \interpret{c\ M} & = a(c) \comp \mathcal{A}^{\emalgsymbol} \interpret{M} \\
	\mathcal{A}^{\emalgsymbol} \interpret{o\ M} & = \emalgsymbol \comp T \mathbf{ev} \comp \strength^T \comp (\identity{} \times \mathbf{Gef}(a(o))) \comp \mathcal{A}^{\emalgsymbol} \interpret{M} \\
	\mathcal{A}^{\emalgsymbol} \interpret{()} &= {!} \\
	\mathcal{A}^{\emalgsymbol} \interpret{(M_1, M_2)} &= \tupling{\mathcal{A}^{\emalgsymbol} \interpret{M_1}}{\mathcal{A}^{\emalgsymbol} \interpret{M_2}} \\
	\mathcal{A}^{\emalgsymbol} \interpret{\pi_i M} &= \pi_i \comp \mathcal{A}^{\emalgsymbol} \interpret{M} \\
	\mathcal{A}^{\emalgsymbol} \interpret{\delta(M)} &= {?} \comp \mathcal{A}^{\emalgsymbol} \interpret{M} \\
	\mathcal{A}^{\emalgsymbol} \interpret{\iota_i M} &= \iota_i \comp \mathcal{A}^{\emalgsymbol} \interpret{M} \\
	\mathcal{A}^{\emalgsymbol} \interpret{\delta(M, x_1 : \rho_1. M_1, x_2 : \rho_2. M_2)} &= [\mathcal{A}^{\emalgsymbol} \interpret{M_1}, \mathcal{A}^{\emalgsymbol} \interpret{M_2}] \comp [\identity{} \times \iota_1, \identity{} \times \iota_2]^{-1} \comp \tupling{\identity{}}{\mathcal{A}^{\emalgsymbol} \interpret{M}} \\
	\mathcal{A}^{\emalgsymbol} \interpret{\lambda x. M} &= \Lambda (\mathcal{A}^{\emalgsymbol} \interpret{M}) \\
	\mathcal{A}^{\emalgsymbol} \interpret{M\ N} &= \mathbf{ev} \comp \tupling{\mathcal{A}^{\emalgsymbol} \interpret{M}}{\mathcal{A}^{\emalgsymbol} \interpret{N}} \\
	\mathcal{A}^{\emalgsymbol} \interpret{\letrec{k}{(x : \rho)}{M}{N}} &= \mathcal{A}^{\emalgsymbol} \interpret{N} \comp \tupling{\identity{}}{(\Lambda (\mathcal{A}^{\emalgsymbol} \interpret{M}))^{\dagger}}
\end{align}

\section{An Example Program for Expected Cost Analysis and Cost Momemt Analysis}\label{sec:advanced-example}
We use the $\lambda_c$-signature defined in Example~\ref{ex:expected-cost-syntax} and assume that we have $\mathbf{int} \in B$ and that the set $K$ of effect-free constants contains basic operations for integers (see Example~\ref{ex:total-partial-correctness-syntax}).
Recall that we have if-then-else expressions as a syntactic sugar.
Since we have a uniform distribution, we extend the probabilistic branching operator as follows.
Given a term $N : \mathbf{real}$, we define $M_1 +_{N} M_2$ as a syntactic sugar for $\operation{unif}(\lambda x.\ \mathbf{if}\ x \le N\ \mathbf{then}\ M_1\ \mathbf{else}\ M_2, ())$ where $x$ is a fresh variable.

Now, we consider a variant of random walk that dynamically changes how to make a step.
\begin{gather}
	\mathrm{walk} : (1 \to \mathbf{int}) \times \mathbf{int} \to 1 \\
	\letrec{\mathrm{walk}}{(s, n)}{\mathbf{if}\ n \le 0\ \mathbf{then}\ ()\ \mathbf{else}\ (\mathrm{walk}\ (\mathrm{update}\ s, n + s\ ()))^{\checkmark}}{\mathrm{walk}\ (\lambda x. 0, 1)}
\end{gather}

We update a step function $s : 1 \to \mathbf{int}$ as follows.
\begin{gather}
	\mathrm{update} : (1 \to \mathbf{int}) \to 1 \to \mathbf{int} \\
	\mathrm{update} \coloneqq \lambda s. \operation{unif}(\lambda p.\ (\lambda x. s\ x +_p (-2)) +_{1/2} (\lambda x. s\ x +_p 1), ())
\end{gather}
That is, we update $s$ to a step function of the form $\lambda x. s\ x +_p a$ where $p$ is sampled from the uniform distribution on $[0, 1]$ and $a$ is either $-2$ or $1$.
Note that this program contains both higher-order functions and continuous distributions, which make the problem challenging.

For expected cost analysis, we apply the CPS transformation (Definition~\ref{def:cps-transformation}) and then pass the constant function $0$ as a postcondition.
Then, we get the following $\lambda_{\mathrm{HFL}}(\Sigma)$-term, which represents the expected cost of $\mathrm{walk}\ (\lambda x. 0, 1)$.
\begin{align}
	&\mathbf{let}\ \mathbf{rec}\ \mathrm{walk}'\ ((s, n), k) = \\
	&\qquad\qquad \mathbf{if}\ n \le 0\ \mathbf{then}\ k\ ()\ \mathbf{else}\ 1 + \mathrm{update}'\ (s, \lambda s'. s\ ((), \lambda y. \mathrm{walk}'\ ((s', n + y), k))) \\
	&\mathbf{in}\ \mathrm{walk}'\ ((\lambda (x, k). k\ 0, 1), \lambda x. 0)
\end{align}
where
\begin{align}
	\mathrm{update}' &\coloneqq \lambda (s, k). \int_{[0, 1]} \frac{1}{2} \cdot (k\ s'_{-2}) + \frac{1}{2} \cdot (k\ s'_1) \, \mathrm{d}p \\
	s'_{-2} &\coloneqq \lambda (x, k). p \cdot (s\ (x, k)) + (1-p) \cdot (k\ (-2)) \\
	s'_{1} &\coloneqq \lambda (x, k). p \cdot (s\ (x, k)) + (1-p) \cdot (k\ 1)
\end{align}
Note that $\mathrm{walk}'\ ((\lambda (x, k). k\ 0, 1), \lambda x. 0)$ has type $\answertype$ where $\answertype$ is interpreted as the type of extended nonnegative real numbers $[0, \infty]$.

Next, consider the cost moment analysis.
Suppose we are interested in up to the second moment.
Then, $\answertype$ is interpreted as $[0, \infty]^2$.
By applying the CPS transformation, we get the following.
\begin{align}
	&\mathbf{let}\ \mathbf{rec}\ \mathrm{walk}'\ ((s, n), k) = \\
	&\qquad\qquad \mathbf{if}\ n \le 0\ \mathbf{then}\ k\ ()\ \mathbf{else}\ 1 \oplus \mathrm{update}'\ (s, \lambda s'. s\ ((), \lambda y. \mathrm{walk}'\ ((s', n + y), k))) \\
	&\mathbf{in}\ \mathrm{walk}'\ ((\lambda (x, k). k\ 0, 1), \lambda x. (0, 0))
\end{align}
where
\begin{align}
	1 \oplus (M_1, M_2) &\coloneqq (1 + M_1, 1 + 2 \cdot M_1 + M_2) \\
	\mathrm{update}' &\coloneqq \lambda (s, k). \int_{[0, 1]} \frac{1}{2} \cdot (k\ s'_{-2}) + \frac{1}{2} \cdot (k\ s'_1) \, \mathrm{d}p \\
	s'_{-2} &\coloneqq \lambda (x, k). p \cdot (s\ (x, k)) + (1-p) \cdot (k\ (-2)) \\
	s'_{1} &\coloneqq \lambda (x, k). p \cdot (s\ (x, k)) + (1-p) \cdot (k\ 1)
\end{align}

\section{Fixing an Error in Existing Work (to be submitted as another paper)}\label{sec:fixing-error}
\newcommand{\yoneda}{y}

We fix an error in the proof of \cite[Thm~12]{katsumata2013}.
We focus on the recursion-free case but the same argument applies to the case with recursion.

\subsection{Preliminaries}
The proof in \cite[Thm~12]{katsumata2013} is based on a fibrational framework of logical relations~\cite{hermida1993}.
In this framework, we consider two layers of models of programs, that is, a functor $p : \category{E} \to \category{B}$.
We require $p$ to be a fibration.

\begin{definition}
	A functor $p : \category{E} \to \category{B}$ is a \emph{fibration} if $p$ satisfies the cartesian lifting property: for any $u : I \to p Y$ in $\category{B}$, there exist $X \in \category{E}$ and a cartesian morphism $f : X \to Y$ above $u$.
	Here, we say a morphism $f : X \to Y$ in $\category{E}$ is \emph{above} $u : I \to J$ if $p f = u$, and $f : X \to Y$ is \emph{cartesian} if for any $h : Z \to Y$ and $v : p Z \to p X$ such that $h$ is above $p f \comp v$, there exists a unique morphism $g : Z \to X$ above $v$ such that $h = g \comp f$.
\end{definition}

We introduce several terminology about fibrations.
Given a fibration $p : \category{E} \to \category{B}$, $\category{E}$ is called the \emph{total} category, and $\category{B}$ is called the \emph{base} category.
Given $I \in \category{B}$, the \emph{fibre category} $\category{E}_I$ is the category whose objects are objects in $\category{E}$ above $I$ and morphisms are morphisms in $\category{E}$ above the identity morphism $\identity{I}$.
By the cartesian lifting property, each morphism $u : I \to J$ in $\category{B}$ induces a \emph{reindexing functor} $u^{*} : \category{E}_J \to \category{E}_I$.
A fibration $p : \category{E} \to \category{B}$ is \emph{preordered} if $\category{E}_I$ is a preorder for each $I \in \category{B}$ and \emph{posetal} if $\category{E}_I$ is a poset.
If $p$ is a preordered fibration, we write $u : X \dot{\to} Y$ if there exists $f : X \to Y$ above $u : p X \to p Y$.
A fibration $p$ is a \emph{bifibration} if each reindexing functor $u^{*}$ has a left adjoint $u_{*} \dashv u^{*}$.

\begin{definition}
	A \emph{fibration for logical relations} $p : \category{E} \to \category{B}$ is a posetal bifibration over a bicartesian closed category $\category{B}$ with fibred small products (small products in each fibre category preserved by reindexing functors) such that $\category{E}$ is a bicartesian closed category and $p$ strictly preserves the bicartesian closed structure.
\end{definition}
If $p : \category{E} \to \category{B}$ is a fibration for logical relations, then we write a dot above each component of the bicartesian closed structure of $\category{E}$ (e.g.\ $\dot{\Rightarrow}$ and $\dot{\times}$) to distinguish it from that of $\category{B}$.

By \cite[Corollary~6, Proposition~7]{katsumata2013}, the subobject fibration $\mathbf{Sub}([\category{C}^{\op}, \Set]) \to [\category{C}^{\op}, \Set]$ of the presheaf category over a small category $\category{C}$ is a fibration for logical relations; and given a fibration for logical relations $p : \category{E} \to \category{C}$ and a finite-product preserving functor $F : \category{B} \to \category{C}$, the change-of-base construction gives a fibration for logical relations if $\category{B}$ is bicartesian closed.

Typically, the base category of $p : \category{E} \to \category{B}$ is an ordinary model of programs ($\lambda$-calculus), and the total category is a category of predicates (or relations) and predicate-preserving (relation-preserving) morphisms.
Since the total category $\category{E}$ is also a model of programs, we can interpret a program in $\category{E}$.
Since we have a functor $p$, the interpretation in $\category{E}$ gives an evidence that the interpretation in $\category{B}$ preserves predicates (or relations), and thus, we get the fundamental theorem of logical relations.

To discuss logical relations for computational effects, we need a monad on the base category and a lifting of the monad on the total category.
\cite{katsumata2005} provided $\top\top$-lifting as a construction of a lifting of a monad, and \cite[Thm~12]{katsumata2013} used it to relate two interpretations $\mathcal{A}_1\interpret{M} : \mathcal{A}_1 \interpret{\Gamma} \to T_1 \mathcal{A}_1 \interpret{\rho}$ and $\mathcal{A}_2\interpret{M} : \mathcal{A}_2 \interpret{\Gamma} \to T_2 \mathcal{A}_2 \interpret{\rho}$ of $\lambda_c$-calculus along a strong monad morphism $\phi : T_1 \to T_2$.

\begin{definition}
	Let $p : \category{E} \to \category{B}$ be a fibration for logical relations and $T$ a strong monad on $\category{B}$.
	A parameter $\mathcal{R} = (R, S)$ for a $\top\top$-lifting is a pair of functors $R : J \to \category{B}$ and $S : J \to \category{E}$ from a set $J$ such that $T \comp R = p \comp S$.
	A $\top\top$-lifting $T^{\top\top(\mathcal{R})}$ with respect to a parameter $\mathcal{R}$ is defined by
	\[ T^{\top\top(\mathcal{R})} X \coloneqq \bigwedge_{j \in J} T^{\top\top(R j, S j)} X \]
	where $T^{\top\top(R j, S j)} X$ is a pullback of $\dotExponential{(\dotExponential{X}{S j})}{S j}$ along the strong monad morphism $\phi_j : T \to \exponential{(\exponential{({-})}{T (R j)})}{T (R j)}$ defined by the (free) EM algebra structure of $T (R i)$.
	\begin{center}
		\begin{tikzcd}
			\category{E} \ar[d, "p"] & T^{\top\top(R j, S j)} X \ar[r] & \dotExponential{(\dotExponential{X}{S j})}{S j} \\
			\category{B} & T I \ar[r, "\phi_j"] & \exponential{(\exponential{I}{T (R j)})}{T (R j)}
		\end{tikzcd}
	\end{center}
\end{definition}

\subsection{The Error in the Proof}
The statement of \cite[Thm~12]{katsumata2013} is as follows.
\begin{conjecture}\label{conj:relating-effects}
	Let $\Sigma = (B, K, O, \mathrm{ar}, \mathrm{car})$ be a $\lambda_c$-signature, $\mathcal{A} = (\category{C}, T_1, A, a)$ be a $\lambda_c(\Sigma)$-structure, $T_2$ be a strong monad on $\category{C}$, and $\phi : T_1 \to T_2$ be a strong monad morphism.
	Then, for any well-typed $\lambda_c(\Sigma)$-term $x_1 : b_1, \dots, x_n : b_n \vdash M : b$, we have $\phi_{A b} \comp \mathcal{A} \interpret{M} = (\phi \mathcal{A}) \interpret{M}$.
\end{conjecture}

In the proof, we consider the following fibration $q : \category{K} \to \category{C} \times \category{C}$ defined by the change-of-base construction.
\begin{equation}
	\begin{tikzcd}
		\category{K} \ar[r] \ar[d, "q"] & \mathbf{Sub}([\category{C}^{\mathrm{op}}, \Set]) \ar[d] \\
		\category{C} \times \category{C} \ar[r, "D"] & {[\category{C}^{\mathrm{op}}, \Set]}
	\end{tikzcd}
\end{equation}
Here, $D : \category{C} \times \category{C} \to [\category{C}^{\mathrm{op}}, \Set]$ is defined by $D (I, J) = \yoneda I \times \yoneda J$ where $\yoneda$ is the Yoneda embedding.
An object in $\category{K}$ is a tuple $(X, I, I')$ where $I, I' \in \category{C}$ and $X$ is a subpresheaf of $\yoneda I \times \yoneda I'$.

In the base category $\category{C} \times \category{C}$, a $\lambda_c(\Sigma)$-term $M$ is interpreted as the pair of $\mathcal{A} \interpret{A}$ and $(\phi \mathcal{A}) \interpret{A}$.
We can define the interpretation of $M$ in $\category{K}$ using the following:
for each $b \in B$, we define
\begin{align}
	V b &\coloneqq \mathbf{Eq} (A b) = (\lambda H \in \category{C}. \{ (f, f) \mid f : H \to A b \}, A b, A b) &&\in \category{K} \\
	C b &\coloneqq (\lambda H \in \category{C}. \{ (f, \phi \comp f) \mid f : H \to T_1 (A b) \}, T_1 (A b), T_2 (A b)) &&\in \category{K}
\end{align}
where $\mathbf{Eq} : \category{C} \to \category{K}$ is defined by
$\mathbf{Eq} X \coloneqq (\lambda H \in \category{C}. \{ (f, f) \mid f : H \to X \}, X, X)$.
Now, we define a $\lambda_c(\Sigma)$-structure $\mathcal{V} = (\category{K}, \dots)$ by \cite[Theorem~7]{katsumata2013}
Since $q \comp C = (T_1 \times T_2) \comp \tupling{A}{A} : B \to \category{C} \times \category{C}$, this defines a $\top\top$-lifting $(T_1 \times T_2)^{\top\top(\tupling{A}{A}, C)}$.
The interpretation of base types in $\category{K}$ is defined by $V$.
The remaining part is to define the interpretation of effect-free constants and algebraic operations.
That is, we need to prove the following.
\begin{gather}
	a(c) : \mathcal{V} \interpret{\mathrm{ar}(c)} \dotTo \mathcal{V} \interpret{\mathrm{car}(c)} \label{eq:effect-free-lift} \\
	a(o)_{A b} : (\dotExponential{\mathcal{V}\interpret{\mathrm{ar}(o)}}{C b}) \dotTo (\dotExponential{\mathcal{V}\interpret{\mathrm{car}(o)}}{C b}) \label{eq:algebraic-operation-lift}
\end{gather}

However, if either the coarity of an effect-free constant or the arity of an algebraic operation contains coproducts, then \eqref{eq:effect-free-lift} does not hold.
For example, let $\category{C} = \Set$ and assume we have an effect-free constant $\mathrm{iszero} : \mathbb{N} \to 1 + 1$ with $a(\mathrm{iszero})(n) = \iota_1\ ()$ if and only if $n = 0$.
Then, $a(\mathrm{iszero}) : \mathcal{V}\interpret{\mathbb{N}} \dot{\to} \mathcal{V}\interpret{1 + 1}$ does not hold because $\mathcal{V}\interpret{1 + 1} = (\lambda H \in \Set. \{ (\iota_1 \comp {!}, \iota_1 \comp {!}), (\iota_2 \comp {!}, \iota_2 \comp {!}) \}, 1 + 1, 1 + 1)$ and $(a(\mathrm{iszero}), a(\mathrm{iszero})) \notin \{ (\iota_1 \comp {!}, \iota_1 \comp {!}), (\iota_2 \comp {!}, \iota_2 \comp {!}) \}$.

\begin{remark}
	The counterexample is for the proof strategy and not for the statement itself.
	In fact, we will later show that Conjecture~\ref{conj:relating-effects} is true if $\category{C}$ is stable, and $\category{C} = \Set$ is an example of a stable bicartesian closed category.
\end{remark}

\subsection{Correction}
\newcommand{\ddotExponential}[2]{#1 \mathrel{\ddot{\Rightarrow}} #2}
\newcommand{\ddotTimes}{\mathrel{\ddot{\times}}}
\newcommand{\ddotPlus}{\mathrel{\ddot{+}}}
\newcommand{\ddotTo}{\mathrel{\ddot{\to}}}

We assume that $\category{C}$ is stable.
The stability condition is used in~\cite{fiore1999} to characterise definability of morphisms to simply typed lambda calculus with products and sums.
W.l.o.g.\ we also assume that $\category{C}$ is small.
This is possible because we can take a small subcategory of $\category{C}$ that is closed under the interpretation of $\lambda_c(\Sigma)$-types/terms.

\begin{theorem}\label{thm:relating-effect}
	Let $\Sigma = (B, K, O, \mathrm{ar}, \mathrm{car})$ be a $\lambda_c$-signature, $\mathcal{A} = (\category{C}, T_1, A, a)$ be a $\lambda_c(\Sigma)$-structure such that $\category{C}$ is stable, $T_2$ be a strong monad on $\category{C}$, and $\phi : T_1 \to T_2$ be a strong monad morphism.
	Then, for any well-typed $\lambda_c(\Sigma)$-term $x_1 : b_1, \dots, x_n : b_n \vdash M : b$, we have $\phi_{A b} \comp \mathcal{A} \interpret{M} = (\phi \mathcal{A}) \interpret{M}$.	
\end{theorem}

We fix the proof by using $\top\top$-closure~\cite{katsumata2008}.
\begin{definition}[$\top\top$-closure]
	Let $p : \category{E} \to \category{B}$ be a fibration for logical relations.
	A \emph{closure parameter} is a functor $S : J \to \category{E}$ from a set $J$.
	A \emph{$\top\top$-closure operator} with respect to $S$ is a mapping $({-})^{\top\top(S)} : \category{E} \to \category{E}$ defined by
	\[ X^{\top\top(S)} \coloneqq \bigwedge_{j \in J} X^{\top\top(S j)} \]
	where $X^{\top\top(S j)}$ is a pullback of $\exponential{(\exponential{X}{S j})}{S j}$ along the unit $\eta$ of the continuation monad $\exponential{(\exponential{({-})}{p (S j)})}{p (S j)}$.
	\begin{center}
		\begin{tikzcd}
			\category{E} \ar[d] & X^{\top\top(S j)} \ar[r] & \exponential{(\exponential{X}{S j})}{S j} \\
			\category{B} & p X \ar[r, "\eta"] & \exponential{(\exponential{p X}{p (S j)})}{p (S j)}
		\end{tikzcd}
	\end{center}
\end{definition}

We take a full reflective subcategory $\category{K}^{\top\top (P)}$ of $\top\top$-closed objects (i.e.\ $X \in \category{K}^{\top\top (P)}$ if $X \in \category{K}$ and $X^{\top\top(P)} = X$) and interpreting well-typed $\lambda_c(\Sigma)$-terms in $\category{K}^{\top\top (P)}$ where the parameter $P$ is defined by $P = [P_{\mathrm{pure}}, C] : \mathbf{GTyp}(B) + B \to \category{K}$ and $P_{\mathrm{pure}} \rho = \mathbf{Eq}(\mathcal{A} \interpret{\rho})$.
The situation is depicted as follows.
\begin{center}
	\begin{tikzcd}
		\category{K}^{\top\top (P)} \ar[r, hookrightarrow, "i"] & \category{K} \ar[r] \ar[d, "q"] & \mathbf{Sub}([\category{C}^{\mathrm{op}}, \Set]) \ar[d] \\
		& \category{C} \times \category{C} \ar[r, "D"] & {[\category{C}^{\mathrm{op}}, \Set]}
	\end{tikzcd}
\end{center}

Recall that bicartesian closed structure of $\category{K}$ and $\category{K}^{\top\top (P)}$ is given as follows.
For $\category{K}$,
\begin{itemize}
	\item $\dot{0} = (\lambda H \in \category{C}. \emptyset, 0, 0)$
	\item $\dot{1} = (\lambda H \in \category{C}. \{ ({!}, {!}) \}, 1, 1)$
	\item $(X, I, I') \dotTimes (Y, J, J') = (\lambda H \in \category{C}. \{ (f, g) \mid (\pi_1 \comp f, \pi_1 \comp g) \in X H \land (\pi_2 \comp f, \pi_2 \comp g) \in Y H \}, I \times J, I' \times J')$
	\item $\dotExponential{(X, I, I')}{(Y, J, J')} = (\lambda H \in \category{C}. \{ (f, g) \mid \forall H' \in \category{C}. \forall h : H' \to H. \forall (x, y) \in X H'. (\eval \comp \tupling{f \comp h}{x}, \eval \comp \tupling{g \comp h}{y}) \in Y H' \}, \exponential{I}{J}, \exponential{I'}{J'})$
	\item $(X, I, I') \dotPlus (Y, J, J') = (\lambda H \in \category{C}. \{ (\iota_1 \comp f, \iota_1 \comp g) \mid (f, g) \in X H \} \cup \{ (\iota_2 \comp f, \iota_2 \comp g) \mid (f, g) \in Y H \}, I + J, I' + J')$
\end{itemize}
The category $\category{K}^{\top\top (P)}$ inherits the cartesian closed structure of $\category{K}$ but has a different co-cartesian structure~\cite[Theorem~4]{katsumata2008}.
\begin{itemize}
	\item $\ddot{1} = \dot{1}$
	\item $(X, I, I') \ddotTimes (Y, J, J') = (X, I, I') \dotTimes (Y, J, J')$
	\item $\ddotExponential{(X, I, I')}{(Y, J, J')} = \dotExponential{(X, I, I')}{(Y, J, J')}$
	\item $\ddot{0} = \dot{0}^{\top\top (P)}$
	\item $(X, I, I') \ddotPlus (Y, J, J') = ((X, I, I') \dotPlus (Y, J, J'))^{\top\top (P)}$
\end{itemize}
Here, we leave the inclusion functor implicit.
Note that $\top\top$-operators of $(X, I, I') \in \category{K}$ is explicitly given as follows.
\begin{align}
	(X, I, I')^{\top\top (P_{\mathrm{pure}} \rho)} &= (\lambda H \in \category{C}. \{ (x_1, x_2) \mid \forall h : H' \to H. \forall (k_1, k_2) \in \mathrm{rel}(\dotExponential{(X, I, I')}{P_{\mathrm{pure}} \rho}) H'. \\
	&\qquad \eval \comp \tupling{k_1}{x_1 \comp h} = \eval \comp \tupling{k_2}{x_2 \comp h} \}, I, I') \\
	(X, I, I')^{\top\top (C b)} &= (\lambda H \in \category{C}. \{ (x_1, x_2) \mid \forall h : H' \to H. \forall (k_1, k_2) \in \mathrm{rel}(\dotExponential{(X, I, I')}{C b}) H'. \\
	&\qquad \phi \comp \eval \comp \tupling{k_1}{x_1 \comp h} = \eval \comp \tupling{k_2}{x_2 \comp h} \}, I, I') \\
	(X, I, I')^{\top\top (P)} &= \bigwedge_{\rho \in \mathbf{GTyp}(B)} (X, I, I')^{\top\top (P_{\mathrm{pure}} \rho)} \land \bigwedge_{b \in B} (X, I, I')^{\top\top (C b)}
\end{align}

We write $\mathrm{rel}(X, I, I') \coloneqq X$ to refer to the subpresheaf part of $(X, I, I') \in \category{K}$.

By~\cite[Proposition~8]{katsumata2008}, $V b$ is closed, and we have $V : B \to \category{K}^{\top\top (P)}$.
We consider interpreting $\lambda_c$-calculus in $\category{K}^{\top\top (P)}$.

We consider the $\top\top$-lifting $(T_1 \times T_2)^{\top\top(\tupling{A}{A}, C)}$ along $q : \category{K} \to \category{C} \times \category{C}$.
By unfolding the definition of the $\top\top$-lifting, we have the following.
\begin{align}
	&(T_1 \times T_2)^{\top\top (C b)} (X, I_1, I_2) \\
	&= (\lambda H \in \category{C}. \{ (x_1, x_2) \mid \forall h : H' \to H. \forall (k_1, k_2) \in \mathrm{rel}(\dotExponential{(X, I_1, I_2)}{C \b}) H'. \\
	&\qquad \phi \comp \mu^{T_1} \comp T_1 \eval \comp \strength^{T_1} \comp \tupling{k_1}{x_1 \comp h} = \mu^{T_2} \comp T_2 \eval \comp \strength^{T_2} \comp \tupling{k_2}{x_2 \comp h} \}, T_1 I_1, T_2 I_2) \\
	&(T_1 \times T_2)^{\top\top(\tupling{A}{A}, C)} (X, I_1, I_2) \\
	&= \bigwedge_{b \in B} (T_1 \times T_2)^{\top\top ((A b, A b), C b)} (X, I_1, I_2)
\end{align}

The $\top\top$-lifting $(T_1 \times T_2)^{\top\top(\tupling{A}{A}, C)}$ is a strong monad on $\category{K}$.
We can restrict this to $\category{K}^{\top\top(P)}$.
\begin{lemma}\label{lem:TT-lifting_TT-closed}
	If $(X, I_1, I_2) \in \category{K}$ is $\top\top(P)$-closed, then so is $(T_1 \times T_2)^{\top\top(\tupling{A}{A}, C)} (X, I_1, I_2)$.
\end{lemma}
\begin{proof}
	Actually, we don't use the assumption that $(X, I_1, I_2) \in \category{K}$ is $\top\top(P)$-closed.

	To prove $((T_1 \times T_2)^{\top\top(\tupling{A}{A}, C)} (X, I_1, I_2))^{\top\top(P)} \le (T_1 \times T_2)^{\top\top(\tupling{A}{A}, C)} (X, I_1, I_2)$, it suffices to show
	\[ ((T_1 \times T_2)^{\top\top((A b, A b), C b)} (X, I_1, I_2))^{\top\top (C b)} \le (T_1 \times T_2)^{\top\top((A b, A b), C b)} (X, I_1, I_2)\]
	for each $b \in B$.
	Let $(x_1, x_2) \in \mathrm{rel}(((T_1 \times T_2)^{\top\top((A b, A b), C b)} (X, I_1, I_2))^{\top\top (C b)}) H$.
	By definition of $(T_1 \times T_2)^{\top\top((A b, A b), C b)}$ on the right-hand side, we need to show
	\begin{equation}
		\phi \comp \mu^{T_1} \comp T_1 \eval \comp \strength^{T_1} \comp \tupling{k_1}{x_1 \comp h} = \mu^{T_2} \comp T_2 \eval \comp \strength^{T_2} \comp \tupling{k_2}{x_2 \comp h} \label{eq:TT-lifting_TT-closed_proof_aim}
	\end{equation}
	for each $h : H' \to H$ and $(k_1, k_2) \in \mathrm{rel}(\dotExponential{(X, I_1, I_2)}{C b}) H'$.
	Let
	\begin{align}
		k'_1 &\coloneqq \Lambda(\mu^{T_1} \comp T_1 \eval \comp \strength^{T_1}) \comp k_1 \\
		k'_2 &\coloneqq \Lambda(\mu^{T_2} \comp T_2 \eval \comp \strength^{T_2}) \comp k_2.
	\end{align}
	It follows that $(k'_1, k'_2) \in \mathrm{rel}(\dotExponential{(T_1 \times T_2)^{\top\top((A b, A b), C b)} (X, I_1, I_2)}{C b}) H'$ because for any $h' : H'' \to H'$ and $(x'_1, x'_2) \in \mathrm{rel}((T_1 \times T_2)^{\top\top((A b, A b), C b)} (X, I_1, I_2)) H''$, we have
	\begin{align}
		&\phi \comp \eval \comp \tupling{k'_1 \comp h'}{x'_1} \\
		&= \phi \comp \mu^{T_1} \comp T_1 \eval \comp \strength^{T_1} \comp  \tupling{k_1 \comp h'}{x'_1} \\
		&= \mu^{T_2} \comp T_2 \eval \comp \strength^{T_2} \comp \tupling{k_2 \comp h'}{x'_2} \\
		&= \eval \comp \tupling{k'_2 \comp h'}{x'_2}
	\end{align}
	by applying the definition of $(T_1 \times T_2)^{\top\top((A b, A b), C b)} (X, I_1, I_2)$ to $(k_1 \comp h', k_2 \comp h') \in \mathrm{rel}(\dotExponential{(X, I_1, I_2)}{C b}) H''$.
	Thus, we have \eqref{eq:TT-lifting_TT-closed_proof_aim} as follows.
	\begin{align}
		\phi \comp \mu^{T_1} \comp T_1 \eval \comp \strength^{T_1} \comp \tupling{k_1}{x_1 \comp h} &= \phi \comp \eval \comp \tupling{k'_1}{x_1 \comp h} \\
		&= \eval \comp \tupling{k'_2}{x_2 \comp h} \\
		&= \mu^{T_2} \comp T_2 \eval \comp \strength^{T_2} \comp \tupling{k_2}{x_2 \comp h}
	\end{align}
	\qed
\end{proof}
Note that we cannot immediately obtain Lemma~\ref{lem:TT-lifting_TT-closed} by~\cite[Lemma~4.3, 4.4]{kammar2022} because $\top\top(P)$-closedness does not imply $\top\top(C)$-closedness.

The interpretation $\mathcal{V} : \mathbf{Typ}(B) \to \category{K}^{\top\top (P)}$ satisfies the following property.

\begin{lemma}\label{lem:ground_eq}
	Assume $\category{C}$ has stable finite coproducts.
	For each $\rho \in \mathbf{GTyp}(B)$, $\mathcal{V}\interpret{\rho} = \mathbf{Eq} (\mathcal{A} \interpret{\rho})$.
\end{lemma}
\begin{proof}
	By induction on $\rho$.
	The key idea is that we use the definition of $\top\top$-closure and stable finite coproducts in the cases for coproduct types.
	\begin{itemize}
		\item The base case $\rho = b \in B$ is trivial.
		\item If $\rho = 1$, then
			\[ \mathcal{V}\interpret{1} = \ddot{1} = \dot{1} = \mathbf{Eq} (\mathcal{A} \interpret{1}). \]
		\item If $\rho = \rho_1 \times \rho_2$, then $\mathcal{V}\interpret{\rho_1 \times \rho_2} = \mathcal{V}\interpret{\rho_1} \ddotTimes \mathcal{V}\interpret{\rho_2} = \mathcal{V}\interpret{\rho_1} \dotTimes \mathcal{V}\interpret{\rho_2} = \mathbf{Eq} (\mathcal{A} \interpret{\rho_1}) \dotTimes \mathbf{Eq} (\mathcal{A} \interpret{\rho_2}) = \mathbf{Eq} (\mathcal{A} \interpret{\rho_1} \times \mathcal{A} \interpret{\rho_2}) = \mathbf{Eq} (\mathcal{A} \interpret{\rho_1 \times \rho_2})$ by IH.
		\item If $\rho = 0$, then we have $\mathbf{Eq}(0) = \dot{0}^{\top\top(P)}$.
			\begin{itemize}
				\item We prove $\mathbf{Eq}(0) \le \dot{0}^{\top\top(P)}$ holds.
				\begin{itemize}
					\item First, we prove $\mathbf{Eq}(0) \le \dot{0}^{\top\top(P_{\mathrm{pure}} \rho)}$ for any $\rho \in \mathbf{GTyp}(B)$.
					For any $(f, f) \in \mathrm{rel}(\mathbf{Eq}(0)) H$, $h : H' \to H$, and $(k_1, k_2) \in \mathrm{rel} (\dotExponential{\dot{0}}{P_{\mathrm{pure}} \rho}) H'$, we have
					\[ \eval \comp \tupling{k_1}{f \comp h} = \eval \comp \tupling{k_2}{f \comp h} : H' \to \mathcal{A}\interpret{\rho} \]
					because by the strictness of initial objects in $\category{C}$, $f \comp h : H' \to 0$ is an isomorphism, which implies $H'$ is an initial object.
					\item Similarly, we have $\mathbf{Eq}(0) \le \dot{0}^{\top\top(C b)}$ for any $b \in B$.
					For any $(f, f) \in \mathrm{rel}(\mathbf{Eq}(0)) H$, $h : H' \to H$, and $(k_1, k_2) \in \mathrm{rel} (\dotExponential{\dot{0}}{C b}) H'$, we have
					\[ \phi \comp \eval \comp \tupling{k_1}{f \comp h} = \eval \comp \tupling{k_2}{f \comp h} : H' \to T_2 \mathcal{A}\interpret{\rho}. \]
				\end{itemize}
				\item We have $\dot{0}^{\top\top(P)} \le \mathbf{Eq} (0)$ because $\dot{0} \le \mathbf{Eq} (0) = V 0$ holds and $V 0$ is $\top\top(P)$-closed by \cite[Proposition~8]{katsumata2008}.
			\end{itemize}
		\item If $\rho = \rho_1 + \rho_2$, then we have
			\[\mathcal{V}\interpret{\rho_1 + \rho_2} = \mathcal{V}\interpret{\rho_1} \ddotPlus \mathcal{V}\interpret{\rho_2} = (\mathbf{Eq} (\mathcal{A} \interpret{\rho_1}) \dotPlus \mathbf{Eq} (\mathcal{A} \interpret{\rho_2}))^{\top\top(P)}\]
			by IH.
			\begin{itemize}
				\item We prove $\mathbf{Eq} (\mathcal{A} \interpret{\rho_1 + \rho_2}) \le (\mathbf{Eq} (\mathcal{A} \interpret{\rho_1}) \dotPlus \mathbf{Eq} (\mathcal{A} \interpret{\rho_2}))^{\top\top(P)}$.
				\begin{itemize}
					\item First, we prove
					\[ \mathbf{Eq} (\mathcal{A} \interpret{\rho_1 + \rho_2}) \le (\mathbf{Eq} (\mathcal{A} \interpret{\rho_1}) \dotPlus \mathbf{Eq} (\mathcal{A} \interpret{\rho_2}))^{\top\top(P_{\mathrm{pure}} \rho)} \]
					for each $\rho \in \mathbf{GTyp}(B)$.
					That is, for each $(f, f) \in \mathrm{rel}(\mathbf{Eq} (\mathcal{A} \interpret{\rho_1 + \rho_2})) H$ where $f : H \to \mathcal{A} \interpret{\rho_1} + \mathcal{A} \interpret{\rho_2}$, we prove $(f, f) \in \mathrm{rel}(\mathbf{Eq} (\mathcal{A} \interpret{\rho_1}) \dotPlus \mathbf{Eq} (\mathcal{A} \interpret{\rho_2}))^{\top\top(P_{\mathrm{pure}} \rho)} H$.
					This is because for any $h : H' \to H$ and $(k_1, k_2) \in \mathrm{rel}(\dotExponential{(\mathbf{Eq} (\mathcal{A} \interpret{\rho_1}) \dotPlus \mathbf{Eq} (\mathcal{A} \interpret{\rho_2}))}{P_{\mathrm{pure}} \rho}) H'$, we have the following equation.
					\begin{align}
						&\eval \comp \tupling{k_1}{f \comp h} \comp \cotupling{(f \comp h)^{*} \fstcoproj}{(f \comp h)^{*} \sndcoproj} \\
						&= \cotupling{\eval \comp \tupling{k_1}{f \comp h} \comp (f \comp h)^{*} \fstcoproj}{\eval \comp \tupling{k_1}{f \comp h} \comp (f \comp h)^{*} \sndcoproj} \\
						&= \cotupling{\eval \comp \tupling{k_1 \comp (f \comp h)^{*} \fstcoproj}{\fstcoproj \comp \fstcoproj^{*} (f \comp h)}}{\eval \comp \tupling{k_1 \comp (f \comp h)^{*} \sndcoproj}{\sndcoproj \comp \sndcoproj^{*} (f \comp h)}} \\
						&= \cotupling{\eval \comp \tupling{k_2 \comp (f \comp h)^{*} \fstcoproj}{\fstcoproj \comp \fstcoproj^{*} (f \comp h)}}{\eval \comp \tupling{k_2 \comp (f \comp h)^{*} \sndcoproj}{\sndcoproj \comp \sndcoproj^{*} (f \comp h)}} \label{eq:TT_coproduct} \\
						&= \eval \comp \tupling{k_2}{f \comp h} \comp \cotupling{(f \comp h)^{*} \fstcoproj}{(f \comp h)^{*} \sndcoproj}.
					\end{align}
					\begin{center}
						\begin{tikzcd}
							{\fstcoproj}^{*} H' \ar[r] \ar[d] & H' \ar[d, "f \comp h"] & {\sndcoproj}^{*} H' \ar[l] \ar[d] \\
							\mathcal{A} \interpret{\rho_1} \ar[r, "\fstcoproj"] & \mathcal{A} \interpret{\rho_1} + \mathcal{A} \interpret{\rho_2} & \mathcal{A} \interpret{\rho_2} \ar[l, swap, "\sndcoproj"]
						\end{tikzcd}
					\end{center}
					In~\eqref{eq:TT_coproduct}, we used
					\[ \eval \comp \tupling{k_1 \comp (f \comp h)^{*} \fstcoproj}{\fstcoproj \comp \fstcoproj^{*} (f \comp h)} = \eval \comp \tupling{k_2 \comp (f \comp h)^{*} \fstcoproj}{\fstcoproj \comp \fstcoproj^{*} (f \comp h)}\]
					(and a similar equation for $\sndcoproj$), which follows from the definition of $\dotExponential{}{}$ and
					\begin{itemize}
						\item $(k_1, k_2) \in \mathrm{rel}(\dotExponential{(\mathbf{Eq} (\mathcal{A} \interpret{\rho_1}) \dotPlus \mathbf{Eq} (\mathcal{A} \interpret{\rho_2}))}{P_{\mathrm{pure}} \rho}) H'$,
						\item $(f \comp h)^{*} \fstcoproj : \fstcoproj^{*} H' \to H'$, and
						\item $(\fstcoproj \comp \fstcoproj^{*} (f \comp h), \fstcoproj \comp \fstcoproj^{*} (f \comp h)) \in \mathrm{rel} ((\mathbf{Eq} (\mathcal{A} \interpret{\rho_1}) \dotPlus \mathbf{Eq} (\mathcal{A} \interpret{\rho_2}))) \fstcoproj^{*} H'$.
					\end{itemize}
					By stability, $\cotupling{(f \comp h)^{*} \fstcoproj}{(f \comp h)^{*} \sndcoproj} : \iota_1^{*} H' + \iota_2^{*} H' \to H'$ is isomorphic.
					We have $\eval \comp \tupling{k_1}{f \comp h} = \eval \comp \tupling{k_2}{f \comp h}$ and thus $(f, f) \in \mathrm{rel}((\mathbf{Eq} (\mathcal{A} \interpret{\rho_1}) \dotPlus \mathbf{Eq} (\mathcal{A} \interpret{\rho_2}))^{\top\top(V b)}) H$.
				\item Similarly, we can also prove
					\[ \mathbf{Eq} (\mathcal{A} \interpret{\rho_1 + \rho_2}) \le (\mathbf{Eq} (\mathcal{A} \interpret{\rho_1}) \dotPlus \mathbf{Eq} (\mathcal{A} \interpret{\rho_2}))^{\top\top(C b)} \]
					for each $b \in B$.
					That is, for each $f : H \to \mathcal{A} \interpret{\rho_1} + \mathcal{A} \interpret{\rho_2}$ and $\rho \in \mathbf{GTyp}(B)$, we prove $(f, f) \in \mathrm{rel}(\mathbf{Eq} (\mathcal{A} \interpret{\rho_1}) \dotPlus \mathbf{Eq} (\mathcal{A} \interpret{\rho_2}))^{\top\top(C b)} H$.
					This is because for any $h : H' \to H$ and $(k_1, k_2) \in \mathrm{rel}(\dotExponential{(\mathbf{Eq} (\mathcal{A} \interpret{\rho_1}) \dotPlus \mathbf{Eq} (\mathcal{A} \interpret{\rho_2}))}{C b}) H'$, we have the following equation.
					\begin{align}
						&\phi \comp \eval \comp \tupling{k_1}{f \comp h} \comp \cotupling{(f \comp h)^{*} \fstcoproj}{(f \comp h)^{*} \sndcoproj} \\
						&= \cotupling{\phi \comp \eval \comp \tupling{k_1}{f \comp h} \comp (f \comp h)^{*} \fstcoproj}{\phi \comp \eval \comp \tupling{k_1}{f \comp h} \comp (f \comp h)^{*} \sndcoproj} \\
						&= \cotupling{\phi \comp \eval \comp \tupling{k_1 \comp (f \comp h)^{*} \fstcoproj}{\fstcoproj \comp \fstcoproj^{*} (f \comp h)}}{\phi \comp \eval \comp \tupling{k_1 \comp (f \comp h)^{*} \sndcoproj}{\sndcoproj \comp \sndcoproj^{*} (f \comp h)}} \\
						&= \cotupling{\eval \comp \tupling{k_2 \comp (f \comp h)^{*} \fstcoproj}{\fstcoproj \comp \fstcoproj^{*} (f \comp h)}}{\eval \comp \tupling{k_2 \comp (f \comp h)^{*} \sndcoproj}{\sndcoproj \comp \sndcoproj^{*} (f \comp h)}} \\
						&= \eval \comp \tupling{k_2}{f \comp h} \comp \cotupling{(f \comp h)^{*} \fstcoproj}{(f \comp h)^{*} \sndcoproj}.
					\end{align}
				\end{itemize}
				\item We have $(\mathbf{Eq} (\mathcal{A} \interpret{\rho_1}) \dotPlus \mathbf{Eq} (\mathcal{A} \interpret{\rho_2}))^{\top\top(P)} \le \mathbf{Eq} (\mathcal{A} \interpret{\rho_1 + \rho_2})$ because
					\[ \mathbf{Eq} (\mathcal{A} \interpret{\rho_1}) \dotPlus \mathbf{Eq} (\mathcal{A} \interpret{\rho_2}) \le \mathbf{Eq} (\mathcal{A} \interpret{\rho_1 + \rho_2}) = P_{\mathrm{pure}} (\rho_1 + \rho_2) \]
					and the right-hand side is $\top\top(P)$-closed.
			\end{itemize}
	\end{itemize}
\end{proof}

\begin{proof}[Proof of Theorem~\ref{thm:relating-effect}]
	We follow the proof of~\cite[Thm~12]{katsumata2013}.
	By Lemma~\ref{lem:ground_eq}, we can interpret effect-free constants in $\category{K}^{\top\top(P)}$, that is, $\dot{a}(k) \coloneqq \mathbf{Eq}(a(k)) : \mathcal{V}\interpret{\mathrm{ar}(k)} \to \mathcal{V}\interpret{\mathrm{car}(k)}$.
	We can also define $\dot{a}(o) \in \mathbf{Alg}(\category{K}^{\top\top(P)}, \mathcal{V}\interpret{\mathrm{ar}(o)}, \mathcal{V}\interpret{\mathrm{car}(o)})$ in the same way as~\cite[Thm~12]{katsumata2013}.
	More concretely, we apply~\cite[Thm~11]{katsumata2013} to obtain $\mathbf{Gef}(\dot{a}(o)) : \mathcal{V}\interpret{\mathrm{car}(o)} \dotTo (T_1 \times T_2)^{\top\top(\tupling{A}{A}, C)} \mathcal{V}\interpret{\mathrm{ar}(o)}$ in $\category{K}$, which is actually a morphism in $\category{K}^{\top\top(P)}$ by Lemma~\ref{lem:TT-lifting_TT-closed}, and then obtain $\dot{a}(o)$ in $\category{K}^{\top\top(P)}$.

	By interpreting the term with the $\lambda_c(\Sigma)$-structure $(\category{K}^{\top\top(P)}, (T_1 \times T_2)^{\top\top(\tupling{A}{A}, C)}, V, \dot{a})$, we have $\mathcal{A} \interpret{M} \times (\phi \mathcal{A}) \interpret{M} : \mathcal{V} \interpret{x_1 : b_1, \dots, x_n : b_n} \ddotTo (T_1 \times T_2)^{\top\top(\tupling{A}{A}, C)} \mathcal{V} \interpret{b} \le C b$ in $\category{K}^{\top\top(P)}$.
	Therefore, we have $\phi \comp \mathcal{A} \interpret{M} = (\phi \mathcal{A}) \interpret{M}$.
\end{proof}

\section{Main Theorem}\label{sec:main-theorem-proof}

See Section~\ref{sec:fixing-error} for preliminaries and the outline of the proof of \cite[Thm~12]{katsumata2013}.

\subsection{Outline of the Proof}
We prove Theorem~\ref{thm:cps-is-wpt-with-recursion} in two steps.
Let $\Gamma \vdash M : \rho$ be a $\lambda_c(\Sigma)$-term and $\mathcal{A} = (\category{C}, T, A, a)$ be $\lambda_c(\Sigma)$-structure.
In the first step, we relate the interpretation $\mathcal{A}\interpret{M} : \mathcal{A} \interpret{\Gamma} \to T \mathcal{A} \interpret{\rho}$ with the interpretation $\mathcal{A}' \interpret{M} : \mathcal{A}' \interpret{\Gamma} \to \contmonad{\answerobj} \mathcal{A}' \interpret{\rho}$ where, roughly speaking, $\mathcal{A}' = (\category{C}, \contmonad{\answerobj}, A, a')$ is defined by replacing a strong monad $T$ in $\mathcal{A}$ with a continuation monad $\contmonad{\answerobj} = \exponential{(\exponential{({-})}{\answerobj})}{\answerobj}$.
Technically, this is done by considering a strong monad morphism $\phi^{\emalgsymbol} : T \to \contmonad{\answerobj}$ defined by an EM $T$-algebra $\emalgsymbol : T \answerobj \to \answerobj$ and relating $\mathcal{A} \interpret{M}$ and $\mathcal{A}' \interpret{M}$ by logical relations studied in~\cite{katsumata2013}.
In the second step, we show $\mathcal{A}^{\emalgsymbol} \interpret{\CPS{M}} = \mathcal{A}' \interpret{M}$ by adapting a result from~\cite{fuhrmann2004} to our languages.
Then, we get the equation~\eqref{eq:cps_is_wpt} by unfolding the definition of $\phi^{\emalgsymbol}$.

\subsubsection{Logical Relations}
Suppose that parameters for Theorem~\ref{thm:cps-is-wpt-with-recursion} are given.
We have a strong monad morphism $\phi^{\emalgsymbol} : T \to \contmonad{\answerobj}$ by the following proposition~\cite[Proposition~1]{hyland2007}.
\begin{proposition}\label{prop:strong-monad-morphism-em-algebra}
	Let $T$ be a strong monad on a cartesian closed category.
	There is a one-to-one correspondence between (i) EM $T$-algebras $\emalgsymbol : T \answerobj \to \answerobj$ and (ii) strong monad morphisms $\phi^{\emalgsymbol} : T \to \contmonad{\answerobj}$ to the continuation monad.
	\qed
\end{proposition}

A strong monad morphism can ``push forward'' an $\omegaCPO$-enriched $\lambda_c(\Sigma)$-structure.

\begin{definition}\label{def:lambda-c-sigma-structure-from-strong-monad-morphism}
	Given a strong monad morphism $\phi : S \to T$ and an $\omegaCPO$-enriched $\lambda_c(\Sigma)$-structure $\mathcal{A} = (\category{C}, S, A, a)$, we define an $\omegaCPO$-enriched $\lambda_c(\Sigma)$-structure $\phi \mathcal{A}$ by the tuple $(\category{C}, T, A, \phi(a))$ where
	$\phi(a)(c) = a(c)$ and
	$\phi(a)(o) = \mathbf{Alg}(\phi, \mathcal{A}\interpret{\mathrm{ar}(k)}, \mathcal{A}\interpret{\mathrm{car}(k)}) (a(o))$
	for each $c \in K$ and $o \in O$.
	Here, $\mathbf{Alg}(\phi, D, C) : \mathbf{Alg}(S, D, C) \to \mathbf{Alg}(T, D, C)$ is defined by the mapping $(e : C \to S D) \mapsto (\phi \comp e : C \to T D)$ of generic effects and a bijection between algebraic operations and generic effects.
\end{definition}

Let $\Gamma \vdash M : \rho$ be a $\lambda_c(\Sigma)$-term.
Now, we show that $(\phi^{\emalgsymbol} \mathcal{A}) \interpret{M} : (\phi^{\emalgsymbol} \mathcal{A}) \interpret{\Gamma} \to \contmonad{\answerobj} (\phi^{\emalgsymbol} \mathcal{A}) \interpret{\rho}$ is equal to the composite of $\mathcal{A}\interpret{M} : \mathcal{A} \interpret{\Gamma} \to T \mathcal{A} \interpret{\rho}$ and the strong monad morphism $\phi^{\emalgsymbol} : T \to \contmonad{\answerobj}$.
This problem is studied in~\cite[Theorem~12]{katsumata2013} using logical relations and $\top\top$-lifting, but two changes should be made here.
Firstly, his result is limited to the case where types in the context $\Gamma$ and the type $\rho$ are base types, while we want to allow ground types in order to make our result as general as possible.
This is done by a mild extension of his proof.
Secondly and more importantly, his proof contains an subtle error in the treatment of coproduct types.
We correct this by adding an assumption about stable coproducts.
This correction will appear in another paper.

\begin{proposition}\label{prop:monad-morphism-ground-type}
	Let $\mathcal{A} = (\category{C}, S, A, a)$ be an $\omegaCPO$-enriched $\lambda_c(\Sigma)$-structure such that $\category{C}$ is stable.
	Given a strong monad morphism $\phi : S \to T$, for each well-typed $\lambda_c$-term $x_1 : \rho_1, \dots, x_n : \rho_n \vdash M : \rho$ where $\rho_1, \dots, \rho_n, \rho$ are ground types, we have
	\[ \phi \comp \mathcal{A}\interpret{x_1 : \rho_1, \dots, x_n : \rho_n \vdash M : \rho} = (\phi \mathcal{A}) \interpret{x_1 : \rho_1, \dots, x_n : \rho_n \vdash M : \rho} . \]
\end{proposition}
\begin{proof}
	As for the extension to ground types, see Section~\ref{sec:fixing-error}.
	The stability condition is required to correct the error in the proof of~\cite[Theorem~12]{katsumata2013} (see Section~\ref{sec:main-theorem-proof}).
\end{proof}

\subsubsection{CPS and Continuation Monads}
We show $\mathcal{A}^{\emalgsymbol} \interpret{\CPS{M}} = (\phi^{\emalgsymbol} \mathcal{A})\interpret{M}$ by extending a result from~\cite{fuhrmann2004}.
The proof is tedious but rather straightforward.
\begin{proposition}\label{prop:cps-continuation-monad}
	For each type $\rho$, there exists a (canonical) family of isomorphism $\psi_{\rho} : \mathcal{A}^{\emalgsymbol} \interpret{\CPS{\rho}} \to (\phi^{\emalgsymbol} \mathcal{A}) \interpret{\rho}$ such that for any well-typed $\lambda_c$-term $\Gamma \vdash M : \rho$, we have
	\[ (\phi^{\emalgsymbol} \mathcal{A}) \interpret{M} \comp \psi_{\Gamma} = \contmonad{\answerobj} \psi_{\rho} \comp \mathcal{A}^{\emalgsymbol} \interpret{\CPS{M}} \]
	where $\psi_{\Gamma} : \mathcal{A}^{\emalgsymbol} \interpret{\CPS{\Gamma}} \to (\phi^{\emalgsymbol} \mathcal{A}) \interpret{\Gamma}$ is an extension of $\psi_{\rho}$ to the context $\Gamma$.
	Moreover, $\psi_{\rho}$ is the identity if $\rho$ is a ground type.
	\qed
\end{proposition}
\begin{proof}
	By induction on $M$.
	See Section~\ref{sec:main-theorem-proof} for details.
\end{proof}

\begin{proof}[Proof of Theorem~\ref{thm:cps-is-wpt-with-recursion}]
	We have
	$\mathcal{A}^{\emalgsymbol} \interpret{\CPS{M}} = (\phi^{\emalgsymbol} \mathcal{A}) \interpret{M} = \phi^{\emalgsymbol} \comp \mathcal{A} \interpret{M}$
	by Proposition~\ref{prop:monad-morphism-ground-type},\ref{prop:cps-continuation-monad}
	where $\phi^{\emalgsymbol}$ is a strong monad morphism in Proposition~\ref{prop:strong-monad-morphism-em-algebra}.
	By unfolding definitions, we get~\eqref{eq:cps_is_wpt}.
\end{proof}

\subsection{Proofs}

\begin{proof}[Proof of Proposition~\ref{prop:monad-morphism-ground-type}]
	The key idea for extending \cite[Thm~12]{katsumata2013} to ground types is to use the following definition for a parameter for $\top\top$-lifting.
	\begin{align}
		V b &\coloneqq \mathbf{Eq} \comp A = (\lambda H \in \category{C}. \{ (f, f) \mid f : H \to A b \}, A b, A b) &&\in \category{K} \\
		C \rho &\coloneqq (\lambda H \in \category{C}. \{ (f, \phi \comp f) \mid f : H \to T_1 (\mathcal{A} \interpret{\rho}) \}, T_1 (\mathcal{A} \interpret{\rho}), T_2 (\mathcal{A} \interpret{\rho})) &&\in \category{K} \label{eq:ground-type-extension}
	\end{align}
	where $b \in B$ and $\rho \in \mathbf{GTyp}(B)$.
	If we restrict the domain of $C : \mathbf{GTyp}(B) \to \category{K}$ to base types $B$, then \eqref{eq:ground-type-extension} gives the same parameter as~\cite[Thm~12]{katsumata2013}.
	Since $(T_1 \times T_2) \tupling{\mathcal{A}\interpret{-}}{\mathcal{A}\interpret{-}} = q \comp C : \mathbf{GTyp}(B) \to \category{C} \times \category{C}$, we have a $\top\top$-lifting of $T_1 \times T_2$ defined by the parameter $(\tupling{\mathcal{A}\interpret{-}}{\mathcal{A}\interpret{-}}, C)$.
	Here, note that $\mathcal{A}\interpret{\rho} = (\phi \mathcal{A})\interpret{\rho}$ for any ground type $\rho \in \mathbf{GTyp}(B)$.

	Using this extension, we can do the same argument as Section~\ref{sec:fixing-error}.
	Specifically, we define a $\lambda_c(\Sigma)$-structure $\mathcal{V} = (\category{K}^{\top\top(P)}, (T_1 \times T_2)^{(\tupling{\mathcal{A}\interpret{-}}{\mathcal{A}\interpret{-}}, C)}, \dots)$ where $\category{K}^{\top\top(P)}$ is the full reflective subcategory of $\top\top$-closed objects in $\category{K}$ and the parameter $P$ is also extended in accordance with $C$.
	We interpret a well-typed $\lambda_c(\Sigma)$-term $x_1 : \rho_1, \dots, x_n : \rho_n \vdash M : \rho$ using $\mathcal{V}$, and if $\rho_1, \dots, \rho_n, \rho$ are ground types, then we get $\mathcal{A} \interpret{M} \times (\phi \mathcal{A}) \interpret{M} : \mathbf{Eq} (\mathcal{A} \interpret{x_1 : \rho_1, \dots, x_n : \rho_n}) \dotTo C \mathcal{A} \interpret{\rho}$ as the interpretation.
	Therefore, we have $\phi \comp \mathcal{A} \interpret{M} = (\phi \mathcal{A}) \interpret{M}$
\end{proof}

\begin{proof}[Proof of Proposition~\ref{prop:cps-continuation-monad}]
	We define $\psi_{\rho}$ as follows.
	\begin{gather}
		\psi_b = \identity{} \qquad
		\psi_1 = \identity{} \qquad
		\psi_0 = \identity{} \qquad
		\psi_{\rho_1 \times \rho_2} = \psi_{\rho_1} \times \psi_{\rho_2} \qquad
		\psi_{\rho_1 + \rho_n} = \psi_{\rho_1} + \psi_{\rho_2} \\
		\psi_{\rho \to \tau} = (\exponential{\psi^{-1}_{\rho}}{\contmonad{\answerobj} \psi_{\tau}}) \comp \Lambda (\Lambda (\eval \comp \associator))
	\end{gather}
	Note $\psi_{\rho \to \tau} : \exponential{\mathcal{A}^{\emalgsymbol}\interpret{\CPS{\rho}} \times (\exponential{\mathcal{A}^{\emalgsymbol}\interpret{\CPS{\tau}}}{\answerobj})}{\answerobj} \to \exponential{(\phi^{\emalgsymbol} \mathcal{A})\interpret{\rho}}{\exponential{(\exponential{(\phi^{\emalgsymbol} \mathcal{A})\interpret{\tau}}{\answerobj})}{\answerobj}}$
	It is obvious that $\psi_{\rho}$ is the identity if $\rho$ is a ground type.

	The rest of the proof is done by induction on $M$.
	The proof is tedious but rather straightforward.
	\begin{itemize}
		\item In many cases, we must be careful about the use of weakening, which is often implicit in the definition of the CPS transformation.
		Lemma~\ref{lem:sem_exchange_target},\ref{lem:sem_weakening_target} are useful when dealing with the interpretation of terms with unused variables.
		For example, consider the case for the CPS of effect-free constant.
		Recall that for a well-typed $\lambda_c(\Sigma)$-term $\Gamma \vdash M : \mathrm{ar}(c)$, we have
		\[ \CPS{(c\ M)} = \lambda k. \CPS{M}\ (\lambda m. k\ (c\ m)). \]
		Now, we prove
		\begin{equation}
			(\phi^{\emalgsymbol} \mathcal{A}) \interpret{c\ M} \comp \psi_{\Gamma} = C_{\answerobj} \psi_{\mathrm{car}(c)} \comp \mathcal{A}^{\emalgsymbol} \interpret{\lambda k. \CPS{M}\ (\lambda m. k\ (c\ m))}.
			\label{eq:cps-continuation-monad-proof-const}
		\end{equation}
		For the left-hand side, we have the following.
		\begin{align}
			(\phi^{\emalgsymbol} \mathcal{A}) \interpret{c\ M} \comp \psi_{\Gamma}
			&= C_{\answerobj} a(c) \comp (\phi^{\emalgsymbol} \mathcal{A}) \interpret{M} \comp \psi_{\Gamma} & \text{by Definition~\ref{def:interpretation-lambda-c-calculus}} \\
			&= C_{\answerobj} a(c) \comp \psi_{\mathrm{ar}(c)} \comp \mathcal{A}^{\emalgsymbol} \interpret{\CPS{M}} & \text{by induction hypothesis}
		\end{align}
		Note that $\CPS{M}$ in the right-hand side of~\eqref{eq:cps-continuation-monad-proof-const} is weakened by a variable $k$ ($\CPS{\Gamma}, k : \mathrm{car}(c) \to \answertype \vdash \CPS{M} : (\mathrm{ar}(c) \to \answertype) \to \answertype$) whereas $\CPS{M}$ in the inductive hypothesis is not ($\CPS{\Gamma} \vdash \CPS{M} : (\mathrm{ar}(c) \to \answertype) \to \answertype$).
		Taking this into account, we have the following.
		\begin{align}
			&\mathcal{A}^{\emalgsymbol} \interpret{\lambda k. \CPS{M}\ (\lambda m. k\ (c\ m))} \\
			&= \Lambda (\eval \comp \tupling{\mathcal{A}^{\emalgsymbol}\interpret{\Gamma, k : (\dots) \vdash \CPS{M} : (\dots)}}{\mathcal{A}^{\emalgsymbol}\interpret{\lambda m. k\ (c\ m)}}) & \text{by Definition~\ref{def:interpretation-response-calculus}} \\
			&= \Lambda (\eval \comp \tupling{\mathcal{A}^{\emalgsymbol}\interpret{\CPS{M}} \comp \pi_1}{\mathcal{A}^{\emalgsymbol}\interpret{\lambda m. k\ (c\ m)}}) & \text{by Lemma~\ref{lem:sem_weakening_target}} \\
			&= \Lambda (\eval \comp \tupling{\mathcal{A}^{\emalgsymbol}\interpret{\CPS{M}} \comp \pi_1}{\Lambda(\eval \comp \tupling{\pi_2 \comp \pi_1}{a(c) \comp \pi_2})})
		\end{align}
		Since $\psi_{\mathrm{ar}(c)}$ and $\psi_{\mathrm{car}(c)}$ are identities, it suffices to prove the following equation.
		\[ C_{\answerobj} a(c) \comp \mathcal{A}^{\emalgsymbol} \interpret{\CPS{M}} = \Lambda (\eval \comp \tupling{\mathcal{A}^{\emalgsymbol}\interpret{\CPS{M}} \comp \pi_1}{\Lambda(\eval \comp \tupling{\pi_2 \comp \pi_1}{a(c) \comp \pi_2})}) \]
		This follows from the axioms for cartesian closed categories.
		\item Sometimes, we need to handle substitution $\mathcal{A}^{\emalgsymbol} \interpret{M[N/x]}$, in which case Lemma~\ref{lem:sem_subst_target} is useful.
		\item For the case of recursion, Lemma~\ref{lem:fixedpoint_naturality},\ref{lem:parameterized_uniformity_simplified} are useful.
	\end{itemize}
\end{proof}

\begin{lemma}[exchange]\label{lem:sem_exchange_target}
	For each well-typed $\lambda_{\mathrm{HFL}}$-term $\Gamma, x : \tau_1, y : \tau_2, \Delta \vdash M : \rho$, we have
	\[ \mathcal{A}^{\emalgsymbol} \interpret{\Gamma, x : \tau_1, y : \tau_2, \Delta \vdash M : \rho} = \mathcal{A}^{\emalgsymbol} \interpret{\Gamma, y : \tau_2, x : \tau_1, \Delta \vdash M : \rho} \comp \psi_{\Gamma; x : \tau_1, y : \tau_2; \Delta} \]
	where $\psi_{\Gamma; x : \tau_1; y : \tau_2; \Delta} : \mathcal{A}^{\emalgsymbol}\interpret{\Gamma, x : \tau_1, y : \tau_2, \Delta} \to \mathcal{A}^{\emalgsymbol}\interpret{\Gamma, y : \tau_2, x : \tau_1, \Delta}$ is an isomorphism defined by
	\begin{gather}
		\psi_{\Gamma; x : \tau_1; y : \tau_2; \cdot} \coloneqq \associator^{-1} \comp (\identity{} \times \braiding) \comp \associator \qquad
		\psi_{\Gamma; x : \tau_1; y : \tau_2; \Delta, z : \tau} \coloneqq \psi_{\Gamma; x : \tau_1; y : \tau_2; \Delta} \times \identity{}.
	\end{gather}
	\qed
\end{lemma}

\begin{lemma}[weakening]\label{lem:sem_weakening_target}
	For each well-typed $\lambda_{\mathrm{HFL}}$-term $\Gamma \vdash M : \rho$, we have
	\[ \mathcal{A}^{\emalgsymbol} \interpret{\Gamma, x : \tau \vdash M : \rho} = \mathcal{A}^{\emalgsymbol} \interpret{\Gamma \vdash M : \rho} \comp \pi_1. \]
	\qed
\end{lemma}
\begin{lemma}[substitution]\label{lem:sem_subst_target}
	For each well-typed $\lambda_{\mathrm{HFL}}$-term $\Gamma, x : \tau \vdash M : \rho$ and $\Gamma \vdash N : \tau$, we have
	\[ \mathcal{A}^{\emalgsymbol} \interpret{M[N/x]} = \mathcal{A}^{\emalgsymbol} \interpret{M} \comp \langle \identity{}, \mathcal{A}^{\emalgsymbol} \interpret{N} \rangle. \]
	\qed
\end{lemma}

\begin{lemma}[naturality]\label{lem:fixedpoint_naturality}
	A parameterized uniform fixed-point operator for $T$-algebras is natural.
	That is, we have $f^{\dagger} \comp g = (f \comp (g \times \identity{}))^{\dagger}$ for each $T$-algebra $\alpha : T A \to A$, $g : X \to Y$, and $f : Y \times A \to A$.
	\qed
\end{lemma}

\begin{lemma}[simplified parameterized uniformity]\label{lem:parameterized_uniformity_simplified}
	Let $({-})^{\dagger}$ be a parameterized uniform fixed-point operator.
	Let $\alpha : T A \to A$ and $\beta : T B \to B$ be $T$-algebras.
	Let $f : X \times A \to A$ and $g : X \times B \to B$ be morphisms.
	For any $T$-algebra morphism $h : A \to B$ such that $g \comp (\identity{} \times h) = h \comp f$, we have $g^{\dagger} = h \comp f^{\dagger}$.
	\qed
\end{lemma}

\section{Details of Instances}\label{sec:detail-instance}

\subsection{Trace Property and May/Must-Reachability}\label{subsec:free-algebra-nondet}
For trace properties and may/must-reachability (Section~\ref{subsec:may-must-reachability}), we need technical preparation.
Here, we want to obtain the set of sequences of events from the free $\mathcal{T}_P$-algebra $T^P X$ generated by $X$.
However, it is tricky to give a concrete construction of $T^P X$ in general.
For simplicity, we assume $X = 1$ and consider two concrete constructions of non-free $\mathcal{T}_P$-algebras $H 1$ and $S 1$.

Given $X \in \omegaCPO$, we define $\mathbf{Seq} X \coloneqq ((E^{*} + E^{\omega}) \times \{ \bot \}) + (E^{*} \times X)$ where $E^{*}$ and $E^{\omega}$ are the set of finite and infinite sequences of events, respectively.
For any $s, t \in E^{*} + E^{\omega}$, we write $s \sqsubseteq t$ if $s$ is a prefix of $t$.
We define a partial order ${\le}_{\mathbf{Seq} X}$ as follows.
For each $(s, x), (t, y) \in \mathbf{Seq} X$, $(s, x) \le_{\mathbf{Seq} X} (t, y)$ is true if one of the following conditions holds: (i) $x = y = \bot$ and $s \sqsubseteq t$ (ii) $x = \bot$, $y \in X$, and $s \sqsubseteq t$ (iii) $s = t$, $x \in X$, $y \in X$, and $x \le y$.
Then, $\mathbf{Seq} X$ is an $\omega$cpo and moreover characterised as a free algebra of the algebraic theory $\mathcal{T}_{\mathrm{St}}$ defined below.
\begin{lemma}
	Let $\mathcal{T}_{\mathrm{St}}$ be an algebraic theory defined by a unary operation $\operation{event}_a$ for each $a \in E$ and a nullary operation $\bot$ together with an inequation $\bot \le x$.
	The $\omega$cpo $\mathbf{Seq} X$ is a free $\mathcal{T}_{\mathrm{St}}$-algebra generated by $X$.
	The unit $\eta_X : X \to \mathbf{Seq} X$ is given by $\eta_X(x) = (\varepsilon, x)$, and the algebra structure is given by
	$\bot^{\mathbf{Seq} X} = (\varepsilon, \bot)$ and
	$\operation{event}_a^{\mathbf{Seq} X} (s, x) = (a \cdot s, x)$
	where $\varepsilon \in E^{*}$ is the empty sequence and $({\cdot})$ is the concatenation of two sequences.
\end{lemma}

Let $\mathcal{T}_H$ and $\mathcal{T}_S$ be algebraic theories defined by adding $x \join y \ge x$ and $x \join y \le x$, respectively, to $\mathcal{T}_P$ defined in Example~\ref{ex:trace-property-semantics}.
We define a $\mathcal{T}_H$-algebra $H 1$ as follows.
\[ H 1 \coloneqq \{ Y \subseteq \mathbf{Seq} 1 \setminus (E^{\omega} \times \{ \bot \}) \mid \text{$Y$ is nonempty and downward-closed} \} \]
This is an $\omega$cpo ordered by the inclusion order ${\subseteq}$ and has a $\mathcal{T}_H$-algebra structure.
\[ \bot^{H 1} = \{(\varepsilon, \bot)\} \qquad \operation{event}_a^{H 1}(Y) = \mathop{\downarrow} \{ (a \cdot s, x) \mid (s, x) \in Y \} \qquad Y_1 \join^{H 1} Y_2 = Y_1 \cup Y_2 \]
Here, $\mathop{\downarrow} Y$ is the downward closure of $Y$.
In fact, $H 1$ is a free $\mathcal{T}_H$-algebra generated by $1$ and the unit is given by $\eta^H_1 (\star) = \mathop{\downarrow} \{(\varepsilon, \star)\}$ for $\star \in 1$.
We also define a $\mathcal{T}_S$-algebra $S 1$ as follows.
\[ S 1 \coloneqq \{ Y \subseteq \mathbf{Seq} 1 \mid \text{$Y$ is flat, closed, and nonempty} \} \]
Here, we say $Y \subseteq \mathbf{Seq} 1$ is \emph{flat} if any two different elements in $Y$ are incomparable (i.e. for any $(s, x), (t, y) \in Y$, if $(s, x) \le_{\mathbf{Seq} 1} (t, y)$, then $(s, x) = (t, y)$) and \emph{closed} if for any infinite sequence of events $a_1, a_2, \dots \in E$, if for any $n$, there exists $(t_n, x_n) \in Y$ such that $(a_1 \dots a_n, \bot) \le_{\mathbf{Seq} 1} (t_n, x_n)$, then $(a_1 a_2 \dots, \bot) \in Y$.
A partial order ${\le}_{\mathrm{EM}}$ on $S 1$ is defined by the Egli--Milner order: $Y_1 {\le}_{\mathrm{EM}} Y_2$ if and only if $\forall (t, y) \in Y_2, \exists (s, x) \in Y_1, (s, x) \le_{\mathbf{Seq} 1} (t, y)$.
Then, $S 1$ is an $\omega$cpo as proved in~\cite{meyer1988}.
We can also prove that $S 1$ is a free $\mathcal{T}_S$-algebra generated by $1$ where the unit is given by $\eta^{S}_1(\star) = \{ (\varepsilon, \star) \}$.
The $\mathcal{T}_S$-algebra structure is given by
\[ \bot^{S 1} = \{ (\varepsilon, \bot) \} \quad \operation{event}_a^{S 1} Y = \{ (a \cdot s, x) \mid (s, x) \in Y \} \quad Y_1 \join^{S 1} Y_2 = \min (Y_1 \cup Y_2) \]
where $\min Y$ is the set of minimal elements in $Y$.
Since any $\mathcal{T}_H$-algebra ($\mathcal{T}_S$-algebra) is a $\mathcal{T}_P$-algebra, we have a morphism of $\mathcal{T}_H$-algebras $h^H : T^P 1 \to H 1$ ($\mathcal{T}_S$-algebras $h^S : T^P 1 \to S 1$) such that $\eta^H_1 = h^H \comp \eta^{T^P}_1$ ($\eta^S_1 = h^S \comp \eta^{T^P}_1$).
These morphisms enable us to extract information about sequences of events without knowing a concrete construction of $T^P 1$.
Intuitively, $h^H (Y)$ gives the prefix-closure of $Y \in T^P 1$, and $h^S(Y)$ gives the set of minimal sequences in $Y$.
\begin{example}[trace property (detailed)]\label{ex:trace-property-wp-detailed}
	We define an EM $T^P$-algebra $\emalgsymbol_{\mathrm{tr}} : T^P \answerobj \to \answerobj$ where $T^P$ is a monad defined in Example~\ref{ex:trace-property-semantics} and then explain that trace properties can be expressed as weakest preconditions for $\emalgsymbol_{\mathrm{tr}}$.
	Let $\mathfrak{A}$ be a deterministic finite automaton $(U, \delta, q_0, F)$ where $U$ is a finite set of states, $\delta \subseteq U \times E \times U$ is a transition relation, $q_0 \in U$ is an initial state, and $F$ is a set of final states.
	Here, we say $\mathfrak{A}$ is \emph{deterministic} if for any $q \in U$ and $a \in E$, there is at most one $q' \in U$ such that $(q, a, q') \in \delta$.
	We also assume that all states are final states $U = F$.
	We write $q \xrightarrow{a} q'$ if $(q, a, q') \in \delta$.
	The language accepted by $\mathfrak{A}$ is denoted by $L(\mathfrak{A})$.

	Now, consider an $\omega$cpo $\answerobj = (2^U, \supseteq)$ (note that the inclusion order is reversed here).
	This means that each truth value $Q \in 2^U$ assigns true or false to each state of $\mathfrak{A}$.
	We define a $\mathcal{T}_P$-algebra on $\answerobj$ as follows.
	\begin{gather}
		\nontermconst^{\answerobj} \coloneqq U \quad
		x \join^{\answerobj} y \coloneqq x \cap y \quad
		\operation{event}^{\answerobj}_a(x) \coloneqq \langle a \rangle x \coloneqq \{ q \in U \mid \exists q' \in x, q \xrightarrow{a} q' \}
		\label{eq:trace-algebra-detail}
	\end{gather}
	Note that operations defined in~\eqref{eq:trace-algebra-detail} are Scott-continuous.
	Note also that $\operation{event}_a(x \join y) = \operation{event}_a(x) \join \operation{event}_a(y)$ holds because we assumed that $(U, \delta)$ is deterministic.
	This $\mathcal{T}_P$-algebra defines an EM $T^P$-algebra $\emalgsymbol_{\mathrm{tr}} : T^P \answerobj \to \answerobj$.

	The weakest precondition transformer defined by $\emalgsymbol_{\mathrm{tr}}$ corresponds to trace properties for the automaton $\mathfrak{A}$.
	For simplicity, consider a morphism $f : 1 \to T^P 1$ that represents a program whose input and output are the unit type.
	In this situation, we can regard $\mathrm{wp}^{\emalgsymbol_{\mathrm{tr}}}[f]$ as a function of type $\answerobj \to \answerobj$ by identifying $\answerobj \cong \omegaCPO(1, \answerobj)$.
	Let $Q \in \answerobj \cong \omegaCPO(1, \answerobj)$ be a postcondition.
	To relate $\mathrm{wp}^{\emalgsymbol_{\mathrm{tr}}}[f](Q)$ with the trace property for $f$, we use $H 1$ because \eqref{eq:trace-algebra-detail} is actually a $\mathcal{T}_H$-algebra.
	By freeness, we have a unique $\mathcal{T}_H$-algebra morphism $\hat{Q} : H 1 \to \answerobj$ such that $Q = \hat{Q} \comp \eta^{H}_1$.
	The morphism $\hat{Q}$ is given by
	$\hat{Q} Y = \bigcap_{(s, \bot) \in Y \cap (E^{*} \times \{ \bot \})} \langle s \rangle \bot^{\answerobj} \cap \bigcap_{(s, \star) \in Y \cap (E^{*} \times 1)} \langle s \rangle Q$
	where $\langle a_1 \dots a_n \rangle x \coloneqq \langle a_1 \rangle \dots \langle a_n \rangle x$ is a shorthand notation for a sequence of events.
	Now, recall that $\mathrm{wp}^{\emalgsymbol_{\mathrm{tr}}}[f](Q)$ is defined by $\emalgsymbol_{\mathrm{tr}} \comp T^P Q \comp f$.
	Note that $\emalgsymbol_{\mathrm{tr}} \comp T^P Q : T^P 1 \to \answerobj$ is a morphism of EM algebras from $\mu^{T^P}_1 : T^P (T^P 1) \to T^P 1$ to $\emalgsymbol_{\mathrm{tr}} : T^P \answerobj \to \answerobj$ such that $Q = \emalgsymbol_{\mathrm{tr}} \comp T^P Q \comp \eta^{T^P}_1$.
	Since we also have $Q = \hat{Q} \comp \eta^H_1 = \hat{Q} \comp h^H \comp \eta^{T^P}_1$, we get $\hat{Q} \comp h^H = \emalgsymbol_{\mathrm{tr}} \comp T^P Q$ by the universal property of the free algebra $T^P 1$.
	Now, let $Q = U \in \answerobj$ be the set of all states.
	Since $s \in L(\mathfrak{A})$ if and only if $q_0 \in \langle s \rangle U$, we have
	\begin{equation}
		q_0 \in \mathrm{wp}^{\emalgsymbol_{\mathrm{tr}}}[f](U) = \bigcap_{(s, x) \in h^H(f(\star))} \langle s \rangle U \quad\iff\quad \forall (s, x) \in h^H(f(\star)), s \in L(\mathfrak{A}).
		\label{eq:wp-trace-iff-detail}
	\end{equation}
	Since $\{ s \mid (s, x) \in h^H(f(\star)) \}$ is (the prefix closure of) the set of sequences of events output by $f$, we can rephrase \eqref{eq:wp-trace-iff-detail} as ``the trace property for $f : 1 \to T^P 1$ is true if and only if $\mathrm{wp}^{\emalgsymbol_{\mathrm{tr}}}[f](U)$ is true at the initial state $q_0$ of the given automaton $\mathfrak{A}$''.
	Intuitively, $\mathrm{wp}^{\emalgsymbol_{\mathrm{tr}}}[f]$ takes a set of ``post-states'', runs the given automaton $\mathfrak{A}$ backwards, and returns the set of ``pre-states'' such that for any pre-state and any output string, there exists a run of $\mathfrak{A}$ that finishes at a post-state.
\end{example}

\subsection{Expected Cost and Cost Moment}
Distributive laws for Example~\ref{ex:expected-cost-cost-moment-semantics} are given as follows.
First, we have a distributive law between $P$ and $\mathbb{W} \times ({-})$.
\begin{lemma}
	Let $\category{C}$ be a symmetric monoidal category and $M$ be a monoid in $\category{C}$.
	Then, for any strong monad $T$, the strength $\strength^T$ of $T$ gives a distributive law $\strength^T_{M, {-}} : M \otimes T {-} \to T (M \otimes {-})$ between strong monads.
	\qed
\end{lemma}
Then, we consider a distributive law between $P(\mathbb{W} \times ({-}))$ and $({-})_{\bot}$.
\begin{lemma}
	Let $T$ be a strong monad on $\omegaQBS$ and assume
	\begin{equation}
		\text{for any $x \in T X$,} \qquad \eta^{T}(\bot) \le T \eta^{({-})_{\bot}}(x) \quad \in T(X_{\bot})
		\label{eq:omegaqbs-lift-distributive-law}
	\end{equation}
	where $X \in \omegaQBS$.
	We have a distributive law $d : (T({-}))_{\bot} \to T(({-})_{\bot})$ between strong monads.
\end{lemma}
\begin{proof}
	For each $X \in \omegaQBS$, we define a function $|d_X| : |T X| + 1 \to |T(X_{\bot})|$ by $|d_X| = [|T \eta^{({-})_{\bot}}_X|, |\eta^{T}| \comp \iota_2]$.
	Since $d_X(\bot) = \eta^T(\bot) \le T \eta^{({-})_{\bot}}(x) = d_X(x)$ for any $x \in T X$, $|d_X|$ is Scott-continuous.
	It is straightforward to check naturality and axioms of distributive laws.
\end{proof}
The strong monad $P(\mathbb{W} \times ({-}))$ satisfies \eqref{eq:omegaqbs-lift-distributive-law}.
In fact, $P(\mathbb{W} \times ({-}))$ has a stronger property.
\begin{lemma}
	Let $X$ be an $\omega$qbs with a bottom element $\bot_X \in X$.
	Then, $\eta^{P(\mathbb{W} \times ({-}))}(\bot_X)$ is a bottom element in $P(\mathbb{W} \times X)$.
\end{lemma}
\begin{proof}
	Recall that $P(\mathbb{W} \times X)$ is a sub-$\omega$cpo of $\exponential{(\exponential{\mathbb{W} \times X}{\mathbb{W}})}{\mathbb{W}}$ and that $f \in P(\mathbb{W} \times X)$ is linear and satisfies $f(1) = 1$.
	We prove $\eta^{P(\mathbb{W} \times ({-}))}(\bot_X) \le f$ for any $f \in P(\mathbb{W} \times X)$.
	It suffices to prove that for any $w : \mathbb{W} \times X \to \mathbb{W}$, we have $\eta^{P(\mathbb{W} \times ({-}))}(\bot_X)(w) \le f(w)$.
	\[ \eta^{P(\mathbb{W} \times ({-}))}(\bot_X)(w) = w(0, \bot_X)  = w(0, \bot_X) \cdot f(1) = f(w(0, \bot_X) \cdot 1) \le f(w) \]
\end{proof}

\section{More Instances}\label{sec:more-instances}
\subsection{May/Must-Reachability}\label{subsec:may-must-reachability}
\subsubsection{Informal Introduction}
May/must-reachability are also studied in~\cite{kobayashi2018} as well as the trace property in Section~\ref{subsec:trace-property}.
For example, consider may/must-reachability about the $\mathrm{close}$ event in \eqref{eq:example-program-trace-property}, that is, whether the program may/must reach $\operation{event}_{\mathrm{close}}$.
Since we are interested in may/must-reachability for $\operation{event}_{\mathrm{close}}$, all the other events are irrelevant for this problem and can be removed for simplicity.
\begin{equation}
	\letrec{f}{x}{\operation{event}_{\mathrm{close}}(x) \mathrel{\square} f\ x}{f\ ()} \label{eq:example-program-may-must}
\end{equation}
For may-reachability, we want to verify whether there exists a non-empty sequence of events.
This is equivalent to the negation of the trace property where the specification is given by an automaton $\mathfrak{A}_0$ that only accepts the empty sequence, i.e., $\mathfrak{A}_0$ is the automaton with only one state $q_0$ and no transition.
We get an HFL formula for may-reachability by applying the same CPS transformation as Section~\ref{subsec:trace-property} and then taking the de Morgan dual.
\begin{equation}
	\letrec{f'}{x\ k}{ [\mathrm{close}] (k\ x) \lor (f'\ x\ k)}{f\ ()\ (\lambda r. \mathbf{false})}
	\label{eq:example-cps-may-reachability}
\end{equation}
Here, $[\mathrm{close}]$ is the dual modal operator of $\langle \mathrm{close} \rangle$, and $\mathbf{let}\ \mathbf{rec}$ is interpreted as the least fixed point since we take the dual.
We can simplify~\eqref{eq:example-cps-may-reachability} by (1) replacing the modal operator $[\mathrm{close}]({-})$ with $\mathbf{true}$ because $\mathfrak{A}_0$ has no transition and (2) defining $F \coloneqq f'\ ()\ (\lambda r. \mathbf{false})$.
Then, we get an HFL formula $\mu F. \mathbf{true} \lor F$, which is the same as what \cite{kobayashi2018} gives.

For must-reachability, we want to verify whether all sequences are non-empty.
We apply a different CPS transformation ($\square \mapsto \land$, $\operation{event}_a \mapsto [a]$, and $\mathbf{let}\ \mathbf{rec} \mapsto \text{lfp}$) from may-reachability ($\square \mapsto \lor$, $\operation{event}_a \mapsto [a]$, and $\mathbf{let}\ \mathbf{rec} \mapsto \text{lfp}$, the dual of trace properties), which reflects the difference between angelic/demonic nondeterminism.
By passing $\lambda r. \mathbf{false}$ as a postcondition, we get the following.
\[ \letrec{f'}{x\ k}{[\mathrm{close}] (k\ x) \land (f'\ x\ k)}{f'\ ()\ (\lambda r. \mathbf{false})} \]
Similarly to the may-reachability, we can replace $[\mathrm{close}]({-})$ with $\mathbf{true}$.
By defining $F \coloneqq f'\ ()\ (\lambda r. \mathbf{false})$, we get an HFL formula $\mu F. \mathbf{true} \land F$, which is the same as what \cite{kobayashi2018} gives.

\subsubsection{Details}
We consider the $\lambda_c$-signature $\Sigma$ defined in Example~\ref{ex:trace-property-syntax} and the $\lambda_c(\Sigma)$-structure $\mathcal{A}$ defined in Example~\ref{ex:trace-property-semantics}.

May/must-reachability can be expressed by weakest preconditions.
\begin{example}[may reachability]\label{ex:may-reachability-wp}
	May-reachability is the negation of a trace property.
	Consider a trivial automaton $\mathfrak{A}_0$.
	The language accepted by $\mathfrak{A}_0$ is $L(\mathfrak{A}_0) = \{ \varepsilon \}$.
	May-reachability asks if there exists an output string that is not accepted by $\mathfrak{A}_0$.
	By Example~\ref{ex:trace-property-wp}, the may-reachability for $f : 1 \to T^P 1$ is true if and only if $q_0 \notin \mathrm{wp}^{\emalgsymbol_{\mathrm{tr}}}[f](\lambda r. U)$.
\end{example}

\begin{example}[must reachability]\label{ex:must-reachability-wp}
	Given a DFA $\mathfrak{A} = (U, \delta, q_0, U)$, consider an $\omega$cpo $\answerobj = (2^U, \subseteq)$ and define a $\mathcal{T}_P$-algebra on $\answerobj$ as follows.
	Then, this defines a EM $T^P$-algebra $\emalgsymbol_{\mathrm{must}} : T^P \answerobj \to \answerobj$.
	\begin{gather}
		\nontermconst^{\answerobj} \coloneqq \emptyset \quad
		x \join^{\answerobj} y \coloneqq x \cap y \quad
		\operation{event}^{\answerobj}_a(x) \coloneqq [a] x \coloneqq \{ q \in U \mid \forall q', q \xrightarrow{a} q' \implies q' \in x \}
		\label{eq:must-algebra}
	\end{gather}

	We show that the weakest precondition for $\emalgsymbol_{\mathrm{must}}$ corresponds to must-reachability of $f : 1 \to T^P 1$ if $\mathfrak{A}$ is the trivial automaton $\mathfrak{A}_0$ and the postcondition is given by $Q = \emptyset \in \answerobj$.
	In this case, we use $S 1$ (see Section~\ref{sec:free-algebra-nondet}) since~\eqref{eq:must-algebra} is a $\mathcal{T}_S$-algebra.
	By the freeness of $S 1$, we have a unique algebra morphism $\hat{Q} : S 1 \to \answerobj$ such that $Q = \hat{Q} \comp \eta^S_1$.
	The morphism $\hat{Q}$ is given by
	$\hat{Q} Y = \bigcap_{(s, \bot) \in Y} [s] \bot^{\answerobj} \cap \bigcap_{(s, \star) \in Y} [s] Q(\star)$
	where $[a_1 \dots a_n] x = [a_1] \dots [a_n] x$ for any $a_1 \dots a_n \in E^{*}$ and $[a_1 \dots ] \bot^{\answerobj} = \bigcup_n [a_1 \dots a_n] \bot^{\answerobj}$ for any $a_1 \dots \in E^{\omega}$.
	Similarly to Example~\ref{ex:trace-property-wp}, we have $\hat{Q} \comp h^S = \emalgsymbol_{\mathrm{must}} \comp T^P Q$ by the freeness of $T^P 1$.
	We also have $[a] x = U$ for any $x \subseteq U$ since the trivial automaton has no transition.
	As a result, we get the following.
	\[ q_0 \in \mathrm{wp}^{\emalgsymbol_{\mathrm{must}}}[f](\emptyset) = \bigcap_{(s, x) \in s(f(\star))} [s] \emptyset \qquad\iff\qquad \forall (s, x) \in h^S(f(\star)), s \neq \varepsilon \]
	This can be read as ``the must-reachability for $f : 1 \to T^P 1$ is true if and only if $q_0 \in \mathrm{wp}^{\emalgsymbol_{\mathrm{must}}}[f](\emptyset)$''.
\end{example}

Then, we apply Theorem~\ref{thm:cps-is-wpt-with-recursion} to the above weakest preconditions.

\begin{example}[must reachability]\label{ex:must-reachability-cps}
	Let $\vdash M : 1$ be a $\lambda_c(\Sigma)$-term.
	By Example~\ref{ex:must-reachability-wp}, the must-reachability for $M$ is true if and only if $q_0 \in \mathrm{wp}^{\emalgsymbol_{\mathrm{must}}}[\mathcal{A}\interpret{M}](\emptyset)$ where $\mathcal{A}$ is defined in Example~\ref{ex:trace-property-semantics}; and by Theorem~\ref{thm:cps-is-wpt-with-recursion}, the must-reachability is true if and only if $q_0 \in \mathcal{A}^{\emalgsymbol_{\mathrm{must}}} \interpret{\CPS{M} (\lambda \_. \mathbf{false})}$ where $\mathcal{A}^{\emalgsymbol_{\mathrm{tr}}} \interpret{\mathbf{false}} = \emptyset$.
	This corresponds to~\cite[Thm~2]{kobayashi2018}.

	In this case, the CPS transformation itself is the same as the trace property (Example~\ref{ex:trace-property-cps}), but the interpretation of $\lambda_{\mathrm{HFL}}(\Sigma)$-terms is different because we use a different EM algebra.
	The modal operator for an event operator $\operation{event}_a$ is interpreted as the always-true since the trivial automaton $\mathfrak{A}_0$ has no transition.
	\[ \mathcal{A}^{\emalgsymbol_{\mathrm{must}}} \interpret{\operation{event}_a(M)}(x) = [a] (\mathcal{A}^{\emalgsymbol_{\mathrm{must}}} \interpret{M}(x)) = \{q_0\} = \mathcal{A}^{\emalgsymbol_{\mathrm{must}}} \interpret{\mathbf{true}}(x) \]
	The modal operator for nondeterministic branching $\join$ is interpreted as conjunction.
	\[ \mathcal{A}^{\emalgsymbol_{\mathrm{must}}} \interpret{M \join N}(x) = \mathcal{A}^{\emalgsymbol_{\mathrm{must}}} \interpret{M}(x) \cap \mathcal{A}^{\emalgsymbol_{\mathrm{must}}} \interpret{N}(x) \]
	Since we use the standard order $(2^{\{ q_0 \}}, \subseteq)$, fixed points are interpreted as the least fixed points.
\end{example}
Extending the syntax of $\lambda_{\mathrm{HFL}}$, we obtain a more convenient CPS transformation.
\begin{example}[must-reachability, continued from Example~\ref{ex:must-reachability-cps}]\label{ex:must-reachability-extended}
	Similarly to Example~\ref{ex:trace-propery-extended}, $\answerobj = (2^U, {\subseteq})$ also has an internal bounded distributive lattice structure $(\answerobj, \mathbf{true}, {\land}, \mathbf{false}, {\lor})$ defined by $(\answerobj, U, {\cap}, \emptyset, {\cup})$.
	We consider $\lambda_{\mathrm{HFL}}$-terms extended with this internal bounded distributive lattice structure.
	By replacing modal operators with bounded-distributive-lattice operations, we can redefine our CPS transformation as $\CPS{(M_1 \join M_2)} = \lambda k. \CPS{M_1}\ k \land \CPS{M_2}\ k$ and $\CPS{\operation{event}_a(M)} = \mathbf{true}$.
\end{example}

\ifthenelse{\boolean{submission}}{}{
\subsection{Total Correctness for Programs with States}
\begin{definition}
	We define a $\lambda_c$-signature $\Sigma = (B, K, O)$ as follows.
	\begin{itemize}
		\item $B = \{ \mathbf{loc}, \mathbf{val}, \dots \}$
		\item $K = \{ l : 1 \rightarrowtriangle \mathbf{loc}, v : 1 \rightarrowtriangle \mathbf{val}, \dots \}$ where $l$ and $v$ range over locations of type $\mathbf{loc}$ and $\mathbf{val}$, respectively.
		\item $O = \{ \operation{lookup} : \mathbf{val} \rightarrowtriangle \mathbf{loc}, \operation{update} : 1 \rightarrowtriangle \mathbf{loc} \times \mathbf{val} \}$
	\end{itemize}
\end{definition}

\begin{definition}
	Let $L$ be a finite set of locations and $V$ be a countable set of values.
	We define a $\lambda_c(\Sigma)$-structure $\mathcal{A} = (\omegaCPO, T, A, a)$ as follows.
	\begin{itemize}
		\item $T \coloneqq \exponential{S}{((({-}) \times S)_{\bot})}$ where we define an $\omega$cpo $S$ by $(V^L, {=})$.
		Here, the monad is defined by the tensor of the algebraic theory of states and non-termination~\cite[Section~5]{hyland2006}.
		The unit is given by $\eta^T = (\exponential{S}{\eta^{({-})_{\bot}}}) \comp \mathbf{coev}$, that is, $\eta^T(x)(s) = \eta^{({-})_{\bot}}(x, s)$ where $\mathbf{coev}_X : X \to \exponential{S}{(X \times S)}$ is the coevaluation.
		The multiplication is given by $\mu^{T} = \exponential{S}{(\mu^{({-})_{\bot}} \comp \eval_{\bot})}$, that is,
		\[ \mu^T(f)(s) = \begin{cases}
			\bot & f(s) = \bot \\
			g(s') & f(s) = \eta^{({-})_{\bot}}(g, s').
		\end{cases} \]
		\item \[ A(\mathbf{loc}) = L \qquad A(\mathbf{val}) = V \]
		\item $\mathbf{Gef}(a(\operation{lookup})) : L \to T V$ is defined by
		\[ \mathbf{Gef}(a(\operation{lookup}))(l)(s) = \eta^{({-})_{\bot}}(s(l), s) \]
		and $\mathbf{Gef}(a(\operation{update})) : L \times V \to T 1$ is
		\[ \mathbf{Gef}(a(\operation{update}))(l, v)(s) = \eta^{({-})_{\bot}}(\star, s[l \mapsto v]) \]
	\end{itemize}
\end{definition}

\begin{definition}
	Let $\answerobj = (2^S, {\subseteq})$ be an $\omega$cpo where $2^S$ is the powerset of $S$.
	Note that $\answerobj$ is isomorphic to $\exponential{S}{2}$ where $2 = \{ \mathbf{false}, \mathbf{true}\}$ is an $\omega$cpo with the order defined by $\mathbf{false} \le \mathbf{true}$.
	We define an EM algebra $\emalgsymbol_{\mathrm{st}}$ on $\answerobj$ by
	\[ \exponential{S}{((\answerobj \times S)_{\bot})} \xrightarrow{\exponential{S}{\eval_{\bot}}} \exponential{S}{2_{\bot}} \xrightarrow{\exponential{S}{\emalgsymbol_{\mathrm{tot}}}} \exponential{S}{2} = \answerobj \]
	where $\emalgsymbol_{\mathrm{tot}}$ is defined in Example~\ref{ex:total-correctness-wp}.
\end{definition}

\begin{lemma}
	The morphism $\emalgsymbol_{\mathrm{st}}$ is an EM algebra.
\end{lemma}
\begin{proof}
	We prove $\emalgsymbol_{\mathrm{st}} \comp \eta^T = \identity{}$ and $\emalgsymbol_{\mathrm{st}} \comp \mu^T = \emalgsymbol_{\mathrm{st}} \comp T \emalgsymbol_{\mathrm{st}}$.
	\begin{align}
		\emalgsymbol_{\mathrm{st}} \comp \eta^T &= (\exponential{S}{\emalgsymbol_{\mathrm{tot}}}) \comp (\exponential{S}{\eval_{\bot}}) \comp (\exponential{S}{\eta^{({-})_{\bot}}}) \comp \mathbf{coev} \\
		&= (\exponential{S}{(\emalgsymbol_{\mathrm{tot}} \comp \eval_{\bot} \comp \eta^{({-})_{\bot}})}) \comp \mathbf{coev} \\
		&= (\exponential{S}{(\emalgsymbol_{\mathrm{tot}} \comp \eta^{({-})_{\bot}} \comp \eval)}) \comp \mathbf{coev} \\
		&= (\exponential{S}{\eval}) \comp \mathbf{coev} \\
		&= \identity{}
	\end{align}
	\begin{align}
		\emalgsymbol_{\mathrm{st}} \comp \mu^T &= (\exponential{S}{\emalgsymbol_{\mathrm{tot}}}) \comp (\exponential{S}{\eval_{\bot}}) \comp (\exponential{S}{\mu^{({-})_{\bot}}}) \comp (\exponential{S}{\eval_{\bot}}) \\
		&= \exponential{S}{(\emalgsymbol_{\mathrm{tot}} \comp \eval_{\bot} \comp \mu^{({-})_{\bot}} \comp \eval_{\bot})} \\
		&= \exponential{S}{(\emalgsymbol_{\mathrm{tot}} \comp \mu^{({-})_{\bot}} \comp (\eval_{\bot})_{\bot} \comp \eval_{\bot})} \\
		&= \exponential{S}{(\emalgsymbol_{\mathrm{tot}} \comp (\emalgsymbol_{\mathrm{tot}})_{\bot} \comp (\eval_{\bot})_{\bot} \comp \eval_{\bot})} \\
		&= \exponential{S}{(\emalgsymbol_{\mathrm{tot}} \comp (\emalgsymbol_{\mathrm{tot}} \comp \eval_{\bot} \comp \eval)_{\bot})} \\
		&= \exponential{S}{(\emalgsymbol_{\mathrm{tot}} \comp (\eval \comp ((\exponential{S}{(\emalgsymbol_{\mathrm{tot}} \comp \eval_{\bot})}) \times S))_{\bot})} \\
		&= (\exponential{S}{(\emalgsymbol_{\mathrm{tot}} \comp \eval_{\bot})}) \comp (\exponential{S}{((\exponential{S}{(\emalgsymbol_{\mathrm{tot}} \comp \eval_{\bot})}) \times S)_{\bot}}) \\
		&= \emalgsymbol_{\mathrm{st}} \comp T \emalgsymbol_{\mathrm{st}}
	\end{align}	
\end{proof}

We identify the pre/post-condition of the form $Q : X \to \answerobj$ with the (upward-closed) subset $\{ (x, s) \mid Q(x)(s) = \mathbf{true} \} \subseteq X \times S$.
Then, for any $f : X \to T Y$ and a postcondition $Q : Y \to \answerobj$, the weakest precondition $\mathrm{wp}[f](Q) : X \to \answerobj$ corresponds to the following subset.
\[ \{ (x, s) \mid \exists (y, s') \in Y \times S, \eta^{({-})_{\bot}}(y, s') = f(x)(s) \land Q(y)(s') = \mathbf{true} \} \]
because we have
\[ \mathrm{wp}[f](Q)(x)(s) = \begin{cases}
	\mathbf{false} & f(x)(s) = \bot \\
	Q(y)(s') & f(x)(s) = \eta^{({-})_{\bot}}(y, s')
\end{cases} \]

\todo{prove}
Let $\Gamma \vdash M : \mathbf{val} \to \answertype$ and $\Gamma \vdash N : \mathbf{loc}$ be $\lambda_{\mathrm{HFL}}(\Sigma)$-terms.
Let $\gamma \in \mathcal{A}^{\emalgsymbol_{\mathrm{st}}}\interpret{\Gamma}$, $m = \mathcal{A}^{\emalgsymbol_{\mathrm{st}}}\interpret{M}(\gamma) : V \to \answerobj$, and $n = \mathcal{A}^{\emalgsymbol_{\mathrm{st}}}\interpret{N}(\gamma) \in L$.
Then, the modal operator for $\operation{lookup}$ is given as follows.
\[ \mathcal{A}^{\emalgsymbol_{\mathrm{st}}}\interpret{\operation{lookup}\ (M, N)}(\gamma)(s) = m(s(n))(s) \]

Let $\Gamma \vdash M : \answertype$ and $\Gamma \vdash N : \mathbf{loc} \times \mathbf{val}$ be $\lambda_{\mathrm{HFL}}(\Sigma)$-terms.
Let $\gamma \in \mathcal{A}^{\emalgsymbol_{\mathrm{st}}}\interpret{\Gamma}$, $m = \mathcal{A}^{\emalgsymbol_{\mathrm{st}}}\interpret{M}(\gamma) \in \answerobj$, and $(n_1, n_2) = \mathcal{A}^{\emalgsymbol_{\mathrm{st}}}\interpret{N}(\gamma) \in L \times V$.
Then, the modal operator for $\operation{update}$ is given as follows.
\[ \mathcal{A}^{\emalgsymbol_{\mathrm{st}}}\interpret{\operation{update}\ (\lambda x. M, N)}(\gamma)(s) = m(s[n_1 \mapsto n_2]) \]

If we decompose the answer type $\answertype$ as $\answertype = (\mathbf{loc} \to \mathbf{val}) \to \answertype'$ where $\answertype'$ is interpreted as $(\{ \mathbf{false}, \mathbf{true} \}, {\le})$, and extend the syntax of $\lambda_{\mathrm{HFL}}(\Sigma)$-terms by adding a ternary operation $({-})[({-}) \mapsto ({-})] : (\mathbf{loc} \to \mathbf{val}) \times \mathbf{loc} \times \mathbf{val} \to \mathbf{loc} \to \mathbf{val}$, then
\[ \mathcal{A}^{\emalgsymbol_{\mathrm{st}}}\interpret{\operation{lookup}\ (M, N)} = \mathcal{A}^{\emalgsymbol_{\mathrm{st}}}\interpret{\lambda s. M\ (s\ N)\ s} \]
\[ \mathcal{A}^{\emalgsymbol_{\mathrm{st}}}\interpret{\operation{update}\ (\lambda x. M, (N_1, N_2))} = \mathcal{A}^{\emalgsymbol_{\mathrm{st}}}\interpret{\lambda s. M\ (s[N_1 \mapsto N_2])} \]
}

\subsection{Exception}
\[ B = \{ \mathbf{ex} \} \]
\[ O = \{ \operation{raise} : 0 \rightarrowtriangle \mathbf{ex} \} \]

Let $E$ be a countable set of exceptions.
\[ \mathcal{A} = (\omegaCPO, (({-}) + E)_{\bot}, A, a) \]
Here, the monad is defined by the sum of $({-}) + E$ and $({-})_{\bot}$~\cite[Section~3]{hyland2006}.
\[ A(\mathbf{ex}) = E \]
\[ \mathbf{Gef}(a(\operation{raise})) = \eta^{({-})_{\bot}} \comp \iota_2 \]

Let $\answerobj = (2, {\le})$.
Given a function $p_{\mathrm{ab}} : E \to \answerobj$, we define an EM algebra $\emalgsymbol_{\mathrm{e}}$ by
\[ (\answerobj + E)_{\bot} \xrightarrow{[\identity{}, p_{\mathrm{ab}}]_{\bot}} \answerobj_{\bot} \xrightarrow{\emalgsymbol_{\mathrm{tot}}} \answerobj \]
Here, $p_{\mathrm{ab}} : E \to \answerobj$ represents an \emph{abnormal postcondition}~\cite{rauch2017}.

\begin{lemma}
	The morphism $\emalgsymbol_{\mathrm{e}}$ is an EM algebra.
\end{lemma}
\begin{proof}
	\begin{align}
		\emalgsymbol_{\mathrm{e}} \comp \eta &= \emalgsymbol_{\mathrm{tot}} \comp [\identity{}, p_{\mathrm{ab}}]_{\bot} \comp \eta^{({-})_{\bot}} \comp \iota_1 \\
		&= \emalgsymbol_{\mathrm{tot}} \comp \eta^{({-})_{\bot}} \comp [\identity{}, p_{\mathrm{ab}}] \comp \iota_1 \\
		&= \identity{}
	\end{align}
	\begin{align}
		\emalgsymbol_{\mathrm{e}} \comp \mu &= \emalgsymbol_{\mathrm{tot}} \comp [\identity{}, p_{\mathrm{ab}}]_{\bot} \comp \mu^{({-})_{\bot}} \comp [\identity{}, \eta^{({-})_{\bot}} \comp \iota_2]_{\bot} \\
		&= \emalgsymbol_{\mathrm{tot}} \comp \mu^{({-})_{\bot}} \comp ([\identity{}, p_{\mathrm{ab}}]_{\bot})_{\bot} \comp [\identity{}, \eta^{({-})_{\bot}} \comp \iota_2]_{\bot} \\
		&= \emalgsymbol_{\mathrm{tot}} \comp (\emalgsymbol_{\mathrm{tot}})_{\bot} \comp ([\identity{}, p_{\mathrm{ab}}]_{\bot})_{\bot} \comp [\identity{}, \eta^{({-})_{\bot}} \comp \iota_2]_{\bot} \\
		&= \emalgsymbol_{\mathrm{tot}} \comp (\emalgsymbol_{\mathrm{tot}} \comp [\identity{}, p_{\mathrm{ab}}]_{\bot} \comp [\identity{}, \eta^{({-})_{\bot}} \comp \iota_2])_{\bot} \\
		&= \emalgsymbol_{\mathrm{tot}} \comp (\emalgsymbol_{\mathrm{tot}} \comp [[\identity{}, p_{\mathrm{ab}}]_{\bot}, [\identity{}, p_{\mathrm{ab}}]_{\bot} \comp \eta^{({-})_{\bot}} \comp \iota_2])_{\bot} \\
		&= \emalgsymbol_{\mathrm{tot}} \comp (\emalgsymbol_{\mathrm{tot}} \comp [[\identity{}, p_{\mathrm{ab}}]_{\bot}, \eta^{({-})_{\bot}} \comp p_{\mathrm{ab}}])_{\bot} \\
		&= \emalgsymbol_{\mathrm{tot}} \comp ([\emalgsymbol_{\mathrm{tot}} \comp [\identity{}, p_{\mathrm{ab}}]_{\bot}, p_{\mathrm{ab}}])_{\bot} \\
		&= \emalgsymbol_{\mathrm{tot}} \comp ([\identity{}, p_{\mathrm{ab}}] \comp ((\emalgsymbol_{\mathrm{tot}} \comp [\identity{}, p_{\mathrm{ab}}]_{\bot}) + E))_{\bot} \\
		&= \emalgsymbol_{\mathrm{tot}} \comp [\identity{}, p_{\mathrm{ab}}]_{\bot} \comp ((\emalgsymbol_{\mathrm{tot}} \comp [\identity{}, p_{\mathrm{ab}}]_{\bot}) + E)_{\bot} \\
		&= \emalgsymbol_{\mathrm{e}} \comp (\emalgsymbol_{\mathrm{e}} + E)_{\bot}
	\end{align}	
\end{proof}

For $\Gamma \vdash M : \mathbf{ex}$, the modal operator $\operation{raise}$ is interpreted as follows.
\[ \mathcal{A}^{\emalgsymbol}\interpret{\operation{raise}\ (\lambda x. \delta(x), M)} = p_{\mathrm{ab}} \comp \mathcal{A}^{\emalgsymbol}\interpret{M} \]

\ifthenelse{\boolean{submission}}{}{
\subsection{Weakest (Liberal) Pre-Expectation}
\todo{}
We define $\mathbb{W}_{[0, 1]} = ([0, 1], \mathbf{Meas}(\mathbb{R}, [0, 1]), {\le})$.
Note that we have $\lnot : \mathbb{W}_{[0, 1]} \to (\mathbb{W}_{[0, 1]})^{\op}$ defined by $\lnot(x) = 1 - x$.

We define $\emalgsymbol^P : P \mathbb{W}_{[0, 1]} \to \mathbb{W}_{[0, 1]}$ by \todo{}
Note that $\lnot \comp \emalgsymbol^P \comp P \lnot^{-1} : P (\mathbb{W}_{[0, 1]})^{\op} \to (\mathbb{W}_{[0, 1]})^{\op}$ is given by \todo{}

\subsection{May/Must for Nondeterminism}

\subsection{WP for Probabilistic Programs with Conditioning}
\[ B = \{ \mathbf{real} \} \]
\[ O = \{ \operation{unif} : \mathbf{real} \rightarrowtriangle 1, \operation{score} : 1 \rightarrowtriangle \mathbf{real} \} \]

Let $T$ be the statistical powerdomain monad~\cite{vakar2019}.
\[ \mathcal{A} = (\omegaQBS, T, A, a) \]
\[ A(\mathbf{real}) = (\mathbb{R}, \mathbf{Meas}(\mathbb{R}, \mathbb{R}), {=}) \]

$\mathbf{Gef}(a(\operation{score})) : \mathbb{R} \to T 1$
\[ \mathbf{Gef}(a(\operation{score}))(r)(w) = |r| \cdot w(\star) \]

Let $\answerobj = \mathbb{W}$.
As an EM $T$-algebra, we define $\emalgsymbol^T : T \mathbb{W} \to \mathbb{W}$ by the expected value $\emalgsymbol^T(f) = f(\identity{})$.

Let $\Gamma \vdash M : \answertype$ and $\Gamma \vdash N : \mathbf{real}$ be $\lambda_{\mathrm{HFL}}(\Sigma)$-terms.
Let $\gamma \in \mathcal{A}^{\emalgsymbol^T}\interpret{\Gamma}$.
\[ \mathcal{A}^{\emalgsymbol^T}\interpret{\operation{score}\ (\lambda x. M, N)}(\gamma) = |\mathcal{A}^{\emalgsymbol^T}\interpret{N}(\gamma)| \cdot \mathcal{A}^{\emalgsymbol^T}\interpret{M}(\gamma) \]
Thus, if we extend the syntax of $\lambda_{\mathrm{HFL}}(\Sigma)$-terms by adding $|{-}| : \mathbf{real} \to \answertype$ and $({\cdot}) : \answertype^2 \to \answertype$, then we can rewrite the modal operator for $\operation{score}$ as follows.
\[ \mathcal{A}^{\emalgsymbol^T}\interpret{\operation{score}\ (\lambda x. M, N)} = \mathcal{A}^{\emalgsymbol^T}\interpret{|N| \cdot M} \]

\subsection{Expected Cost for Unnormalised Distributions}

Let $T$ be the statistical powerdomain monad~\cite{vakar2019}.

EM algebra:
$\answerobj = \mathbb{W}^2$

\begin{itemize}
	\item $\emalgsymbol^{T}_{2} : T \answerobj \to \answerobj$ is the 2-fold product of the expectation $\emalgsymbol^{T} : T \mathbb{W} \to \mathbb{W}$, which is given by $\emalgsymbol^{T}_{2}(f) = (f(\pi_1), f(\pi_2))$ for each $f \in T \answerobj$.
	\item $\emalgsymbol^{\mathbb{W} \times ({-})} : \mathbb{W} \times \answerobj \to \answerobj$ is defined by $\emalgsymbol^{\mathbb{W} \times ({-})}(a, (x_0, x_1)) = (x_0, a x_0 + x_1)$.
	This defines a $(\mathbb{W}, 0, {+})$-action on $\answerobj$.
	\item $\emalgsymbol^{({-})_{\bot}} : \answerobj_{\bot} \to  \answerobj$ defined by $\emalgsymbol^{({-})_{\bot}}(\bot) = (0, 0)$.
\end{itemize}

\begin{lemma}
	Two EM algebras $\emalgsymbol^{T}_{2}$ and $\emalgsymbol^{\mathbb{W} \times ({-})}$ satisfy the composite law in Lemma~\ref{lem:distributive-law-em-algebra}.
	\[ \emalgsymbol^{T}_{2} \comp T \emalgsymbol^{\mathbb{W} \times ({-})} \comp \strength^{T} = \emalgsymbol^{\mathbb{W} \times ({-})} \comp (\mathbb{W} \times \emalgsymbol^{T}_{2}) \]
\end{lemma}
\begin{proof}
	Let $f \in T \answerobj$ and $r \in \mathbb{W}$.
	Note that $f$ is linear.
	\begin{align}
		(\emalgsymbol^{T}_{2} \comp T \emalgsymbol^{\mathbb{W} \times ({-})} \comp \strength^{T})(r, f) &= (\emalgsymbol^{T}_{2} \comp T \emalgsymbol^{\mathbb{W} \times ({-})})(\lambda w. f(w(r, {-}))) \\
		&= \emalgsymbol^{T}_{2} (\lambda w. f(w(\emalgsymbol^{\mathbb{W} \times ({-})}(r, {-})))) \\
		&= (f(\pi_1(\emalgsymbol^{\mathbb{W} \times ({-})}(r, {-}))), f(\pi_2(\emalgsymbol^{\mathbb{W} \times ({-})}(r, {-})))) \\
		&= (f(\pi_1), f(\lambda (x_0, x_1). r \cdot x_0 + x_1)) \\
		&= (f(\pi_1), r \cdot f(\pi_1) + f(\pi_2)) \\
		&= \emalgsymbol^{\mathbb{W} \times ({-})} (r, (f(\pi_1), f(\pi_2))) \\
		&= \emalgsymbol^{\mathbb{W} \times ({-})} \comp (\mathbb{W} \times \emalgsymbol^{T}_{2}) (r, f)
	\end{align}
\end{proof}
The composite law for $\emalgsymbol^{({-})_{\bot}}$ also holds.

The weakest precondition transformer defined by the composite EM algebra is given as follows.
\[ \mathrm{wp}[f](\lambda x. (1, 0)) = \emalgsymbol^T \comp \tupling{h_1}{h_2} \comp f \]
where
\begin{gather}
	h_1(y, c) = \begin{cases}
		0 & y = \bot \text{ \todo{this isn't good?}} \\
		1 & y \in Y
	\end{cases} \qquad\qquad
	h_2(y, c) = \begin{cases}
		0 & y = \bot \\
		c & y \in Y
	\end{cases}
\end{gather}

\[ \mathcal{A}^{\emalgsymbol}\interpret{M^{\checkmark}}(\gamma) = \emalgsymbol^{\mathbb{W} \times ({-})}(1, \mathcal{A}^{\emalgsymbol}\interpret{M}(\gamma)) \]
We decompose $\answertype = \answertype' \times \answertype'$ where $\answertype'$ is \todo{}.

\[ \mathcal{A}^{\emalgsymbol}\interpret{M^{\checkmark}} = \mathcal{A}^{\emalgsymbol}\interpret{(\pi_1\ M, \pi_1\ M + \pi_2\ M)} \]
}

\section{De Morgan Duality}\label{sec:duality}
De Morgan duality deserves a detailed explanation among possible extensions of the target language.
This extension allows us to take the negation of a $\lambda_{\mathrm{HFL}}$-term, and we apply this to may-reachability.

\subsection{Syntax}
To keep track of variance, we extend the target language as follows.
\begin{definition}[$\lambda^{\pm}_{\mathrm{HFL}}(\Sigma)$-types/terms]\label{def:target-language-de-morgan}
	Let $\Sigma$ be a $\lambda_c$-signature.
	We extend $\lambda_{\mathrm{HFL}}(\Sigma)$-types/terms as follows and call them $\lambda^{\pm}_{\mathrm{HFL}}(\Sigma)$-types/terms.
	\begin{gather}
		\answertype^{?} \coloneqq \answertype^{+} \mid \answertype^{-} \qquad
		\rho, \tau \coloneqq \answertype^{?} \mid b \mid 1 \mid \rho \times \tau \mid 0 \mid \rho + \tau \mid \rho \to \answertype^{?} \\
		\begin{aligned}
			M, N \coloneqq &\dots \mid o^{?}\ M \mid \letrecchurch{f}{x}{\rho}{\answertype^{?}}{M}{N} \mid \dots \\
			&\top^{?} \mid \bot^{?} \mid \lnot^{?} M \mid M \land^{?} N \mid M \lor^{?} N \qquad\qquad\qquad
			\text{where\hspace{1em}${?} \coloneqq {+} \mid {-}$}
		\end{aligned}
	\end{gather}
	We call $\answertype^{?}$ a \emph{proposition type}.
	A \emph{ground type} is a type constructed without $\answertype^{?}$ and $\tau \to \answertype^{?}$.
\end{definition}
In Definition~\ref{def:target-language-de-morgan}, the answer type $\answertype$ is annotated by ${+}$/${-}$, and terms are annotated when their typing rules involve $\answertype^{?}$.
We may omit annotations ${+}/{-}$ for terms if they are clear from context, but we do not omit annotations for proposition types.
Note that the target language $\lambda_{\mathrm{HFL}}$ defined in Section~\ref{subsec:target-language} can be embedded to $\lambda^{\pm}_{\mathrm{HFL}}$ by letting $\answertype = \answertype^{+}$.

Typing rules for $\lambda_{\mathrm{HFL}}$ are also extended to $\lambda^{\pm}_{\mathrm{HFL}}$ according to the following principle: taking negation $\lnot$ is the only way for a positive/negative proposition to interact with negative/positive propositions.
For example, given $M, N : \answertype^{+}$ (or $M, N : \answertype^{-}$), we can construct a conjunction $M \land^{+} N : \answertype^{+}$ (or $M \land^{-} N : \answertype^{-}$), but we cannot construct a conjunction $M \land N$ of $M : \answertype^{+}$ and $N : \answertype^{-}$.
See Fig.~\ref{fig:typing-rules-de-morgan} and Section~\ref{subsec:de-morgan-typing-rules} for the full definition.

\begin{figure}[tbp]
	\begin{mathpar}
		\inferrule{
			\Gamma \vdash M : \answertype^{+}
		}{
			\Gamma \vdash \lnot^{+} M : \answertype^{-}
		}
		\and
		\inferrule{
			\Gamma \vdash M : \answertype^{-}
		}{
			\Gamma \vdash \lnot^{-} M : \answertype^{+}
		}
		\and
		\inferrule{
			\Gamma \vdash M : \answertype^{?} \\
			\Gamma \vdash N : \answertype^{?}
		}{
			\Gamma \vdash M \land^{?} N : \answertype^{?}
		}
		\and
		\inferrule{
			\Gamma \vdash M : (\mathrm{ar}(o) \to \answertype^{?}) \times \mathrm{car}(o)
		}{
			\Gamma \vdash o^{?}\ M : \answertype^{?}
		}
		\and
		\inferrule{
			\Gamma, f : \rho \to \answertype^{?}, x : \rho \vdash M : \answertype^{?} \\
			\Gamma, f : \rho \to \answertype^{?} \vdash N : \tau
		}{
			\Gamma \vdash \letrecchurch{f}{x}{\rho}{\answertype^{?}}{M}{N} : \tau
		}
	\end{mathpar}
	\caption{Selected typing rules for $\lambda^{\pm}_{\mathrm{HFL}}(\Sigma)$-terms. In the rules, $? \in \{ {+}, {-} \}$.}
	\label{fig:typing-rules-de-morgan}
\end{figure}

\subsection{Typing Rules}\label{subsec:de-morgan-typing-rules}
We define typing rules for $\lambda^{\pm}_{\mathrm{HFL}}(\Sigma)$-terms.
The following rules are changed from typing rules for $\lambda_{\mathrm{HFL}}(\Sigma)$-terms (Section~\ref{subsec:source-typing-rules}).
\begin{mathpar}
	\inferrule{
		\Gamma \vdash M : 0
	}{
		\Gamma \vdash \delta^{?}(M) : \answertype^{?}
	}
	\and
	\inferrule{
		\Gamma \vdash M : \rho_1 + \rho_2 \\
		\Gamma, x_1 : \rho_1 \vdash M_1 : \answertype^{?} \\
		\Gamma, x_2 : \rho_2 \vdash M_2 : \answertype^{?}
	}{
		\Gamma \vdash \delta^{?}(M, x_1 : \rho_1. M_1, x_2 : \rho_2. M_2) : \answertype^{?}
	}
	\and
	\inferrule{
		\Gamma, x : \rho \vdash M : \answertype^{?}
	}{
		\Gamma \vdash (\lambda x : \rho. M)^{?} : \rho \to \answertype
	}
	\and
	\inferrule{
		\Gamma \vdash M : \rho \to \answertype^{?} \\
		\Gamma \vdash N : \rho
	}{
		\Gamma \vdash (M\ N)^{?} : \answertype^{?}
	}
	\and
	\inferrule{
		\Gamma \vdash M : (\mathrm{ar}(o) \to \answertype^{?}) \times \mathrm{car}(o)
	}{
		\Gamma \vdash o^{?}\ M : \answertype^{?}
	}
	\and
	\inferrule{
		\Gamma, f : \rho \to \answertype^{?}, x : \rho \vdash M : \answertype^{?} \\
		\Gamma, f : \rho \to \answertype^{?} \vdash N : \tau
	}{
		\Gamma \vdash \letrecchurch{f}{x}{\rho}{\answertype^{?}}{M}{N} : \tau
	}
\end{mathpar}
For logical connectives, typing rules are defined as follows.
\begin{mathpar}
	\inferrule{ }{
		\Gamma \vdash \top^{?} : \answertype^{?}
	}
	\and
	\inferrule{ }{
		\Gamma \vdash \bot^{?} : \answertype^{?}
	}
	\and
	\inferrule{
		\Gamma \vdash M : \answertype^{?}
	}{
		\Gamma \vdash \lnot^{?} M : \answertype^{\overline{?}}
	}
	\and
	\inferrule{
		\Gamma \vdash M : \answertype^{?} \\
		\Gamma \vdash N : \answertype^{?}
	}{
		\Gamma \vdash M \land^{?} N : \answertype^{?}
	}
	\and
	\inferrule{
		\Gamma \vdash M : \answertype^{?} \\
		\Gamma \vdash N : \answertype^{?}
	}{
		\Gamma \vdash M \lor^{?} N : \answertype^{?}
	}
\end{mathpar}
For each typing rule, all annotations ${?}$ must be instantiated by the same sign ${+}/{-}$.
In the typing rule for $\lnot^{?} M$, we define $\overline{({+})} = ({-})$ and $\overline{({-})} = ({+})$.

\subsection{Semantics}
Proposition types are interpreted by a de Morgan algebra.
\begin{definition}[internal de Morgan algebra]\label{def:de-morgan-algebra}
	An \emph{de Morgan algebra internal to $\category{C}$} is a tuple $(\mathbf{\Omega}^{+}, \mathbf{\Omega}^{-}, \lnot)$ where $\mathbf{\Omega}^{+} = (\answerobj^{+}, \top^{+}, {\land}^{+}, \bot^{+}, {\lor}^{+})$ and $\mathbf{\Omega}^{-} = (\answerobj^{-}, \bot^{-}, {\lor}^{-}, \top^{-}, {\land}^{-})$ are bounded distributive lattices internal to $\category{C}$ (note the difference of orders of operations between $\mathbf{\Omega}^{+}$ and $\mathbf{\Omega}^{-}$); and ${\lnot} : \answerobj^{+} \to \answerobj^{-}$ is an isomorphism of internal bounded distributive lattices, that is, the following equations are satisfied.
	\begin{gather}
		\lnot \comp \top^{+} = \bot^{-} \qquad
		\lnot \comp \bot^{+} = \top^{-} \\
		\lnot \comp {\land}^{+} = {\lor}^{-} \comp (\lnot \times \lnot) \qquad
		\lnot \comp {\lor}^{+} = {\land}^{-} \comp (\lnot \times \lnot)
	\end{gather}
\end{definition}

\begin{remark}
	Definition~\ref{def:de-morgan-algebra} is a multi-sorted version of internal de Morgan algebras, that is, we consider a set of operations on multiple objects $\answerobj^{+}, \answerobj^{-}$, like ${\lnot} : \answerobj^{+} \to \answerobj^{-}$.
	When considering ordered settings like $\omegaCPO$, the multi-sorted version of internal de Morgan algebras is better than the single-sorted version for internalising a negation ${\lnot} : \answerobj^{+} \to \answerobj^{-}$ because $\lnot$ is often an anti-monotonic function and cannot be defined as an endomorphism ${\lnot} : \answerobj \to \answerobj$.
\end{remark}

\begin{definition}[de Morgan EM algebra]\label{def:de-morgan-em-algebra}
	Let $(\mathbf{\Omega}^{+}, \mathbf{\Omega}^{-}, \lnot)$ be an internal de Morgan algebra.
	A \emph{de Morgan EM $T$-algebra} on $(\mathbf{\Omega}^{+}, \mathbf{\Omega}^{-}, \lnot)$ is a tuple $((\mathbf{\Omega}^{+}, \mathbf{\Omega}^{-}, \lnot), \emalgsymbol^{+}, \emalgsymbol^{-})$ such that $\emalgsymbol^{+} : T \answerobj^{+} \to \answerobj^{+}$ and $\emalgsymbol^{-} : T \answerobj^{-} \to \answerobj^{-}$ are EM algebras, and $\lnot$ is a morphism of EM algebras from $\emalgsymbol^{+}$ to $\emalgsymbol^{-}$.
\end{definition}
In Definition~\ref{def:de-morgan-em-algebra}, one of $\emalgsymbol^{+}$ and $\emalgsymbol^{-}$ is redundant because given an internal de Morgan algebra $(\mathbf{\Omega}^{+}, \mathbf{\Omega}^{-}, \lnot)$, an EM algebra $\emalgsymbol^{+}$ uniquely defines the other by $\emalgsymbol^{-} \coloneqq \lnot \comp \emalgsymbol^{+} \comp T \lnot^{-1} : T \answerobj^{-} \to \answerobj^{-}$ (and vice versa).
Therefore, we sometimes say ``$\emalgsymbol^{+}$ is a de Morgan EM algebra'' when the internal de Morgan algebra $(\mathbf{\Omega}^{+}, \mathbf{\Omega}^{-}, \lnot)$ is clear from the context.
We say $\emalgsymbol^{+}$ (or $\emalgsymbol^{-}$) is the \emph{dual} of $\emalgsymbol^{-}$ (or $\emalgsymbol^{+}$).

\begin{definition}\label{def:semantics-extended-target-language}
	Let $\mathcal{A} = (\category{C}, T, A, a)$ be an $\omegaCPO$-enriched $\lambda_c(\Sigma)$-structure $\emalgsymbol = ((\mathbf{\Omega}^{+}, \mathbf{\Omega}^{-}, \lnot), \emalgsymbol^{+}, \emalgsymbol^{-})$ be a de Morgan EM $T$-algebra.
	We extend the interpretation of $\lambda_{\mathrm{HFL}}(\Sigma)$-types/terms (Definition~\ref{def:interpretation-response-calculus}) to $\lambda^{\pm}_{\mathrm{HFL}}(\Sigma)$-types/terms as follows.
	We define the interpretation $\mathcal{A}^{\emalgsymbol} \interpret{\rho}$ of $\lambda^{\pm}_{\mathrm{HFL}}(\Sigma)$-types by $\mathcal{A}^{\emalgsymbol} \interpret{\answertype^{+}} \coloneqq \answerobj^{+}$ and $\mathcal{A}^{\emalgsymbol} \interpret{\answertype^{-}} \coloneqq \answerobj^{-}$ for proposition types and extend this to all types $\rho$ in the same way as Definition~\ref{def:interpretation-response-calculus}.
	The interpretation of $\lambda^{\pm}_{\mathrm{HFL}}(\Sigma)$-terms is defined using the structure of the de Morgan EM $T$-algebra.
	For example, $\lnot^{?} M$, $M \land^{?} N$, and $o^{?}\ M$ are interpreted as follows.
	See Section~\ref{subsec:de-morgan-semantics} for the full definition.
	\begin{gather}
		\mathcal{A}^{\emalgsymbol} \interpret{\lnot^{+} M} \coloneqq \lnot \comp \mathcal{A}^{\emalgsymbol} \interpret{M} \qquad
		\mathcal{A}^{\emalgsymbol} \interpret{\lnot^{-} M} \coloneqq \lnot^{-1} \comp \mathcal{A}^{\emalgsymbol} \interpret{M} \\
		\mathcal{A}^{\emalgsymbol} \interpret{M \land^{?} N} \coloneqq {\land}^{?} \comp \tupling{\mathcal{A}^{\emalgsymbol} \interpret{M}}{\mathcal{A}^{\emalgsymbol} \interpret{N}} \\
		\mathcal{A}^{\emalgsymbol} \interpret{o^{?}\ M} \coloneqq \emalgsymbol^{?} \comp T \mathbf{ev} \comp \strength^T \comp (\identity{} \times \mathbf{Gef}(a(o))) \comp \mathcal{A}^{\emalgsymbol} \interpret{M}
	\end{gather}
\end{definition}
Note that $\mathcal{A}^{\emalgsymbol} \interpret{-}$ in Definition~\ref{def:semantics-extended-target-language} extends Definition~\ref{def:interpretation-response-calculus} in the sense that we get the same interpretation $\mathcal{A}^{\emalgsymbol} \interpret{M} = \mathcal{A}^{\emalgsymbol^{+}} \interpret{M}$ if we regard a $\lambda_{\mathrm{HFL}}(\Sigma)$-term $M$ as a $\lambda^{\pm}_{\mathrm{HFL}}(\Sigma)$-term by letting $\answertype = \answertype^{+}$.

\begin{example}[total/partial correctness]
	Let $\answerobj^{+} = (\{ \mathbf{false}, \mathbf{true} \}, {\le})$ and $\answerobj^{-} = (\{ \mathbf{false}, \mathbf{true} \}, {\ge})$ where $\le$ is an order such that $\mathbf{false} \le \mathbf{true}$.
	Then, the pair of $\answerobj^{+}$ and $\answerobj^{-}$ has a de Morgan structure internal to $\omegaCPO$ where $\lnot : \answerobj^{+} \to \answerobj^{-}$ defined by $\lnot \mathbf{true} = \mathbf{false}$ and $\lnot \mathbf{false} = \mathbf{true}$.
	We have two EM algebras $\emalgsymbol_{\mathrm{tot}} : (\answerobj^{+})_{\bot} \to \answerobj^{+}$ (Example~\ref{ex:total-correctness-wp}) and $\emalgsymbol_{\mathrm{par}} : (\answerobj^{-})_{\bot} \to \answerobj^{-}$ (Example~\ref{ex:partial-correctness-wp}), and $\lnot : \answerobj^{+} \to \answerobj^{-}$ is a morphism between these EM algebras.
	Thus, we have a de Morgan EM algebra.

	In this situation, $\lambda^{\pm}_{\mathrm{HFL}}$ has both least and greatest fixed points: fixed points for $\answertype^{+}$ and $\answertype^{-}$ are interpreted as least and greatest fixed points, respectively.
\end{example}

\begin{example}[trace property/may reachability]\label{ex:trace-may-dual}
	We define a de Morgan algebra internal to $\omegaCPO$ by $\answerobj^{+} = (2^U, \supseteq)$ and $\answerobj^{-} = (2^U, \subseteq)$ with $\lnot X = U \setminus X$.
	We define an EM algebra $\emalgsymbol_{\mathrm{may}} : T^H \answerobj^{-} \to \answerobj^{-}$ as the dual of $\emalgsymbol_{\mathrm{tr}} : T^H \answerobj^{+} \to \answerobj^{+}$ defined in Example~\ref{ex:trace-property-wp}.
	This gives a de Morgan EM algebra $\emalgsymbol = ((\mathbf{\Omega}^{+}, \mathbf{\Omega}^{-}, \lnot), \emalgsymbol_{\mathrm{tr}}, \emalgsymbol_{\mathrm{may}})$.
	The EM algebra $\emalgsymbol_{\mathrm{may}} : T^H \answerobj^{-} \to \answerobj^{-}$ will be used for may reachability later in Example~\ref{ex:may-reachability-cps}.

	In this case, $\lambda^{\pm}_{\mathrm{HFL}}$ is a similar language to the HFL of~\cite{kobayashi2018,viswanathan2004}.
	When we use $\emalgsymbol$ to interpret $\lambda^{\pm}_{\mathrm{HFL}}(\Sigma)$-terms, $\operation{event}_a^{+}(M)$ in $\lambda^{\pm}_{\mathrm{HFL}}$-terms corresponds to $\langle a \rangle M$ in their HFL, and $\operation{event}_a^{-}(M)$ corresponds to $[ a ] M$.
	Fixed points for $\answertype^{+}$ and $\answertype^{-}$ are interpreted as greatest and least fixed points, respectively.

	On one hand, our $\lambda^{\pm}_{\mathrm{HFL}}$ generalises their HFL by considering general modal operators.
	On the other hand, $\lambda^{\pm}_{\mathrm{HFL}}$ interpreted by $\emalgsymbol$ is still a proper subset of their HFL.
	For example, the typing rules of $\lambda^{\pm}_{\mathrm{HFL}}$ does not allow terms like $\operation{event}_a^{+}(\operation{event}_a^{-}(M))$ because $\operation{event}_a^{+}$ cannot be applied to $\operation{event}_a^{-}(M) : \answertype^{-}$.
	We also restrict models for interpreting $\lambda^{\pm}_{\mathrm{HFL}}$-terms by allowing only deterministic automaton while they allow nondeterministic automaton.
	However, our aim here is to provide a target language that subsumes the image of the CPS transformation and not to make a target language as rich as possible.
\end{example}

\subsection{Semantics (Full Definition)}\label{subsec:de-morgan-semantics}
Let $\emalgsymbol$ be a de Morgan EM algebra.
Most of the interpretation of $\lambda^{\pm}_{\mathrm{HFL}}(\Sigma)$-terms are defined in the same way as that of $\lambda_{\mathrm{HFL}(\Sigma)}$-terms (Section~\ref{subsec:target-semantics}) except for the following terms that use EM algebra structures.
\begin{gather}
	\mathcal{A}^{\emalgsymbol} \llbracket o^{?}\ M \rrbracket = \emalgsymbol^{?} \comp T \mathbf{ev} \comp \strength^T \comp (\identity{} \times \mathbf{Gef}(a(o))) \comp \mathcal{A}^{\emalgsymbol} \interpret{M} \\
	\mathcal{A}^{\emalgsymbol} \interpret{\letrecchurch{f}{x}{\rho}{\answertype^{?}}{M}{N}} = \mathcal{A}^{\emalgsymbol} \interpret{N} \comp \tupling{\identity{}}{(\Lambda (\mathcal{A}^{\emalgsymbol} \interpret{M}))^{\dagger}}
\end{gather}

Logical connectives are interpreted by internal de Morgan algebra structures.
\begin{gather}
	\mathcal{A}^{\emalgsymbol} \interpret{\lnot^{+} M} \coloneqq \lnot \comp \mathcal{A}^{\emalgsymbol} \interpret{M} \qquad
	\mathcal{A}^{\emalgsymbol} \interpret{\lnot^{-} M} \coloneqq \lnot^{-1} \comp \mathcal{A}^{\emalgsymbol} \interpret{M} \\
	\mathcal{A}^{\emalgsymbol} \interpret{\top^{?}} \coloneqq \top^{?} \comp {!} \qquad
	\mathcal{A}^{\emalgsymbol} \interpret{\bot^{?}} \coloneqq \bot^{?} \comp {!} \\
	\mathcal{A}^{\emalgsymbol} \interpret{M \land^{?} N} \coloneqq {\land}^{?} \comp \tupling{\mathcal{A}^{\emalgsymbol} \interpret{M}}{\mathcal{A}^{\emalgsymbol} \interpret{N}} \qquad
	\mathcal{A}^{\emalgsymbol} \interpret{M \lor^{?} N} \coloneqq {\lor}^{?} \comp \tupling{\mathcal{A}^{\emalgsymbol} \interpret{M}}{\mathcal{A}^{\emalgsymbol} \interpret{N}}
\end{gather}

\subsection{Duality}
We explain that $\lambda^{\pm}_{\mathrm{HFL}}$ has a duality structure.
We define a syntactic translation that gives the dual of $\lambda^{\pm}_{\mathrm{HFL}}(\Sigma)$-types/terms and then list several properties about the duality.

\begin{figure}[tbp]
	\centering
	\begin{tikzpicture}
		\node (topp) at (0, 0.5) {$\top^{+}$};
		\node (topm) at (0, -0.5) {$\top^{-}$};
		\node (botp) at (1, 0.5) {$\bot^{+}$};
		\node (botm) at (1, -0.5) {$\bot^{-}$};
		\draw[<->] (topp) -- (botm);
		\draw[<->] (botp) -- (topm);
		\node (landp) at (2, 0.5) {$\land^{+}$};
		\node (landm) at (2, -0.5) {$\land^{-}$};
		\node (lorp) at (3, 0.5) {$\lor^{+}$};
		\node (lorm) at (3, -0.5) {$\lor^{-}$};
		\draw[<->] (landp) -- (lorm);
		\draw[<->] (lorp) -- (landm);
		\node (algopp) at (4, 0.5) {$o^{+}$};
		\node (algopm) at (4, -0.5) {$o^{-}$};
		\draw[<->] (algopp) -- (algopm);
		\node (fixp) at (7, 0.5) {$\mathbf{let}\ \mathbf{rec}\ f\ x\ : \answertype^{+} = \dots$};
		\node (fixm) at (7, -0.5) {$\mathbf{let}\ \mathbf{rec}\ f\ x\ : \answertype^{-} = \dots$};
		\draw[<->] (fixp) -- (fixm);
	\end{tikzpicture}
	\caption{Duality of $\lambda^{\pm}_{\mathrm{HFL}}(\Sigma)$-terms.}
	\label{fig:duality}
\end{figure}

\begin{definition}\label{def:de-morgan-involution}
	For any $\lambda^{\pm}_{\mathrm{HFL}}(\Sigma)$-type $\rho$, we define a $\lambda^{\pm}_{\mathrm{HFL}}(\Sigma)$-type $\overline{\rho}$ as follows.
	For proposition types, we define $\overline{\answertype^{?}} \coloneqq \answertype^{\overline{?}}$ where $\overline{({+})} = {-}$ and $\overline{({-})} = {+}$,
	and extend this to other types homomorphically.
	For any $\lambda^{\pm}_{\mathrm{HFL}}(\Sigma)$-term $M$, we define a $\lambda^{\pm}_{\mathrm{HFL}}(\Sigma)$-term $\overline{M}$ by
	\begin{gather}
		\overline{\bot^{?}} \coloneqq \top^{\overline{?}} \qquad
		\overline{\top^{?}} \coloneqq \bot^{\overline{?}} \qquad
		\overline{\lnot^{?} M} = \lnot^{\overline{?}} \overline{M} \\
		\overline{M \lor^{?} N} \coloneqq \overline{M} \land^{\overline{?}} \overline{N} \qquad
		\overline{M \land^{?} N} \coloneqq \overline{M} \lor^{\overline{?}} \overline{N}
	\end{gather}
	and for other terms, we just substitute $\overline{?}$ for $?$, e.g., $\overline{o^{?}\ M} \coloneqq o^{\overline{?}}\ \overline{M}$.
	The situation is depicted in Fig.~\ref{fig:duality}.
\end{definition}

By definition, $\overline{({-})}$ is involutive, and if $\rho$ is a ground type, then $\overline{\rho} = \rho$.

\begin{lemma}
	For any well-typed term $\Gamma \vdash M : \rho$, we have $\overline{\Gamma} \vdash \overline{M} : \overline{\rho}$ where $\overline{\Gamma}$ is defined by $\overline{x_1 : \rho_1, \dots, x_n : \rho_n} = x_1 : \overline{\rho_1}, \dots, x_n : \overline{\rho_n}$.
	\qed
\end{lemma}

\begin{theorem}[de Morgan's law]\label{thm:de-morgan-law}
	Let $\lnot_{\rho} : \mathcal{A}^{\emalgsymbol} \interpret{\rho} \to \mathcal{A}^{\emalgsymbol} \interpret{\overline{\rho}}$ be an extension of the negation $\lnot : \answerobj^{+} \to \answerobj^{-}$ to any type $\rho$, that is, for proposition types, we define $\lnot_{\answertype^{+}} \coloneqq {\lnot}$ and $\lnot_{\answertype^{-}} \coloneqq {\lnot}^{-1}$, and for other types, $\lnot_{\rho}$ is defined homomorphically.
	Then, we have the following equation.
	\begin{equation}
		\lnot_{\rho} \comp \mathcal{A}^{\emalgsymbol} \interpret{\Gamma \vdash M : \rho} = \mathcal{A}^{\emalgsymbol} \interpret{\overline{\Gamma} \vdash \overline{M} : \overline{\rho}} \comp \lnot_{\Gamma}
		\tag*{\qed}
	\end{equation}
\end{theorem}

\begin{corollary}\label{cor:negation-de-morgan-dual}
	For $\Gamma \vdash M : \answertype^{?}$ with $\Gamma$ consisting of ground types, $\mathcal{A}^{\emalgsymbol} \interpret{\lnot M} = \mathcal{A}^{\emalgsymbol} \interpret{\overline{M}}$.
	\qed
\end{corollary}

\begin{proposition}\label{prop:algebraic-operation-dual}
	For each $n$-ary $o \in O$, modal operators $o^{+}$ and $o^{-}$ have the following relation.
	\begin{equation}
		\mathcal{A}^{\emalgsymbol} \interpret{o^{\overline{?}}\ (M_1, \dots, M_n)} = \mathcal{A}^{\emalgsymbol} \interpret{\lnot (o^{?}\ (\lnot M_1, \dots, \lnot M_n))}
		\tag*{\qed}
	\end{equation}
\end{proposition}

\begin{proposition}
	Fixed points for $\answertype^{+}$ and $\answertype^{-}$ have the following relation.
	\begin{align}
		&\mathcal{A}^{\emalgsymbol} \interpret{\letrecchurch{f}{x}{\rho}{\answertype^{\overline{?}}}{M}{N}} \\
		&= \mathcal{A}^{\emalgsymbol}\interpret{\letrecchurch{f}{x}{\rho}{\answertype^{?}}{\lnot M[\lambda x. \lnot (f\ x)/f]}{N[\lambda x. \lnot (f\ x)/f]}} \tag*{\qed}
	\end{align}
\end{proposition}

\begin{proposition}
	For any well-typed $\lambda^{\pm}_{\mathrm{HFL}}(\Sigma)$-term $\Gamma \vdash M : \rho$, we have $\mathcal{A}^{\overline{\emalgsymbol}} \interpret{M} = \mathcal{A}^{\emalgsymbol} \interpret{\overline{M}}$ where $\overline{\emalgsymbol} \coloneqq ((\mathbf{\Omega}^{-}, \mathbf{\Omega}^{+}, \lnot^{-1}), \emalgsymbol^{-}, \emalgsymbol^{+})$ is the dual of a de Morgan EM algebra $\emalgsymbol = ((\mathbf{\Omega}^{+}, \mathbf{\Omega}^{-}, \lnot), \emalgsymbol^{+}, \emalgsymbol^{-})$.
	\qed
\end{proposition}

\begin{proposition}\label{prop:wp-dual}
	For any $f : X \to T Y$ and $Q : Y \to \answerobj^{-}$, we have $\mathrm{wp}^{\emalgsymbol^{-}}[f](Q) = \lnot \comp \mathrm{wp}^{\emalgsymbol^{+}}[f](\lnot^{-1} \comp Q)$.
	\qed
\end{proposition}

By Theorem~\ref{thm:cps-is-wpt-with-recursion}, we have a CPS transformation from $\lambda_c(\Sigma)$-terms to $\lambda^{\pm}_{\mathrm{HFL}}(\Sigma)$-terms that corresponds to the weakest precondition for $\emalgsymbol^{+}$ if we use $\answertype^{+}$ as an answer type.
By duality, we also have a CPS transformation for the weakest precondition for $\emalgsymbol^{-}$ if we use $\answertype^{-}$.
\begin{corollary}\label{cor:dual-cps}
	Let $\Gamma \vdash M : \rho$ be a well-typed $\lambda_c(\Sigma)$-term, $x : \rho \vdash Q_{+} : \answertype^{+}$ and $x : \rho \vdash Q_{-} : \answertype^{-}$ be well-typed $\lambda^{\pm}_{\mathrm{HFL}}(\Sigma)$-terms.
	The weakest preconditions for $\emalgsymbol^{+}$ and $\emalgsymbol^{-}$ are given as follows.
	\begin{align}
		\mathrm{wp}^{\emalgsymbol^{+}}[\mathcal{A} \interpret{M}](\mathcal{A}^{\emalgsymbol} \interpret{Q_{+}}) &= \mathcal{A}^{\emalgsymbol} \interpret{\CPS{M}\ (\lambda x. Q_{+})} \\
		\mathrm{wp}^{\emalgsymbol^{-}}[\mathcal{A} \interpret{M}](\mathcal{A}^{\emalgsymbol} \interpret{Q_{-}}) &= \mathcal{A}^{\emalgsymbol} \interpret{\overline{\CPS{M}}\ (\lambda x. Q_{-})}
		\tag*{\qed}
	\end{align}
\end{corollary}

\begin{example}[may-reachability, continued from Example~\ref{ex:trace-may-dual}]\label{ex:may-reachability-cps}
	Recall that may-reachability is the negation of the trace property for the trivial automaton with $L(\mathcal{A}_0) = \{ \varepsilon \}$.
	By Theorem~\ref{thm:cps-is-wpt-with-recursion} and by definition of $\lnot$, the may-reachability for a $\lambda_c(\Sigma)$-term $\vdash M : 1$ is true if and only if $q_0 \in \mathcal{A}^{\emalgsymbol} \interpret{\lnot (\CPS{M}\ (\lambda x. \mathbf{true}))}$.
	By the de Morgan law (Corollary~\ref{cor:negation-de-morgan-dual}), this is equivalent to $q_0 \in \mathcal{A}^{\emalgsymbol} \interpret{\overline{\CPS{M}}\ (\lambda x. \mathbf{false})}$.
	This gives essentially the same formula as~\cite[Thm~1]{kobayashi2018} because this translation $M \mapsto \overline{\CPS{M}}$ maps the event operator $\operation{event}_a$ to the modal operator $\operation{event}^{-}_a = [a]$ (and this is equivalent to $\mathbf{true}$ since there is no transition in $\mathcal{A}_0$), the nondeterministic branching operator $\join$ to the disjunction $\lor^{-}$, and a recursive function to the least fixed point.

	We can alternatively understand may-reachability as a weakest precondition for $\emalgsymbol_{\mathrm{may}}$.
	By Proposition~\ref{prop:wp-dual}, the may-reachability for $M$ is true if and only if $q_0 \in \mathrm{wp}^{\emalgsymbol_{\mathrm{may}}}[\mathcal{A} \interpret{M}](\emptyset)$.
	By Corollary~\ref{cor:dual-cps}, this is equivalent to $q_0 \in \mathcal{A}^{\emalgsymbol} \interpret{\overline{\CPS{M}}\ (\lambda x. \mathbf{false})}$, which is the same condition as above.
\end{example}

\section{Quantifiers}\label{sec:quantifiers}

\subsection{Syntax}
Let $\Sigma$ be a $\lambda_c$-signature.
We consider extending $\lambda_{\mathrm{HFL}}(\Sigma)$-terms with quantifiers.
\[ M, N \coloneqq \dots \mid \forall x : \rho. M \mid \exists x : \rho. M \]

\begin{mathpar}
	\inferrule{
		\Gamma, x : \rho \vdash M : \answertype
	}{
		\Gamma \vdash \forall x : \rho. M : \answertype
	}
	\and
	\inferrule{
		\Gamma, x : \rho \vdash M : \answertype
	}{
		\Gamma \vdash \exists x : \rho. M : \answertype
	}
\end{mathpar}

\subsection{Semantics}
Let $\mathcal{A} = (\category{C}, T, A, a)$ be a $\lambda_c(\Sigma)$-structure and $\emalgsymbol : T \answerobj \to \answerobj$ be an EM algebra where $\answerobj$ is an ordered object with the partial order structure given by $\{ (\category{C}_0(X, \answerobj), {\Rightarrow}_X) \}_{X \in \category{C}}$.
We have the codomain fibration $\mathrm{cod} : \category{C} / \answerobj \to \category{C}$ where the total category is the lax slice category~\cite{aguirre2020}.
We assume that $\mathrm{cod} : \category{C} / \answerobj \to \category{C}$ has simple products and simple coproducts.
That is, we have the left and the right adjoint of $\category{C}_0(\pi_1, \answerobj) : \category{C}_0(X, \answerobj) \to \category{C}_0(X \times Y, \answerobj)$ that satisfies the Beck--Chevalley condition.

\paragraph{Adjunctions.}
\[ \exists \dashv \category{C}_0(\pi_1, \answerobj) \dashv \forall \]
\begin{mathpar}
	\mprset{fraction={===}}
	\inferrule{
		f \comp \pi_1 \Rightarrow g
	}{
		f \Rightarrow \forall g
	}
	\and
	\inferrule{
		f \Rightarrow g \comp \pi_1
	}{
		\exists f \Rightarrow g
	}
\end{mathpar}

\paragraph{The Beck--Chevalley condition.}
For any $u : X \to X'$ and $f : X' \times Y \to \answerobj$,
\begin{center}
	\begin{tikzcd}
		\category{C}_0(X' \times Y, \answerobj) \ar[r, "\forall"] \ar[d, "{\category{C}_0(u \times \identity{}, \answerobj)}"] & \category{C}_0(X', \answerobj) \ar[d, "{\category{C}_0(u, \answerobj)}"] \\
		\category{C}_0(X \times Y, \answerobj) \ar[r, "\forall"] & \category{C}_0(X, \answerobj)
	\end{tikzcd}
\end{center}
\[ (\forall f) \comp u = \forall (f \comp (u \times \identity{})) \qquad (\exists f) \comp u = \exists (f \comp (u \times \identity{})) \]

We define the interpretation of quantifiers as follows
\[ \mathcal{A}^{\emalgsymbol}\interpret{\forall x : \rho. M} \coloneqq \forall \mathcal{A}^{\emalgsymbol}\interpret{M} \qquad \mathcal{A}^{\emalgsymbol}\interpret{\exists x : \rho. M} \coloneqq \exists \mathcal{A}^{\emalgsymbol}\interpret{M} \]

\subsection{Examples of Simple Products and Simple Coproducts}

Summary: Simple products (= universal quantifiers) are well-behaved in $\omegaCPO$ but simple coproducts (= existential quantifiers) are not.
Later, we consider $\nu\mathrm{HFL}(\mathbb{Z})$ in which only universal quantifiers are used.

\begin{lemma}\label{lem:forall-set}
	Suppose $\category{C} = \Set$.
	Let $\answerobj = \{ \mathbf{false}, \mathbf{true} \}$ be the ordered object defined by the standard pointwise order: for any $f, g \in \Set(X, \answerobj)$,
	\[ f \Rightarrow g \qquad\iff\qquad \forall x \in X, f(x) = \mathbf{true} \implies g(x) = \mathbf{true}. \]
	Then, simple products for $\cod : \Set / \answerobj \to \Set$ is given as follows.
	\begin{equation}
		(\forall f)(x) = \begin{cases}
			\mathbf{true} & \forall y \in Y, f(x, y) = \mathbf{true} \\
			\mathbf{false} & \text{otherwise}
		\end{cases}
		\label{eq:forall-set}
	\end{equation}
\end{lemma}
\begin{proof}
	\begin{itemize}
		\item Adjunction:
		\begin{itemize}
			\item If $f \comp \pi_1 \Rightarrow g$ and $f(x) = \mathbf{true}$ where $f : X \to \answerobj$ and $g : X \times Y \to \answerobj$, then for any $y \in Y$, we have $g(x, y) = \mathbf{true}$.
			Thus, $(\forall g)(x) = \mathbf{true}$.
			\item If $f \Rightarrow \forall g$ and $(f \comp \pi_1)(x, y) = \mathbf{true}$, then we have $(\forall g)(x) = \mathbf{true}$.
			Thus, $g(x, y) = \mathbf{true}$.
		\end{itemize}
		\item BC condition: We prove $(\forall f) \comp u = \forall (f \comp (u \times \identity{}))$.
		Let $x \in X$.
		\begin{align}
			(\forall f) (u(x)) = \mathbf{true} &\iff \forall y \in Y, f(u(x), y) = \mathbf{true} \\
			&\iff \forall (f \comp (u \times \identity{}))(x) = \mathbf{true}
		\end{align}
	\end{itemize}
\end{proof}

Simple products in $\mathbf{Poset}$ and $\omegaCPO$ are defined in the same way as $\Set$ (Lemma~\ref{lem:forall-set}).
To show this, it suffices to prove that monotonicity and Scott-continuity are preserved by ${\forall} : \category{C}_0(X \times Y, \answerobj) \to \category{C}_0(X, \answerobj)$.
\begin{lemma}
	Suppose $\category{C} = \mathbf{Poset}$.
	Let $\answerobj = (\{ \mathbf{false}, \mathbf{true} \}, {\ge})$ be the ordered object defined in the same way as Lemma~\ref{lem:forall-set}.
	(Note that $\answerobj = (\{ \mathbf{false}, \mathbf{true} \}, {\ge})$ has the opposite order as an $\omega$cpo but has the standard order as an ordered object.)
	If $f$ is monotone with respect to ${\ge}$, then so is \eqref{eq:forall-set}.
\end{lemma}
\begin{proof}
	Suppose $x \le x'$.
	We prove $(\forall f)(x) \ge (\forall f)(x')$.
	It suffices to prove that $(\forall f)(x) = \mathbf{false}$ implies $(\forall f)(x') = \mathbf{false}$.
	If $(\forall f)(x) = \mathbf{false}$, then there exists $y \in Y$ such that $f(x, y) = \mathbf{false}$.
	By the monotonicity of $f$, we have $f(x', y) = \mathbf{false}$.
	Therefore, $(\forall f)(x) = \mathbf{false}$.
\end{proof}

\begin{lemma}
	Suppose $\category{C} = \omegaCPO$.
	Let $\answerobj = (\{ \mathbf{false}, \mathbf{true} \}, {\ge})$ be the ordered object defined in the same way as Lemma~\ref{lem:forall-set}.
	If $f$ is Scott-continuous with respect to ${\ge}$, then so is \eqref{eq:forall-set}.
\end{lemma}
\begin{proof}
	Suppose we have an $\omega$-chain $\{ x_n \}_n$.
	It suffices to prove that if $(\forall f)(x_n) = \mathbf{true}$ for any $n$, then $(\forall f)(\sup_n x_n) = \mathbf{true}$.
	If $(\forall f)(x_n) = \mathbf{true}$, then for any $y \in Y$, we have $f(x_n, y) = \mathbf{true}$.
	By the Scott-continuity of $f$, we have $f(\sup_n x_n, y) = \mathbf{true}$ for any $y$.
	Therefore, we have $(\forall f)(\sup_n x_n) = \mathbf{true}$.
\end{proof}

We can define simple coproducts in $\Set$ and $\mathbf{Poset}$ in a similar way.
However, such simple coproducts do not preserve Scott-continuity.
\begin{lemma}
	Consider the same situation as Lemma~\ref{lem:forall-set}.
	The simple coproduct for $\cod : \Set / \answerobj \to \Set$ is given as follows.
	\begin{equation}
		(\exists f)(x) = \begin{cases}
			\mathbf{true} & \exists y \in Y, f(x, y) = \mathbf{true} \\
			\mathbf{false} & \text{otherwise}
		\end{cases}
		\label{eq:exists-set}
	\end{equation}
\end{lemma}
\begin{proof}
	\begin{itemize}
		\item Adjunction:
		\begin{itemize}
			\item If $f \Rightarrow g \comp \pi_1$ and $(\exists f)(x) = \mathbf{true}$, then there exists $y \in Y$ such that $f(x, y) = \mathbf{true}$.
			Therefore, $g(x) = (g \comp \pi_1)(x, y) = \mathbf{true}$.
			\item If $\exists f \Rightarrow g$ and $f(x, y) = \mathbf{true}$, then $(\exists f)(x) = \mathbf{true}$, which implies $g(x) = \mathbf{true}$.
		\end{itemize}
		\item BC condition: We prove $(\exists f) \comp u = \exists (f \comp (u \times \identity{}))$.
		Let $x \in X$.
		\begin{align}
			((\exists f) \comp u)(x) = \mathbf{true} &\iff \exists y \in Y, f(u(x), y) = \mathbf{true} \\
			&\iff \exists (f \comp (u \times \identity{}))(x) = \mathbf{true}
		\end{align}
	\end{itemize}
\end{proof}

\begin{lemma}
	Suppose $\category{C} = \mathbf{Poset}$.
	Let $\answerobj = (\{ \mathbf{false}, \mathbf{true} \}, {\ge})$ be the ordered object defined in the same way as Lemma~\ref{lem:forall-set}.
	If $f : X \times Y \to \answerobj$ is monotone, then so is $\exists f : X \to \answerobj$.
\end{lemma}
\begin{proof}
	Suppose $x \le x'$.
	We prove $(\exists f)(x) \ge (\exists f)(x')$.
	If $(\exists f)(x) = \mathbf{false}$, then for any $y \in Y$, $f(x, y) = \mathbf{false}$.
	By the monotonicity of $f$, we have $f(x', y) = \mathbf{false}$ for any $y \in Y$.
	Therefore, $(\exists f)(x') = \mathbf{false}$.
\end{proof}

\begin{lemma}
	Suppose $\category{C} = \omegaCPO$.
	Let $\answerobj = (\{ \mathbf{false}, \mathbf{true} \}, {\ge})$ be the ordered object defined in the same way as Lemma~\ref{lem:forall-set}.
	In this case, \eqref{eq:exists-set} does not preserve Scott-continuity.
\end{lemma}
\begin{proof}
	Let $f : (\{0, 1, \dots, \omega\}, {\le}) \times (\mathbb{N}, {=}) \to \answerobj$ be a function defined as follows.
	\[ f(x, y) \coloneqq \begin{cases}
		\mathbf{true} & x \le y \\
		\mathbf{false} & x > y
	\end{cases} \]
	The function $f$ is Scott-continuous, but $\exists f$ is not.
	Consider the $\omega$-chain $\{ x_n \}_n$ defined by $x_n = n \in \{ 0, 1, \dots, \omega \}$.
	For any $n$, we have $(\exists f)(x_n) = \mathbf{true}$.
	However, $(\exists f)(\sup_n x_n) = (\exists f)(\omega) = \mathbf{false}$.
\end{proof}

\subsection{Relation to $\nu\mathrm{HFL}(\mathbb{Z})$}
We compare $\nu\mathrm{HFL}(\mathbb{Z})$~\cite{katsura2020,kobayashi2018} and an instance of our target language.
There are two main differences.
The first one is the difference of syntax, which we will handle by defining a syntactic translation from $\nu\mathrm{HFL}(\mathbb{Z})$ to $\lambda_{\mathrm{HFL}}$.
The second one is the difference of semantic models.
The interpretation of $\nu\mathrm{HFL}(\mathbb{Z})$ is defined in $\mathbf{Poset}$ while the interpretation of $\lambda_{\mathrm{HFL}}$ is defined in $\omegaCPO$.
We relate these interpretations by considering logical relations.
As a result, we prove that given a closed $\nu\mathrm{HFL}(\mathbb{Z})$-formula, we can translate it to a $\lambda_{\mathrm{HFL}}(\Sigma)$-term whose validity coincides with the original $\nu\mathrm{HFL}(\mathbb{Z})$-formula (Theorem~\ref{thm:nuHFL-as-instance}).

We consider the following instance of $\lambda_{\mathrm{HFL}}$.
Note that the definitions are almost the same as those for partial correctness.
\begin{itemize}
	\item $\lambda_c$-signature: $\Sigma = (B, K, O)$ where
	\begin{itemize}
		\item $B = \{ \mathbf{int} \}$,
		\item $K$ consists of basic binary operations on integers $\mathrm{op} : \mathbf{int} \times \mathbf{int} \rightarrowtriangle \mathbf{int}$ and integer constants $n : 1 \rightarrowtriangle \mathbf{int}$ for each $n \in \mathbb{Z}$,
		\item $O = \emptyset$.
	\end{itemize}
	\item $\lambda_c(\Sigma)$-structure: $\mathcal{A} = (\omegaCPO, ({-})_{\bot}, A, a)$ where
	\begin{itemize}
		\item $A(\mathbf{int}) = (\mathbb{Z}, {=})$
		\item $a$ gives natural interpretations of operations on integers.
	\end{itemize}
	\item EM algebra: $\answerobj = (\{ \mathbf{false}, \mathbf{true} \}, {\ge})$ (same as partial correctness).
	\item $\lambda_{\mathrm{HFL}}(\Sigma)$-terms are extended by basic predicates on integers, the distributive lattice structure, and universal quantifiers.
\end{itemize}

Translation of types:
Recall that types for $\nu\mathrm{HFL}(\mathbb{Z})$ is defined as follows.
\[ \rho \coloneqq \bullet \mid \eta \to \rho \qquad \eta \coloneqq \rho \mid \mathbf{int} \]
We translate this to $\lambda_{\mathrm{HFL}}(\Sigma)$-types as follows.
\begin{itemize}
	\item For $\rho$, we define a list $\rho^{\natural}$ of $\lambda_{\mathrm{HFL}}(\Sigma)$-types as follows.
	\begin{itemize}
		\item $\bullet^{\natural} = []$
		\item $(\eta \to \rho)^{\natural} = \eta^{\sharp} :: \rho^{\natural}$
	\end{itemize}
	\item We define $\rho^{\sharp} = \prod \rho^{\natural} \to \answertype$ where $\prod [] = 1$ and $\prod (\rho :: l) = \rho \times \prod l$.
	\item For $\eta$, we define $\eta^{\sharp}$ as follows.
	\begin{itemize}
		\item If $\eta = \rho$, $\eta^{\sharp} = \rho^{\sharp}$.
		\item If $\eta = \mathbf{int}$, $\eta^{\sharp} = \mathbf{int}$.
	\end{itemize}
\end{itemize}

Then, we can translate well-typed terms of $\nu\mathrm{HFL}(\mathbb{Z})$ to well-typed $\lambda_{\mathrm{HFL}}(\Sigma)$-terms.
\[ \Gamma \vdash \psi : \eta \qquad\mapsto\qquad \Gamma^{\sharp} \vdash \psi^{\sharp} : \eta^{\sharp} \]

\begin{itemize}
	\item $(\nu X : \rho. \psi)^{\sharp} = \letrec{X}{x}{\psi^{\sharp}\ x}{X}$
	\begin{mathpar}
		\inferrule{
			\Gamma, X : \rho \vdash \psi : \rho
		}{
			\Gamma \vdash \nu X : \rho. \psi : \rho
		}
	\end{mathpar}
	\begin{mathpar}
		\inferrule{
			\Gamma^{\sharp}, X : \prod \rho^{\natural} \to \answertype \vdash \psi^{\sharp} : \prod \rho^{\natural} \to \answertype
		}{
			\Gamma^{\sharp} \vdash \letrec{X}{x}{\psi^{\sharp}\ x}{X} : \rho^{\sharp}
		}
	\end{mathpar}
	\item $(\lambda X : \eta. \psi)^{\sharp} = \lambda (X, Y) : \eta^{\sharp} \times \prod \rho^{\natural}. \psi^{\sharp}\ Y$
	\begin{mathpar}
		\inferrule{
			\Gamma, X : \eta \vdash \psi : \rho
		}{
			\Gamma \vdash \lambda X : \eta. \psi : \eta \to \rho
		}
	\end{mathpar}
	\begin{mathpar}
		\inferrule{
			\Gamma^{\sharp}, X : \eta^{\sharp} \vdash \psi : \prod \rho^{\natural} \to \answertype
		}{
			\Gamma^{\sharp} \vdash \lambda (X, Y) : \eta^{\sharp} \times \prod \rho^{\natural}. \psi^{\sharp}\ Y : (\eta \to \rho)^{\sharp}
		}
	\end{mathpar}
	\item $(\psi_1\ \psi_2)^{\sharp} = \lambda Y : \prod \rho^{\natural}. \psi_1^{\sharp}\ (\psi_2^{\sharp}, Y)$
	\begin{mathpar}
		\inferrule{
			\Gamma \vdash \psi_1 : \eta \to \rho \\
			\Gamma \vdash \psi_2 : \eta
		}{
			\Gamma \vdash \psi_1\ \psi_2 : \rho
		}
	\end{mathpar}
	\begin{mathpar}
		\inferrule{
			\Gamma^{\sharp} \vdash \psi_1^{\sharp} : \eta^{\sharp} \times \prod \rho^{\natural} \to \answertype \\
			\Gamma^{\sharp} \vdash \psi_2^{\sharp} : \eta^{\sharp}
		}{
			\Gamma^{\sharp} \vdash \lambda Y : \prod \rho^{\natural}. \psi_1^{\sharp}\ (\psi_2^{\sharp}, Y) : \rho^{\sharp}
		}
	\end{mathpar}
	\item $(\forall X : \mathbf{int}. \psi)^{\sharp} = \lambda Y : 1. \forall X : \mathbf{int}. \psi^{\sharp}\ Y$
	\begin{mathpar}
		\inferrule{
			\Gamma, X : \mathbf{int} \vdash \psi : \bullet
		}{
			\Gamma \vdash \forall X : \mathbf{int}. \psi : \bullet
		}
	\end{mathpar}
	\begin{mathpar}
		\inferrule{
			\Gamma^{\sharp}, X : \mathbf{int} \vdash \psi^{\sharp} : 1 \to \answertype
		}{
			\Gamma^{\sharp} \vdash \lambda Y : 1. \forall X : \mathbf{int}. \psi^{\sharp}\ Y : \bullet^{\sharp}
		}
	\end{mathpar}
	\item $\mathbf{true}^{\sharp} = \lambda X : 1. \mathbf{true}$
	\begin{mathpar}
		\inferrule{ }{
			\Gamma \vdash \mathbf{true} : \bullet
		}
	\end{mathpar}
	\begin{mathpar}
		\inferrule{ }{
			\Gamma^{\sharp} \vdash \lambda X : 1. \mathbf{true} : \bullet^{\sharp}
		}
	\end{mathpar}
	\item $\mathbf{false}^{\sharp} = \lambda X : 1. \mathbf{false}$
	\item $(\psi_1 \land \psi_2)^{\sharp} = \lambda X : 1. \psi_1^{\sharp}\ X \land \psi_2^{\sharp}\ X$
	\begin{mathpar}
		\inferrule{
			\Gamma \vdash \psi_1 : \bullet \\
			\Gamma \vdash \psi_2 : \bullet
		}{
			\Gamma \vdash \psi_1 \land \psi_2 : \bullet
		}
	\end{mathpar}
	\begin{mathpar}
		\inferrule{
			\Gamma^{\sharp} \vdash \psi_1^{\sharp} : 1 \to \answertype \\
			\Gamma^{\sharp} \vdash \psi_2^{\sharp} : 1 \to \answertype
		}{
			\Gamma^{\sharp} \vdash \lambda X : 1. \psi_1^{\sharp}\ X \land \psi_2^{\sharp}\ X : \bullet^{\sharp}
		}
	\end{mathpar}
	\item $(\psi_1 \lor \psi_2)^{\sharp} = \lambda X : 1. \psi_1^{\sharp}\ X \lor \psi_2^{\sharp}\ X$
	\item $X^{\sharp} = X$
	\item $(p(\psi_1, \dots, \psi_k))^{\sharp} = \lambda X : 1. p(\psi_1^{\sharp}, \dots, \psi_k^{\sharp})$
	\begin{mathpar}
		\inferrule{
			\Gamma \vdash \psi_i : \mathbf{int}
		}{
			\Gamma \vdash p(\psi_1, \dots, \psi_k) : \bullet
		}
	\end{mathpar}
	\begin{mathpar}
		\inferrule{
			\Gamma^{\sharp} \vdash \psi_i^{\sharp} : \mathbf{int}
		}{
			\Gamma^{\sharp} \vdash \lambda X : 1. p(\psi_1^{\sharp}, \dots, \psi_k^{\sharp}) : \bullet^{\sharp}
		}
	\end{mathpar}
	\item $n^{\sharp} = n$
	\item $(\psi_1 \mathrel{\mathrm{op}} \psi_2)^{\sharp} = \psi_1^{\sharp} \mathrel{\mathrm{op}} \psi_2^{\sharp}$
\end{itemize}

\begin{definition}
	We define a functor $({-})^{\op} : \mathbf{Poset} \to \mathbf{Poset}$ by $(X, {\le}_X)^{\op} \coloneqq (X, {\ge}_X)$ and $f^{\op} \coloneqq f$ for any $(X, {\le}_X), (Y, {\le}_Y) \in \mathbf{Poset}$ and $f : (X, {\le}_X) \to (Y, {\le}_Y)$.
\end{definition}

\begin{lemma}
	The functor $({-})^{\op} : \mathbf{Poset} \to \mathbf{Poset}$ preserves the cc structure.
	\[ 1^{\op} = 1 \qquad (X \times Y)^{\op} = X^{\op} \times Y^{\op} \qquad (\exponential{X}{Y})^{\op} = \exponential{X^{\op}}{Y^{\op}} \]
\end{lemma}

\begin{definition}
	We define the interpretation of $\nu\mathrm{HFL}(\mathbb{Z})$ in $\mathbf{Poset}$ as follows.
	\begin{gather}
		\interpret{\bullet} \coloneqq (\{ \mathbf{false}, \mathbf{true} \}, {\ge}) \\
		\interpret{\eta \to \rho} \coloneqq \exponential{\interpret{\eta}}{\interpret{\rho}} \\
		\interpret{\mathbf{int}} \coloneqq (\mathbb{Z}, {=})
	\end{gather}
	\begin{gather}
		\interpret{\psi_1 \mathrel{\mathrm{op}} \psi_2} \coloneqq \interpret{\mathrm{op}} \comp \tupling{\interpret{\psi_1}}{\interpret{\psi_2}} \\
		\interpret{p(\psi_1, \dots, \psi_k)} \coloneqq \interpret{p} \comp \langle \interpret{\psi_1}, \dots, \interpret{\psi_k} \rangle \\
		\interpret{\nu X : \rho. \psi}(\gamma) \coloneqq \mathrm{lfp}\interpret{\psi}(\gamma, {-})
	\end{gather}
\end{definition}
Note that $\interpret{\psi}^{\op}$ is the same as the interpretation used in, e.g., \cite{katsura2020}.
Note also that the existence of the least fixed point in the definition above is guaranteed by the fact that $\interpret{\rho}$ is a complete lattice, which can be easily proved by induction.
We can further prove that $\interpret{\psi}$ is Scott-continuous.
Thus,
\[ \interpret{\nu X : \rho. \psi}(\gamma) = \mathrm{lfp}\interpret{\psi}(\gamma, {-}) = \sup_n (\interpret{\psi}(\gamma, {-}))^n(\bot) \]

\begin{definition}
	interpretation of $\nu\mathrm{HFL}(\mathbb{Z})$ in $\omegaCPO$
	\begin{gather}
		\interpret{\bullet}' \coloneqq (\{ \mathbf{false}, \mathbf{true} \}, {\ge}) \\
		\interpret{\eta \to \rho}' \coloneqq \exponential{\interpret{\eta}'}{\interpret{\rho}'} \\
		\interpret{\mathbf{int}}' \coloneqq (\mathbb{Z}, {=})
	\end{gather}
	\begin{gather}
		\interpret{\nu X : \rho. \psi}' \coloneqq (\interpret{\psi}')^{\dagger} \qquad
		\interpret{\forall X : \mathbf{int}. \psi}' \coloneqq \forall \interpret{\psi}' \\
		\interpret{\lambda X : \rho. \psi} \coloneqq \Lambda (\interpret{\psi}') \qquad
		\interpret{\psi_1\ \psi_2}' \coloneqq \eval \comp \tupling{\interpret{\psi_1}'}{\interpret{\psi_2}'} \\
		\interpret{p(\psi_1, \dots, \psi_k)}' \coloneqq \interpret{p} \comp \langle \interpret{\psi_1}', \dots, \interpret{\psi_k}' \rangle \\
		\interpret{\psi_1 \mathrel{\mathrm{op}} \psi_2}' \coloneqq \interpret{\mathrm{op}} \comp \tupling{\interpret{\psi_1}'}{\interpret{\psi_2}'} \qquad
		\interpret{n}' \coloneqq n \comp {!} \\
		\interpret{\mathbf{true}}' \coloneqq \top \comp {!} \qquad
		\interpret{\mathbf{false}}' \coloneqq \bot \comp {!} \\
		\interpret{\psi_1 \land \psi_2}' \coloneqq {\land} \comp \tupling{\interpret{\psi_1}'}{\interpret{\psi_2}'} \qquad
		\interpret{\psi_1 \lor \psi_2}' \coloneqq {\lor} \comp \tupling{\interpret{\psi_1}'}{\interpret{\psi_2}'} \\
		\interpret{\Gamma, Y : \eta' \vdash X : \eta}' \coloneqq \interpret{\Gamma \vdash X : \eta}' \comp \pi_1 \qquad
		\interpret{\Gamma, X : \eta \vdash X : \eta}' \coloneqq \pi_2
	\end{gather}
\end{definition}

Note that the forgetful functor $U : \omegaCPO \to \mathbf{Poset}$ does not preserves exponentials because for any $X, Y \in \omegaCPO$, the exponential object $\exponential{X}{Y}$ in $\omegaCPO$ is the set of Scott-continuous functions $f : X \to Y$ whereas the exponential object $\exponential{X}{Y}$ in $\mathbf{Poset}$ is the set of monotone functions $f : X \to Y$.

\begin{remark}
	Existential quantifiers are ill-behaved in terms of Scott-continuity (with respect to the reversed order on $\{ \mathbf{false}, \mathbf{true} \}$).
	For example, consider the following formula.
	\[ F : \mathbf{int} \to \bullet \vdash \exists X : \mathbf{int}. F\ X : \bullet \]
	For each $n$, we define an environment $\gamma_n$ as follows.
	\[ \gamma_n(F) \coloneqq \lambda x. |x| \ge n \]
	Then, $\{ \gamma_n \}_n$ is an $\omega$-chain with respect to the reversed order.
	We have $\interpret{\exists X : \mathbf{int}. F\ X}(\gamma_n) = \mathbf{true}$ for any $n$.
	However, since $(\sup_n \gamma_n)(F) = \lambda x. \mathbf{false}$, we have $\interpret{\exists X : \mathbf{int}. F\ X}(\sup_n \gamma_n) = \mathbf{false}$.
	This also gives an evidence of $\interpret{(\mathbf{int} \to \bullet) \to \bullet} \neq \interpret{(\mathbf{int} \to \bullet) \to \bullet}'$.
\end{remark}

\begin{proposition}\label{prop:nuHFL-cpo-poset}
	For any well-typed $\nu\mathrm{HFL}(\mathbb{Z})$-term $\vdash \psi : \bullet$,
	\[ \interpret{\psi}(\emptyset) = \interpret{\psi}'(\emptyset) \]
\end{proposition}
\begin{proof}
	Use sconing:
	\begin{center}
		\begin{tikzcd}[column sep=10em]
			\category{K} \ar[rd, phantom, very near start, "\lrcorner"] \ar[r] \ar[d, "p"] & \mathbf{Sub}(\Set) \ar[d] \\
			\omegaCPO \times \mathbf{Poset} \ar[r, "{(\omegaCPO \times \mathbf{Poset})((1, 1), {-})}"] & \Set
		\end{tikzcd}
	\end{center}
	The total category $\category{K}$ is cartesian closed and $p$ strictly preserves the cc-structure~\cite{hermida1993}.

	We define the interpretation $\interpret{-}''$ in $\category{K}$ by
	\begin{gather}
		\interpret{\bullet}'' \coloneqq ((2, {\ge}), (2, {\ge}), \{ (f, f) \mid \text{$f : 1 \to 2$ is a function} \}) \\
		\interpret{\mathbf{int}}'' \coloneqq (\mathbb{Z}, \mathbb{Z}, \{ (f, f) \mid f : 1 \to \mathbb{Z} \}) \qquad
		\interpret{\eta \to \rho}'' \coloneqq \dotExponential{\interpret{\eta}''}{\interpret{\rho}''}
	\end{gather}
	where $2 = \{ \mathbf{false}, \mathbf{true} \}$ and $\mathbb{Z} = (\mathbb{Z}, {=})$.

	\begin{itemize}
		\item Basic predicates: We have $(\interpret{p}, \interpret{p}) : (\interpret{\mathbf{int}}'')^k \dotTo \interpret{\bullet}''$.
		\item Basic operations and constants for $\mathbf{int}$: We have $(\interpret{\mathrm{op}}, \interpret{\mathrm{op}}') : (\interpret{\mathbf{int}}'')^2 \to \interpret{\mathbf{int}}''$ and $(n, n) : 1 \to \interpret{\mathbf{int}}''$ for each $n \in \mathbb{Z}$.
		\item True/False/And/Or: We have $({\land}, {\land}), ({\lor}, {\lor}) : (\interpret{\bullet}'')^2 \to \interpret{\bullet}''$ and $\mathbf{true}, \mathbf{false} : 1 \to \interpret{\bullet}''$.
		\item Universal quantifiers: We have a mapping $\forall : \category{K}((X, Y, R) \times \interpret{\mathbf{int}}'', \interpret{\bullet}'') \to \category{K}((X, Y, R), \interpret{\bullet}'')$ such that the following diagram commutes.
		\begin{center}
			\begin{tikzcd}
				\category{K}((X, Y, R) \times \interpret{\mathbf{int}}'', \interpret{\bullet}'') \ar[r, "\forall"] \ar[d, "p"] & \category{K}((X, Y, R), \interpret{\bullet}'') \ar[d, "p"] \\
				\omegaCPO(X \times \mathbb{Z}, (2, {\ge})) \times \mathbf{Poset}(Y \times \mathbb{Z}, (2, {\ge})) \ar[r, "\forall \times \forall"] & \omegaCPO(X, (2, {\ge})) \times \mathbf{Poset}(Y, (2, {\ge}))
			\end{tikzcd}
		\end{center}
		This is proved as follows.
		Let $(f, g) \in \category{K}((X, Y, R) \times \interpret{\mathbf{int}}'', \interpret{\bullet}'')$.
		That is, $f \in \omegaCPO(X \times \mathbb{Z}, (2, {\ge}))$ and $g \in \mathbf{Poset}(Y \times \mathbb{Z}, (2, {\ge}))$; and for any $(x, y) \in R$ and $n : 1 \to \mathbb{Z}$, we have $f \comp \tupling{x}{n} = g \comp \tupling{y}{n}$.
		Then, we have $\forall (x, y) \in R, (\forall f) \comp x = (\forall g) \comp y$ because
		\begin{itemize}
			\item if $(\forall f) \comp x = \mathbf{true}$, then for any $n : 1 \to \mathbb{Z}$, we have $g \comp \tupling{y}{n} = f \comp \tupling{x}{n} = \mathbf{true}$, thus $(\forall g) \comp y = \mathbf{true}$,
			\item if $(\forall f) \comp x = \mathbf{false}$, then there exists $n : 1 \to \mathbb{Z}$ such that $g \comp \tupling{y}{n} = f \comp \tupling{x}{n} = \mathbf{false}$, thus $(\forall g) \comp y = \mathbf{false}$.
		\end{itemize}
		Note that this proof depends on the definition of $\interpret{\mathbf{int}}''$.
		Note that $\forall : \category{K}((X, Y, R) \times \interpret{\mathbf{int}}'', \interpret{\bullet}'') \to \category{K}((X, Y, R), \interpret{\bullet}'')$ satisfies the Beck--Chevalley condition because $\forall \times \forall$ in the base category satisfies the BC condition.
		\item Note that $\interpret{\rho}''$ is admissible with respect to the reversed order $(2, {\ge})$.
		This is proved by induction.
		The base case is obvious.
		The step case follows because $\dotExponential{(X, Y, R)}{(X', Y', R')} = (\exponential{X}{X'}, \exponential{Y}{Y'}, \{ (f, g) \mid \forall (x, y) \in R, (\eval \comp \tupling{f}{x}, \eval \comp \tupling{g}{y}) \in R' \})$ is admissible if $(X', Y', R')$ is admissible.

		Let $\Gamma, X : \rho \vdash \psi : \rho$ be a well-typed term.
		Since we have
		\[ \interpret{\nu X : \rho. \psi}(\gamma) = \sup_n (\interpret{\psi}(\gamma, {-}))^n(\bot) \qquad \interpret{\nu X : \rho. \psi}'(\gamma) = \sup_n (\interpret{\psi}'(\gamma, {-}))^n(\bot) \]
		and $\interpret{\rho}''$ is admissible, we have $(\interpret{\nu X : \rho. \psi}, \interpret{\nu X : \rho. \psi}') : \interpret{\Gamma}'' \dotTo \interpret{\rho}''$.
		\item For other term constructions, we use cartesian closed structure of $\category{K}$.
	\end{itemize}
	Therefore, for any well-typed term $\Gamma \vdash \psi : \eta$, we have $(\interpret{\psi}', \interpret{\psi}) : \interpret{\Gamma}'' \dotTo \interpret{\eta}''$.
	Specifically, for any $\vdash \psi : \bullet$, we have $(\interpret{\psi}', \interpret{\psi}) : \dot{1} \dotTo \interpret{\bullet}''$, which implies $\interpret{\psi}' = \interpret{\psi}$ as functions.
\end{proof}

We define $\mathbf{curry}$ and $\mathbf{uncurry}$ as follows.
\begin{lemma}
	The following morphisms are mutually inverse.
	\begin{align}
		\mathbf{curry} &\coloneqq \Lambda \Lambda (\eval \comp \associator) &&: \exponential{X \times Y}{A} \to \exponential{X}{\exponential{Y}{A}} \\
		\mathbf{uncurry} &\coloneqq \Lambda(\eval \comp (\eval \times \identity{}) \comp \associator^{-1}) &&: \exponential{X}{\exponential{Y}{A}} \to \exponential{X \times Y}{A}
	\end{align}
\end{lemma}
\begin{proof}
	\begin{align}
		\mathbf{curry} \comp \mathbf{uncurry} &= \Lambda \Lambda (\eval \comp \associator \comp ((\Lambda(\eval \comp (\eval \times \identity{}) \comp \associator^{-1}) \times \identity{}) \times \identity{})) \\
		&= \Lambda \Lambda (\eval \comp (\eval \times \identity{}) \comp \associator^{-1} \comp \associator) \\
		&= \Lambda (\Lambda (\eval) \comp \eval) \\
		&= \Lambda (\eval) \\
		&= \identity{}
	\end{align}
	\begin{align}
		\mathbf{uncurry} \comp \mathbf{curry} &= \Lambda(\eval \comp (\eval \times \identity{}) \comp \associator^{-1} \comp (\Lambda \Lambda (\eval \comp \associator) \times \identity{})) \\
		&= \Lambda(\eval \comp (\Lambda (\eval \comp \associator) \times \identity{}) \comp \associator^{-1}) \\
		&= \Lambda(\eval \comp \associator \comp \associator^{-1}) \\
		&= \identity{}
	\end{align}
\end{proof}

\begin{lemma}\label{lem:curry-EM-algebra-morphism}
	For any EM $T$-algebra $\alpha : T A \to A$ and $X, Y \in \category{C}$, $\mathbf{uncurry} : \exponential{X}{\exponential{Y}{A}} \to \exponential{X \times Y}{A}$ is an isomorphism of EM algebras.
\end{lemma}
\begin{proof}
	Let $\strength'^{T} \coloneqq T \braiding \comp \strength^T \comp \braiding$.
	\begin{align}
		&\Lambda(\alpha \comp T \eval \comp \strength'^{T}) \comp T \mathbf{uncurry} \\
		&= \Lambda(\alpha \comp T \eval \comp \strength'^{T} \comp (T \mathbf{uncurry} \times \identity{})) \\
		&= \Lambda(\alpha \comp T \eval \comp T (\mathbf{uncurry} \times \identity{}) \comp \strength'^{T}) \\
		&= \Lambda(\alpha \comp T (\eval \comp (\eval \times \identity{}) \comp \associator^{-1}) \comp \strength'^{T})
	\end{align}
	\begin{align}
		&\mathbf{uncurry} \comp \Lambda(\Lambda(\alpha \comp T \eval \comp \strength'^{T}) \comp T \eval \comp \strength'^{T}) \\
		&= \Lambda(\eval \comp (\eval \times \identity{}) \comp \associator^{-1} \comp (\Lambda(\Lambda(\alpha \comp T \eval \comp \strength'^{T}) \comp T \eval \comp \strength'^{T}) \times \identity{})) \\
		&= \Lambda(\eval \comp ((\Lambda(\alpha \comp T \eval \comp \strength'^{T}) \comp T \eval \comp \strength'^{T}) \times \identity{}) \comp \associator^{-1}) \\
		&= \Lambda(\alpha \comp T \eval \comp \strength'^{T} \comp ((T \eval \comp \strength'^{T}) \times \identity{}) \comp \associator^{-1}) \\
		&= \Lambda(\alpha \comp T \eval \comp T (\eval \times \identity{}) \comp \strength'^{T} \comp (\strength'^{T} \times \identity{}) \comp \associator^{-1}) \\
		&= \Lambda(\alpha \comp T \eval \comp T (\eval \times \identity{}) \comp T \associator^{-1} \comp \strength'^{T})
	\end{align}
\end{proof}

\begin{proposition}\label{prop:nuHFL-translation-sound}
	For any well-typed $\nu\mathrm{HFL}(\mathbb{Z})$-term $\Gamma \vdash \psi : \eta$,
	\[ \kappa_{\eta} \comp \interpret{\psi}' = \mathcal{A}^{\emalgsymbol} \interpret{\psi^{\sharp}} \comp \kappa_{\Gamma} \]
	where $\kappa_{\eta} : \interpret{\eta}' \to \mathcal{A}^{\emalgsymbol}\interpret{\eta^{\sharp}}$ is a canonical isomorphism.
	\begin{itemize}
		\item $\kappa_{\bullet} : \answerobj \to \exponential{1}{\answerobj}$
		\[ \kappa_{\bullet} \coloneqq \Lambda(\pi_1) \qquad \kappa_{\bullet}^{-1} \coloneqq \eval \comp \tupling{\identity{}}{{!}} \]
		\item $\kappa_{\eta \to \rho} : \exponential{\interpret{\eta}'}{\interpret{\rho}'} \to \mathcal{A}^{\emalgsymbol} \exponential{\interpret{\eta^{\sharp}} \times \mathcal{A}^{\emalgsymbol} \interpret{\prod \rho^{\natural}}}{\answerobj}$
		\[ \kappa_{\eta \to \rho} \coloneqq \mathbf{uncurry} \comp (\exponential{\kappa_{\eta}^{-1}}{\kappa_{\rho}}) \qquad \kappa_{\eta \to \rho}^{-1} \coloneqq (\exponential{\kappa_{\eta}}{\kappa_{\rho}^{-1}}) \comp \mathbf{curry} \]
		\item $\kappa_{\mathbf{int}} = \identity{}$
	\end{itemize}
\end{proposition}
\begin{proof}
	\allowdisplaybreaks
	\begin{itemize}
		\item $\nu X : \rho. \psi$
		\begin{align}
			&\mathcal{A}^{\emalgsymbol}\interpret{\letrec{X}{x}{\psi^{\sharp}\ x}{X}} \comp \kappa_{\Gamma} \\
			&= (\Lambda (\eval \comp \tupling{\mathcal{A}^{\emalgsymbol}\interpret{\psi^{\sharp}} \comp \pi_1}{\pi_2}))^{\dagger} \comp \kappa_{\Gamma} \\
			&= (\mathcal{A}^{\emalgsymbol}\interpret{\psi^{\sharp}})^{\dagger} \comp \kappa_{\Gamma} \\
			&= (\mathcal{A}^{\emalgsymbol}\interpret{\psi^{\sharp}} \comp (\kappa_{\Gamma} \times \identity{}))^{\dagger} \\
			&= (\kappa_{\rho} \comp \interpret{\psi}' \comp (\identity{} \times \kappa_{\rho}^{-1}))^{\dagger} \\
			&= \kappa_{\rho} \comp (\interpret{\psi}')^{\dagger} \\
			&= \kappa_{\rho} \comp \interpret{\nu X : \rho. \psi}'
		\end{align}
		Note that $\kappa_{\rho}$ is a morphism of EM algebras by Lemma~\ref{lem:curry-EM-algebra-morphism}.
		\item $\Gamma \vdash \lambda X : \eta. \psi : \eta \to \rho$
		\begin{align}
			&\kappa_{\eta \to \rho} \comp \interpret{\lambda X : \eta. \psi}' \\
			&= \kappa_{\eta \to \rho} \comp \Lambda \interpret{\psi}' \\
			&= \kappa_{\eta \to \rho} \comp \Lambda (\kappa_{\rho}^{-1} \comp \mathcal{A}^{\emalgsymbol} \interpret{\psi^{\sharp}} \comp \kappa_{\Gamma, X : \eta}) \\
			&= \Lambda(\eval \comp (\eval \times \identity{}) \comp \associator^{-1}) \comp (\exponential{\kappa_{\eta}^{-1}}{\kappa_{\rho}}) \comp \Lambda (\kappa_{\rho}^{-1} \comp \mathcal{A}^{\emalgsymbol} \interpret{\psi^{\sharp}} \comp \kappa_{\Gamma, X : \eta}) \\
			&= \Lambda(\eval \comp (\eval \times \identity{}) \comp \associator^{-1}) \comp \Lambda (\kappa_{\rho} \comp \kappa_{\rho}^{-1} \comp \mathcal{A}^{\emalgsymbol} \interpret{\psi^{\sharp}} \comp \kappa_{\Gamma, X : \eta} \comp (\identity{} \times \kappa_{\eta}^{-1})) \\
			&= \Lambda(\eval \comp (\eval \times \identity{}) \comp \associator^{-1}) \comp \Lambda (\mathcal{A}^{\emalgsymbol} \interpret{\psi^{\sharp}} \comp (\kappa_{\Gamma} \times \identity{})) \\
			&= \Lambda(\eval \comp (\eval \times \identity{}) \comp \associator^{-1}) \comp \Lambda (\mathcal{A}^{\emalgsymbol} \interpret{\psi^{\sharp}}) \comp \kappa_{\Gamma} \\
			&= \Lambda(\eval \comp (\mathcal{A}^{\emalgsymbol} \interpret{\psi^{\sharp}} \times \identity{}) \comp \associator^{-1}) \comp \kappa_{\Gamma}
		\end{align}
		\begin{align}
			\mathcal{A}^{\emalgsymbol}\interpret{\lambda (X, Y) : \eta^{\sharp} \times \prod \rho^{\natural}. \psi^{\sharp}\ Y} &= \Lambda (\eval \comp \tupling{\mathcal{A}^{\emalgsymbol}\interpret{\psi^{\sharp}} \comp (\identity{} \times \pi_1)}{\pi_2 \comp \pi_2}) \\
			&= \Lambda (\eval \comp (\mathcal{A}^{\emalgsymbol}\interpret{\psi^{\sharp}} \times \identity{}) \comp \associator^{-1})
		\end{align}
		\item $\Gamma \vdash \psi_1\ \psi_2 : \rho$
		\begin{align}
			&\mathcal{A}^{\emalgsymbol}\interpret{\lambda Y : \prod \rho^{\natural}. \psi_1^{\sharp}\ (\psi_2^{\sharp}, Y)} \comp \kappa_{\Gamma} \\
			&= \Lambda (\eval \comp \tupling{\mathcal{A}^{\emalgsymbol}\interpret{\psi_1^{\sharp}} \comp \pi_1}{\tupling{\mathcal{A}^{\emalgsymbol}\interpret{\psi_2^{\sharp}} \comp \pi_1}{\pi_2}}) \comp \kappa_{\Gamma} \\
			&= \Lambda (\eval \comp \tupling{\mathcal{A}^{\emalgsymbol}\interpret{\psi_1^{\sharp}} \comp \pi_1}{\mathcal{A}^{\emalgsymbol}\interpret{\psi_2^{\sharp}} \times \identity{}} \comp (\kappa_{\Gamma} \times \identity{})) \\
			&= \Lambda (\eval \comp \tupling{\kappa_{\eta \to \rho} \comp \interpret{\psi_1}' \comp \pi_1}{(\kappa_{\eta} \comp \interpret{\psi_2}') \times \identity{}}) \\
			&= \Lambda (\eval \comp (\kappa_{\eta \to \rho} \times (\kappa_{\eta} \times \identity{})) \comp \tupling{\interpret{\psi_1}' \comp \pi_1}{\interpret{\psi_2}' \times \identity{}}) \\
			&= \Lambda (\eval \comp (\eval \times \identity{}) \comp \associator^{-1} \comp ((\exponential{\kappa_{\eta}^{-1}}{\kappa_{\rho}}) \times (\kappa_{\eta} \times \identity{})) \comp \tupling{\interpret{\psi_1}' \comp \pi_1}{\interpret{\psi_2}' \times \identity{}}) \\
			&= \Lambda (\eval \comp (\eval \times \identity{}) \comp (((\exponential{\kappa_{\eta}^{-1}}{\kappa_{\rho}}) \times \kappa_{\eta}) \times \identity{}) \comp ((\interpret{\psi_1}' \times \interpret{\psi_2}') \times \identity{}) \comp \associator^{-1} \comp \tupling{\pi_1}{\identity{}}) \\
			&= \Lambda (\eval \comp (\kappa_{\rho} \times \identity{}) \comp (\eval \times \identity{}) \comp ((\interpret{\psi_1}' \times \interpret{\psi_2}') \times \identity{}) \comp (\tupling{\identity{}}{\identity{}} \times \identity{})) \\
			&= \Lambda (\eval) \comp \kappa_{\rho} \comp \eval \comp \tupling{\interpret{\psi_1}'}{\interpret{\psi_2}'} \\
			&= \kappa_{\rho} \comp \interpret{\psi_1\ \psi_2}' \\
		\end{align}
		\item $\Gamma \vdash \forall X : \mathbf{int}. \psi : \bullet$
		\begin{align}
			&\mathcal{A}^{\emalgsymbol}\interpret{\lambda Y : 1. \forall X : \mathbf{int}. \psi^{\sharp}\ Y} \comp \kappa_{\Gamma} \\
			&= \Lambda (\forall (\eval \comp \tupling{\mathcal{A}^{\emalgsymbol}\interpret{\psi^{\sharp}} \comp (\pi_1 \times \identity{})}{\pi_2 \comp \pi_1})) \comp \kappa_{\Gamma} \\
			&= \Lambda (\forall (\eval \comp \tupling{\mathcal{A}^{\emalgsymbol}\interpret{\psi^{\sharp}} \comp (\pi_1 \times \identity{})}{\pi_2 \comp \pi_1}) \comp (\kappa_{\Gamma} \times \identity{})) \\
			&= \Lambda (\forall (\eval \comp \tupling{\mathcal{A}^{\emalgsymbol}\interpret{\psi^{\sharp}} \comp (\pi_1 \times \identity{})}{\pi_2 \comp \pi_1} \comp ((\kappa_{\Gamma} \times \identity{}) \times \identity{}))) & \text{BC condition} \\
			&= \Lambda (\forall (\eval \comp \tupling{\mathcal{A}^{\emalgsymbol}\interpret{\psi^{\sharp}} \comp (\kappa_{\Gamma} \times \identity{}) \comp (\pi_1 \times \identity{})}{\pi_2 \comp \pi_1})) \\
			&= \Lambda (\forall (\eval \comp \tupling{\kappa_{\bullet} \comp \interpret{\psi}' \comp (\pi_1 \times \identity{})}{\pi_2 \comp \pi_1})) \\
			&= \Lambda (\forall (\eval \comp (\kappa_{\bullet} \times \identity{}) \comp (\interpret{\psi}' \times \identity{}) \comp \tupling{\pi_1 \times \identity{}}{\pi_2 \comp \pi_1})) \\
			&= \Lambda (\forall (\pi_1 \comp (\interpret{\psi}' \times \identity{}) \comp \tupling{\pi_1 \times \identity{}}{\pi_2 \comp \pi_1})) \\
			&= \Lambda (\forall (\interpret{\psi}' \comp (\pi_1 \times \identity{}))) \\
			&= \Lambda (\forall (\interpret{\psi}') \comp \pi_1) & \text{BC condition} \\
			&= \Lambda (\pi_1 \comp (\interpret{\forall X : \mathbf{int}. \psi}' \times \identity{})) \\
			&= \kappa_{\bullet} \comp \interpret{\forall X : \mathbf{int}. \psi}'
		\end{align}
		\item $\Gamma \vdash \mathbf{true} : \bullet$ (and similarly for $\Gamma \vdash \mathbf{false} : \bullet$)
		\begin{align}
			\mathcal{A}^{\emalgsymbol}\interpret{\lambda X : 1. \mathbf{true}} \comp \kappa_{\Gamma} &= \Lambda(\top \comp {!}) \comp \kappa_{\Gamma} \\
			&= \Lambda(\pi_1 \comp ((\top \comp {!}) \times \identity{})) \comp \kappa_{\Gamma} \\
			&= \Lambda(\pi_1) \comp \top \comp {!} \comp \kappa_{\Gamma} \\
			&= \kappa_{\bullet} \comp \interpret{\mathbf{true}}'
		\end{align}
		\item $\Gamma \vdash \psi_1 \land \psi_2 : \bullet$ (and $\Gamma \vdash \psi_1 \lor \psi_2 : \bullet$)
		\begin{align}
			&\mathcal{A}^{\emalgsymbol}\interpret{\lambda X : 1. \psi_1^{\sharp}\ X \land \psi_2^{\sharp}\ X} \comp \kappa_{\Gamma} \\
			&= \Lambda (\land \comp \tupling{\eval \comp \tupling{\mathcal{A}^{\emalgsymbol}\interpret{\psi_1^{\sharp}} \comp \pi_1}{\pi_2}}{\eval \comp \tupling{\mathcal{A}^{\emalgsymbol}\interpret{\psi_2^{\sharp}} \comp \pi_1}{\pi_2}})\comp \kappa_{\Gamma} \\
			&= \Lambda (\land \comp \tupling{\eval \comp \tupling{\mathcal{A}^{\emalgsymbol}\interpret{\psi_1^{\sharp}} \comp \pi_1}{\pi_2}}{\eval \comp \tupling{\mathcal{A}^{\emalgsymbol}\interpret{\psi_2^{\sharp}} \comp \pi_1}{\pi_2}}\comp (\kappa_{\Gamma} \times \identity{})) \\
			&= \Lambda (\land \comp \tupling{\eval \comp \tupling{\mathcal{A}^{\emalgsymbol}\interpret{\psi_1^{\sharp}} \comp \kappa_{\Gamma} \comp \pi_1}{\pi_2}}{\eval \comp \tupling{\mathcal{A}^{\emalgsymbol}\interpret{\psi_2^{\sharp}} \comp \kappa_{\Gamma} \comp \pi_1}{\pi_2}}) \\
			&= \Lambda (\land \comp \tupling{\eval \comp \tupling{\kappa_{\bullet} \comp \interpret{\psi_1}' \comp \pi_1}{\pi_2}}{\eval \comp \tupling{\kappa_{\bullet} \comp \interpret{\psi_1}' \comp \pi_1}{\pi_2}}) \\
			&= \Lambda (\land \comp \tupling{\interpret{\psi_1}' \comp \pi_1}{\interpret{\psi_2}' \comp \pi_1}) \\
			&= \Lambda (\land \comp \tupling{\interpret{\psi_1}'}{\interpret{\psi_2}'} \comp \pi_1) \\
			&= \Lambda (\pi_1) \comp \land \comp \tupling{\interpret{\psi_1}'}{\interpret{\psi_2}'} \\
			&= \kappa_{\bullet} \comp \interpret{\psi_1 \land \psi_2}'
		\end{align}
		\item $\Gamma, Y : \eta' \vdash X : \eta$
		\begin{align}
			&\mathcal{A}^{\emalgsymbol}\interpret{\Gamma^{\sharp}, Y : \eta'^{\sharp} \vdash X : \eta^{\sharp}} \comp \kappa_{\Gamma, Y : \eta'} \\
			&= \mathcal{A}^{\emalgsymbol}\interpret{\Gamma^{\sharp}\vdash X : \eta^{\sharp}} \comp \pi_1 \comp (\kappa_{\Gamma} \times \kappa_{\eta'}) \\
			&= \mathcal{A}^{\emalgsymbol}\interpret{\Gamma^{\sharp}\vdash X : \eta^{\sharp}} \comp \kappa_{\Gamma} \comp \pi_1 \\
			&= \kappa_{\eta} \comp \interpret{\Gamma \vdash X : \eta}' \comp \pi_1 \\
			&= \kappa_{\eta} \comp \interpret{\Gamma, Y : \eta' \vdash X : \eta}'
		\end{align}
		\item $\Gamma, X : \eta \vdash X : \eta$
		\begin{align}
			&\mathcal{A}^{\emalgsymbol}\interpret{\Gamma^{\sharp}, X : \eta^{\sharp} \vdash X : \eta^{\sharp}} \comp \kappa_{\Gamma, X : \eta} \\
			&= \pi_2 \comp (\kappa_{\Gamma} \times \kappa_{\eta}) \\
			&= \kappa_{\eta} \comp \pi_2 \\
			&= \kappa_{\eta} \comp \interpret{\Gamma, X : \eta \vdash X : \eta}'
		\end{align}
		\item $\Gamma \vdash p(\psi_1, \dots, \psi_k) : \bullet$
		\begin{align}
			&\mathcal{A}^{\emalgsymbol}\interpret{\lambda X : 1. p(\psi_1^{\sharp}, \dots, \psi_k^{\sharp})} \comp \kappa_{\Gamma} \\
			&= \Lambda(\interpret{p}' \comp \langle \mathcal{A}^{\emalgsymbol}\interpret{\psi_1^{\sharp}} \comp \pi_1, \dots, \mathcal{A}^{\emalgsymbol}\interpret{\psi_k^{\sharp}} \comp \pi_1 \rangle) \comp \kappa_{\Gamma} \\
			&= \Lambda(\interpret{p}' \comp \langle \mathcal{A}^{\emalgsymbol}\interpret{\psi_1^{\sharp}}, \dots, \mathcal{A}^{\emalgsymbol}\interpret{\psi_k^{\sharp}} \rangle \comp \pi_1 \comp (\kappa_{\Gamma} \times \identity{})) \\
			&= \Lambda(\interpret{p}' \comp \langle \mathcal{A}^{\emalgsymbol}\interpret{\psi_1^{\sharp}} \comp \kappa_{\Gamma}, \dots, \mathcal{A}^{\emalgsymbol}\interpret{\psi_k^{\sharp}} \comp \kappa_{\Gamma} \rangle \comp \pi_1) \\
			&= \Lambda(\pi_1) \comp \interpret{p}' \comp \langle \interpret{\psi_1}', \dots, \interpret{\psi_k}' \rangle \\
			&= \kappa_{\bullet} \comp \interpret{p(\psi_1, \dots, \psi_k)}'
		\end{align}
		\item $\Gamma \vdash n : \mathbf{int}$
		\begin{align}
			&\mathcal{A}^{\emalgsymbol}\interpret{n} \comp \kappa_{\Gamma} \\
			&= n \comp {!} \comp \kappa_{\Gamma} \\
			&= \interpret{n}'
		\end{align}
		\item $\Gamma \vdash \psi_1 \mathrel{\mathrm{op}} \psi_2 : \mathbf{int}$
		\begin{align}
			&\mathcal{A}^{\emalgsymbol}\interpret{\psi_1^{\sharp} \mathrel{\mathrm{op}} \psi_2^{\sharp}} \comp \kappa_{\Gamma} \\
			&= \interpret{\mathrm{op}} \comp \tupling{\mathcal{A}^{\emalgsymbol}\interpret{\psi_1^{\sharp}}}{\mathcal{A}^{\emalgsymbol}\interpret{\psi_2^{\sharp}}} \comp \kappa_{\Gamma} \\
			&= \interpret{\mathrm{op}} \comp \tupling{\interpret{\psi_1}'}{\interpret{\psi_1}'} \\
			&= \interpret{\psi_1 \mathrel{\mathrm{op}} \psi_2}'
		\end{align}
	\end{itemize}
\end{proof}

\begin{theorem}\label{thm:nuHFL-as-instance}
	For any well-typed $\nu\mathrm{HFL}(\mathbb{Z})$-term $\vdash \psi : \bullet$, $\interpret{\psi}(\emptyset) = \mathcal{A}^{\emalgsymbol} \interpret{\psi^{\sharp}\ ()}(\emptyset)$.
\end{theorem}
\begin{proof}
	By Proposition~\ref{prop:nuHFL-cpo-poset},\ref{prop:nuHFL-translation-sound}, we have the following equation.
	\[ \interpret{\psi} = \interpret{\psi}' = \kappa_{\bullet}^{-1} \comp \mathcal{A}^{\emalgsymbol} \interpret{\psi^{\sharp}} = \eval \comp \tupling{\identity{}}{{!}} \comp \mathcal{A}^{\emalgsymbol} \interpret{\psi^{\sharp}} = \mathcal{A}^{\emalgsymbol} \interpret{\psi^{\sharp}\ ()} \]
\end{proof}
	
}{}

\end{document}